%% file: arXiv_MMS.tex
\documentclass[letterpaper,11pt]{article}
\usepackage{fullpage}
\usepackage{times}
\usepackage{soul}
\usepackage{url}
\usepackage[hidelinks]{hyperref}
\usepackage[utf8]{inputenc}
\usepackage{graphicx}
\usepackage{amsmath}
\usepackage{amsthm}
\usepackage{booktabs}
\usepackage{natbib}

\usepackage{latexsym}

\usepackage[switch]{lineno} 
\usepackage[ruled, vlined, linesnumbered]{algorithm2e}
\usepackage[shortlabels]{enumitem}

\usepackage{amssymb,amsthm,amsmath,mathtools}
\usepackage{thmtools,thm-restate}

\usepackage[table]{xcolor}
\usepackage{hyperref}
\hypersetup{
	colorlinks=false,
	linkcolor=red,
	urlcolor=green,
	citecolor=blue
}

\usepackage[capitalise,noabbrev]{cleveref}

\usepackage{pifont}

\usepackage{subfiles}
\usepackage{subcaption}

\usepackage{tabularx}
\newcolumntype{Y}{>{\centering\arraybackslash}X}
\usepackage{boldline}
\usepackage{arydshln} 
\usepackage{tikz}
\usetikzlibrary{fadings}
\usepackage[textsize=footnotesize,color=red!25,bordercolor=red]{todonotes}

\usepackage{nicefrac}

\graphicspath{ {./images/} }

\DeclarePairedDelimiter\ceil{\lceil}{\rceil}
\DeclarePairedDelimiter\floor{\lfloor}{\rfloor}

\DeclareMathOperator*{\argmin}{arg\,min}

\usepackage{thm-restate}

\newtheorem{theorem}{Theorem}
\newtheorem{lemma}{Lemma}
\newtheorem{proposition}{Proposition}
\newtheorem{corollary}{Corollary}
\theoremstyle{definition}
\newtheorem{definition}{Definition}

\newtheorem{remark}{Remark}

\newcommand{\MMS}[2]{(#1, #2)\text{-MMS}}

\newcommand{\ins}[1]{\langle #1 \rangle}

\newcommand{\Targ}{\floor{\frac{2n}{3}}}

\usepackage{xspace}

\newcommand{\ie}{\emph{i.e.}\xspace}

\makeatletter
\newcommand\nocaption{
    \renewcommand\p@subfigure{}
    \renewcommand{\thesubfigure}{\arabic{figure}\alph{subfigure}}
}
\makeatother

\newcommand{\repeatcaption}[2]{%
  \renewcommand{\thefigure}{\ref{#1}}%
  \captionsetup{list=no}%
  \caption{#2}%
}

\newcommand{\repeatsubcaption}[2]{%
  \renewcommand{\thesubfigure}{\ref{#1}}%
  \captionsetup{list=no}%
  \caption{#2}%
}

\title{Guaranteeing Maximin Shares: Some Agents Left Behind}

\author{
Hadi Hosseini\\
Pennsylvania State University\\
University Park, PA, USA\\
\texttt{hadi@psu.edu}
\and
Andrew Searns\\
Rochester Institute of Technology\\
Rochester, NY, USA\\
\texttt{abs2157@rit.edu}
}

\date{}

\begin{document}

\maketitle

\begin{abstract}
The maximin share (MMS) guarantee is a desirable fairness notion for allocating indivisible goods. While MMS allocations do not always exist, several approximation techniques have been developed to ensure that all agents receive a fraction of their maximin share. 
We focus on an alternative approximation notion, based on the population of agents, that seeks to guarantee MMS for a fraction of agents. We show that no optimal approximation algorithm can satisfy more than a constant number of agents, and discuss the existence and computation of MMS for all but one agent and its relation to approximate MMS guarantees. We then prove the existence of allocations that guarantee MMS for $\frac{2}{3}$ of agents, and devise a polynomial time algorithm that achieves this bound for up to nine agents. 
A key implication of our result is the existence of allocations that guarantee $\text{MMS}^{\ceil{3n/2}}$, \ie, the value that agents receive by partitioning the goods into $\ceil{\frac{3}{2}n}$ bundles, improving 
the best known guarantee of $\text{MMS}^{2n-2}$. Finally, we provide empirical experiments using synthetic data.
\end{abstract}

\section{Introduction}
Fair division deals with the allocation of a set of resources to a set of agents in a fair manner \citep{brams1996fair,foley1967resource,hosseini2020fair}. One of its most notable application domains deals with the allocation of \textit{indivisible} (and non-shareable) goods. These applications arise, for example, in dividing inheritance among heirs and dispute resolution \citep{brams1996fair}, task assignment or course allocation \citep{budish2011combinatorial}, and have been  popularized in recent years due to publicly accessible platforms such as \textit{Spliddit} \citep{goldman2014spliddit}.

The most desirable fairness notion, envy-freeness, requires that each agent weakly prefers his own allocation to that of all other agents. A weaker notion, called proportionality, requires that each agent receives $\frac{1}{n}$ of her valuation of all goods, where $n$ is the total number of agents. Unfortunately, in the indivisible domain neither of these notions are guaranteed to exist: consider a single good and two interested agents. Neither envy-freeness nor proportionality can be satisfied as one agent is destined to remain empty handed.

A third fairness notion, proposed by  \cite{budish2011combinatorial}, is the \textit{maximin share} (MMS) guarantee, which can be seen as a generalization of the \textit{cut-and-choose} protocol \citep{brams1996fair,brams1995envy}.
In a nutshell, the maximin share is the value that an agent can guarantee by dividing the goods into $n$ bundles, assuming that all other agents choose a bundle before she does. 
There is no reason to believe that such a bound can always be satisfied for \textit{all} agents. It turns out that an MMS allocation does not always exist \citep{kurokawa2018fair}, and even when it exists, computing an MMS partition is intractable \citep{Bouveret2016}. 
Thus, several approximation techniques were developed to guarantee that each agent receives a fraction of her MMS. 
While these techniques are compelling, they may leave a large fraction of agents without their MMS guarantee.

We propose to circumvent this obstacle by taking an orthogonal direction based on the degree of fairness in the society of agents. Our goal is to find allocations that give a large fraction of agents their maximin share guarantee, possibly leaving a small fraction of agents behind (since MMS allocations do not always exist).
This approach is related to ordinal approximations of MMS and is of interest in several application domains: resources in a hospital must be distributed among tasks/procedures in a manner that a subset of critical tasks can be fully accomplished; senior college students (as a fraction of all students) must be guaranteed seats in courses to avoid graduation delays.

Unfortunately, envy-freeness and proportionality cannot be approximated in this dimension.  
For instance, one may be interested in minimizing the fraction of envious agents. Envy-freeness (similarly proportionality) may leave almost all but one agent unsatisfied: consider $n$ goods and $n$ agents with identical valuations, having $1- (n-1)\epsilon$ value for one good and $\epsilon$ for the rest. Any distribution of goods will leave $n-1$ agents envious.
Therefore, focusing on the maximin share guarantee we ask the following questions:

\begin{quote}\textit{
What fraction of agents can be guaranteed their MMS value?
And can we compute allocations that guarantee MMS for the majority of the agents?
}
\end{quote}

\paragraph{Our Contributions.}

We propose a novel fairness framework, called $\MMS{\alpha}{\beta}$, wherein $\alpha$ fraction of agents receive $\beta$ approximations of their MMS value and investigate the interplay between the two approximation parameters. We establish the connection between $\MMS{\alpha}{\beta}$ and ordinal approximations of $\text{MMS}$ (\Cref{prop:MMSAlpha}), and investigate the computational boundaries of $\alpha$ and $\beta$ as follows:

\begin{itemize}[leftmargin=*]
    \item \textbf{Optimal MMS}: We prove the existence of a family of instances where almost all agents do not receive their MMS by any optimal-MMS allocation which aims to maximize $\beta$ for all agents. We give an intricate counterexample that uses a \textit{quadratic} number of goods, where $n-3$ agents ($n-4$ if $n$ is odd) do not receive their MMS value in any optimal-MMS allocation (\Cref{thm:optimal_mms_failure}).
    
    \item \textbf{Computing $\MMS{\frac{n-1}{n}}{\beta}$}: We devise a polynomial-time algorithm achieving $\MMS{\alpha}{\beta}$ when $\alpha =  \frac{n-1}{n}$ (\Cref{thm:GeneralN}). 
    A consequence of this result is a tight approximation for $n \leq 4$ that immediately implies an algorithm for computing $\text{MMS}^{n+1}$ for $n < 4$ (\Cref{cor:MMS4}).
    
    \item \textbf{Existence of $\MMS{\frac{2}{3}}{1}$}: We prove the existence of $\MMS{\frac{2}{3}}{1}$ for any number of agents (\cref{thm:two_thirds_existence}), and provide an algorithm that achieves this bound in polynomial time for $n<9$ (\cref{thm:two_thirds_poly}).
    A key implication of our result is the existence of allocations that guarantee $\text{MMS}^{\lceil3n/2\rceil}$, \ie, the value that an agent receives by partitioning the goods into $\lceil\frac{3}{2}n\rceil$ bundles (\cref{cor:MMS3n2}). Our result significantly improves the best known guarantee of $\text{MMS}^{2n-2}$ \citep{aigner2019envy}.
\end{itemize}

On the experimental front, we show that a simplified, polynomial-time variant of our algorithm satisfies a large fraction of agents on the most `difficult' instances and this fraction grows as the ratio of goods to agents increases.

\paragraph{Related Work.}

On a high level, our approach is related to notions defined to measure the degree of fairness in a society of agents whether it pertains to minimizing the maximum or sum of envy \citep{chevaleyre2007reaching,chen2017ignorance,nguyen2014minimizing}, minimizing envy ratio \citep{lipton2004approximately}, balancing the amount of envy experienced in a society \citep{tadenuma1995refinements}, or promoting fairness through equitable allocations where agents receive the same level of utility \citep{schneckenburger2017atkinson,freeman2019equitable}.
Another closely related line of work suggests the notion of counting instances of envy (pairs of agents)---among several other plausible measures---as opposed to those that measure the intensity of envy \citep{feldman1974fairness}.
Similar ideas were briefly mentioned by \citep{chevaleyre2017distributed} and \citep{netzer2016distributed}; all regarding the degree of envy in a society.

On a technical level, the non-existence of MMS allocations \citep{procaccia2014fair} and its intractability \citep{Bouveret2016,woeginger1997polynomial} has given rise to a number of approximation techniques. These algorithms guarantee that each agent receives an approximation of their maximin share. Recently, \cite{nguyen2017approximate} gave a Polynomial Time Approximation Scheme (PTAS) for a notion defined as \textit{optimal-MMS}, that is, the largest value, $\beta$, for which each agent $i$ receives the value of $\beta\text{MMS}_{i}$. Since the number of possible partitions is finite, an optimal-MMS allocation always exists, and it is an MMS allocation if $\beta \geq 1$. The current best results guarantee $\beta \geq 2/3$ \citep{kurokawa2018fair,garg2018approximating} and $\beta \geq 3/4$ \citep{garg2019improved,ghodsi2018fair} in general, and $\beta \geq 7/8$ \citep{amanatidis2017approximation} and $\beta \geq 8/9$ \citep{gourves2019maximin} in the case of three agents .

\section{Preliminaries} \label{sec:preliminaries}
Let $N = \{1,\ldots, n\}$ be a set of agents and $M$ denote a set of $m$ indivisible goods, where $m > n$. We denote the value of agent $i\in N$ for good $g\in M$ by $v_{i}(g) \geq 0$. We assume that the valuation functions are \textit{additive}; that is, for each subset $G\subseteq M$, $v_{i}(G) = \sum_{g\in G} v_{i}(g)$.
An \textbf{instance} of the problem is $I = \ins{N, M, V}$ where $V$ is the valuation profile of agents.
An allocation $A = (A_1, \ldots, A_n)$ is an $n$-partition of $M$ that allocates the bundle of goods in $A_{i}$ to each agent $i\in N$.

\begin{definition}[\textbf{Envy-freeness}]
	An allocation $A$ is \emph{envy-free} (EF) if for every pair of agents $i,j \in N$, $v_i(A_i) \geq v_i(A_j)$. An allocation $A$ is \emph{envy-free up to one good} (EF1) if for every pair of agents $i,j \in N$ such that $A_j \neq \emptyset$, there exists some good $g \in A_j$ such that $v_i(A_i) \geq v_i(A_j \setminus \{g\})$. 
    An allocation $A$ is \emph{envy-free up to any good} (EFX) if for every pair $i,j \in N$ such that $A_j \neq \emptyset$, for any good $\forall g \in A_j$, $v_i(A_i) \geq v_i(A_j \setminus \{g\})$.
    These definitions are due to \citet{foley1967resource}, \citet{budish2011combinatorial}, and \citet{caragiannis2016unreasonable} respectively.
\end{definition}

\begin{definition}[\textbf{Maximin Share Guarantee}]
Let $\Pi_{k}(M)$ denote the set of $k$-partitions of $M$. 
The \textbf{$k$-maximin share guarantee} of agent $i\in N$ on $\Pi_{k}(M)$ is 
$$\text{MMS}_{i}^{k}(M) = \max_{(A_{1}, A_{2}, \ldots A_{k})\in \Pi_{k}(M)}\min_{j \in [k]} v_{i}(A_{j}),$$
where $[k] = \{1,\ldots, k\}$. Intuitively, this is the minimum value that can be guaranteed if agent $i$ partitions the goods into $k$ bundles and chooses the least valued bundle.  

An allocation $A = (A_1, \ldots, A_k) \in \Pi_{k}(M)$ is an \textit{MMS$^{k}$ allocation} if and only if $\forall i\in N, v_{i}(A_{i}) \geq \text{MMS}_{i}^{k}$. 
Note that  $\text{MMS}_{i}^{n}(M) \leq \frac{v_{i}(M)}{n}$ since  proportionality implies MMS.
When it is clear from the context, we write $\text{MMS}_{i}$ instead of $\text{MMS}_{i}^{n}(M)$ and simply refer to it as agent $i$'s \textit{MMS value}.
\end{definition}

\begin{definition}[\textbf{Optimal-MMS}]
While an MMS allocation does not always exist, an optimal relaxation of MMS guarantees that agents receive a fraction of their MMS value~\citep{nguyen2017approximate}.
Given an instance $I = \ins{N, M, V}$, the optimal-MMS value is defined by %
$$\lambda^{*}(I) = \max_{(A_{1}, \ldots, A_{n})\in \Pi_{n}(M)}\min_{i}\frac{v_{i}(A_{i})}{\text{MMS}_{i}^{n}}.$$ 
By definition, an allocation that gives each agent a $\lambda^{*}(I)$ fraction of its $\text{MMS}_{i}^{n}(M)$ value is guaranteed to exist.
\end{definition}

We now provide a set of relevant lemmas and observations that will be used throughout the paper.

\paragraph{Ordered instance.}
An instance is \textit{ordered} when all agents agree on the linear ordering of the goods, irrespective of their valuations. 
Formally, $I$ is an \emph{ordered instance} if there exists an ordering of goods, $(g_{1}, g_{2}, \ldots, g_{m})$ such that for all agents $i\in N$ we have 
$v_{i}(g_{1}) \geq v_{i}(g_{2}) \geq \ldots \geq v_{i}(g_{m})$.
\citet{Bouveret2016} showed that ordered instances are the `hardest' for achieving MMS.
In fact, these instances are the only known structures for which MMS does not exist~\citep{kurokawa2018fair}. 
The next lemma states that given an \textit{un}ordered instance, it is always possible to generate a corresponding ordered instance in polynomial time. Furthermore, if the ordered instance admits an MMS allocation, the original instance also admits an MMS allocation which can be computed in polynomial-time.

\begin{lemma}[\citet{barman2017approximation}]\label{lem:order}
Let $I' = \ins{N, M, V'}$ be an ordered instance constructed from the original instance $I = \ins{N, M, V}$. 
Given allocation $A'$ on $I'$, a corresponding allocation $A$ on $I$ can be computed in polynomial time such that for all $i\in N, v_{i}(A_{i}) \geq v'_{i}(A'_{i})$.
\end{lemma}

\paragraph{Scale invariance.}
The scale invariance property of MMS states that if an agent's valuations are scaled by a factor, then its MMS value scales by the same factor.

\begin{lemma}[\citet{ghodsi2018fair}]\label{lem:scale_invariance}
Let $I = \ins{N, M, V}$ be an instance and $c > 0$ be a real scalar. Let $I' = \ins{N, M, V'}$ be constructed so that $v'_{i}(g) = cv_{i}(g)$ for all $g \in M$. Then $\text{MMS}_{i}^{'k}(M) = c\text{MMS}_{i}^{k}(M)$.
\end{lemma}

Thus, an instance $I = \ins{N, M, V}$ and a real value $k$ can be scaled to form a new instance $I' = \ins{N, M, V'}$, such that for each agent $i\in N$, $v'_{i}(g) = \frac{k}{v_{i}(M)} v_{i}(g)$, and $v'_{i}(M) = k$.

\paragraph{Valid reduction.} We call the act of removing a set $A_i \subseteq M$ of goods and an agent $i$ a valid reduction if the following two conditions hold:
i) $v_i(A_i) \geq \text{MMS}_{i}^{n}(M)$ and
ii) $\forall j \in N \setminus \{i\}, \text{MMS}_{j}^{n-1}(M \setminus A_i) \geq \text{MMS}_{j}^{n}(M)$.

\begin{lemma}[\citet{garg2018approximating}]\label{lem:n_goods_reduction}
Given an ordered instance $I = \ins{N, M, V}$ with $|N|=n$ such that $\text{MMS}_{i}^{n} \leq 1$, if $v_{i}(\{g_{n},g_{n+1}\}) \geq 1$, then removing $A_{i} = \{g_{n}, g_{n+1}\}$ and agent $i$ forms a valid reduction. Similarly, the removal of $\{g_{1}\}$ and agent $i$ forms a valid reduction if $v_{i}(\{g_{1}\}) \geq 1$.
\end{lemma}

\paragraph{Normalized instance.} An instance is normalized if all of the following properties hold:
i) the instance is ordered,
ii) it is scaled so that $v_{i}(M) = n$, and
iii) it is reduced so that $v_{i}(g_{1}) < 1$ and $v_{i}(\{g_{n}, g_{n+1}\}) < 1$.
By combining Lemma~\ref{lem:order}, \ref{lem:scale_invariance}, and \ref{lem:n_goods_reduction}, we may assume that instances are normalized. We prove this claim in Lemma~\ref{lem:normalize}.

\section{Approximating Maximin Share}\label{sec:alpha-beta}
We introduce a fairness concept that allows for interpolation between two dimensions in approximating MMS pertaining to the fraction of agents $\alpha$ that receive a $\beta$ approximation of their maximin share.

\begin{definition}[\textbf{$\MMS{\alpha}{\beta}$}] 
An allocation $A$ guarantees $\MMS{\alpha}{\beta}$ if $\alpha \in (0,1]$ fraction of agents receive at least their $\beta \in (0,1]$ approximation of their $\text{MMS}_{i}^{n}$.
Formally, given an instance $I = \ins{N, M, V}$, an allocation $A$ guarantees $\MMS{\alpha}{\beta}$ if there exists a subset $N'\subseteq N$ with $|N'| \geq \floor{\alpha |N|}$ such that for all $i\in N', v_{i}(A_{i}) \geq \beta \text{MMS}_{i}^{n}$. 
We say that $\MMS{\alpha}{\beta}$ exists if for any instance $I = \ins{N, M, V}$, for \textit{every} subset $N' \subseteq N$ of agents such that $|N'| = \floor{\alpha |N|}$, there exists an allocation $A$ such that for all $i \in N'$, $v_{i}(A_{i}) \geq \beta \text{MMS}_{i}^{n}$. 
\end{definition}
Previous MMS approximation results can be seen in this context as efforts to tighten the approximation bound for all agents ($\alpha = 1$). Notably, greedy algorithms exist to compute $\MMS{1}{\frac{2}{3}}$ \citep{barman2017approximation,garg2018approximating} and $\MMS{1}{\frac{3}{4}}$ \citep{ghodsi2018fair} allocations.

\begin{remark} It is crucial to highlight two key distinctions: First, in contrast to  previous works \citep{ortega2018social,segal2019democratic,NYMAN2020115}, the definition of $\MMS{\alpha}{\beta}$ existence does not only hold for a fixed subset of agents. Rather, it is a stronger concept that holds for every subset of $\floor{\alpha n}$ agents.
Second, $\MMS{\alpha}{\beta}$ enables a social planner to pre-select the $N'\subset N$ of $\floor{\alpha n}$ agents--independent of their preferences--according to some priority ordering over the agents or by selecting the agents uniformly at random. 
For instance, a higher priority may be given to senior college students; a practice that is already common in most course allocation procedures.
\end{remark}

\subsection{$\MMS{\alpha}{1}$ Implies  $\text{MMS}^{k}$ for $k\geq \lceil \frac{n}{\alpha} \rceil$}

\citet{budish2011combinatorial} showed that approximating the competitive equilibrium from equal incomes (A-CEEI) guarantees $\text{MMS}^{n+1}(M)$, \ie adding a dummy agent and asking all agents to partition the goods into $n+1$ bundles. However, this result does not imply the existence of $\text{MMS}^{n+1}$ in allocating indivisible goods because allocations achieved by A-CEEI may have excess supply or excess demands. This approach can result in infeasible allocations in fair division settings that do not allow the addition of excess goods. Therefore, the existence of $\text{MMS}^{k}$ for $n+1 \leq k \leq 2n-2$ remains an open problem.
In \cref{prop:MMSAlpha} we show the relation between $\text{MMS}^{k}$ with $\MMS{\alpha}{\beta}$.

\begin{proposition}\label{prop:MMSAlpha}
The existence of $\MMS{\alpha}{1}$ implies the existence of $\text{MMS}^{k}$ for $k\geq \lceil \frac{n}{\alpha} \rceil$. 
\end{proposition}
\begin{proof}
Suppose that $\MMS{\alpha}{1}$ exists. Given an instance $I = \ins{N, M, V}$ with $n$ agents, we construct an instance $I'$ from $I$ by adding $\ceil{\frac{n}{\alpha}} - n$ dummy agents. Therefore, $|N'| = \ceil{\frac{n}{\alpha}}$. Since $\MMS{\alpha}{1}$ exists, for every subset $N'' \subseteq N'$ of size $\floor{\alpha |N'|}$, there exists an allocation which satisfies $\MMS{\alpha}{1}$ on that set of agents. Thus, we may choose the $\floor{\alpha |N'|}$ agents to contain exactly the well-defined original set of agents in $N$, that is, $N'' \coloneqq N$. Hence, each agent $i \in N$ receives at least $v_{i}(A_{i}) \geq \text{MMS}^{|N'|}$, which implies $\text{MMS}^{\ceil{\frac{n}{\alpha}}}$ for agents in $N$.
\end{proof}

\subsection{Failure of Optimal MMS Algorithms} \label{sec:optimalMMS}
There are two primary motivations behind the approximation parameter $\alpha$ (the fraction of agents). First, \cref{prop:MMSAlpha} states that fixing $\beta = 1$ immediately implies an ordinal approximation of MMS by partitioning goods into $k > n$ bundles. Second, our next theorem shows that maximizing the value of $\beta$ for all agents may result in only a small fraction of agents ($\alpha$) receiving their MMS value. 
\cref{thm:optimal_mms_failure} shows that there exists a family of instances where any optimal-MMS allocation can only give a small constant number of agents their MMS value. Therefore, an optimal-MMS algorithm (e.g. \citet{nguyen2017approximate}) will result in an $\MMS{\alpha}{\lambda^{*}}$ allocation such that $\alpha$ goes to zero as the number of agents, $n$, increases.

\begin{theorem}\label{thm:optimal_mms_failure}
For every $n \geq 4$, there exists an instance with $O(n^{2})$ goods where every optimal-MMS allocation guarantees at most $3$ ($4$ if $n$ is odd) agents their MMS value.
\end{theorem}

\begin{figure*}[t]
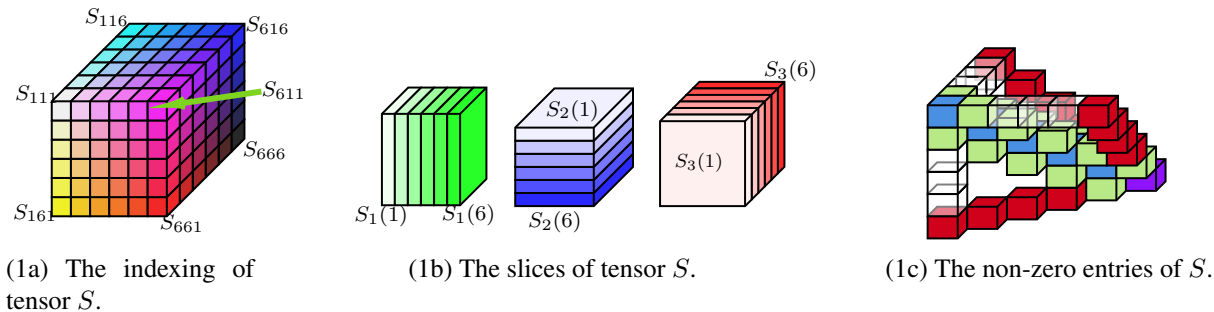

\nocaption
\centering \scriptsize
    \begin{subfigure}[t]{0.2\textwidth}
        \centering 
        \subfile{images/S_Shape_fig}
        \caption{The indexing of tensor $S$.}
        \label{fig:tensor_shape}
    \end{subfigure} \hfill
    \begin{subfigure}[t]{0.33\textwidth}
        \centering 
        \subfile{images/Slices_fig}
        \caption{The slices of tensor $S$.}
        \label{fig:slices}
    \end{subfigure} \hfill
    \begin{subfigure}[t]{0.32\textwidth}
    \centering
    \subfile{images/T_Shape_fig}
    \caption{The non-zero entries of $S$.}
    \label{fig:T_shape}
    \end{subfigure}
   \caption{\small The shape and slices of order $3$ tensors with dimensions $(6 \times 6 \times 6)$.
   \label{fig:tensors}}
\end{figure*}

\begin{proof}[Proof sketch.]
The proof is by a construction inspired by those proposed by \citet{kurokawa2016can,kurokawa2018fair}, but it includes a few intricate modifications.\footnote{The details of the construction, the necessary lemmas, and proofs in this section are relegated to the appendix.}
We illustrate our construction by separating agents into two groups who value items according to perturbations of a common valuation matrix. Any optimal-MMS allocation must allocate items according to the rows or columns the $n \times n$ valuation matrix. Half of the agents have $\text{MMS}_{i} = 1$ by selecting the rows of the matrix; these agents have value $1-\epsilon$ for each column of the matrix except the last. Similarly, the other half of the agents form an MMS partition by selecting the columns of the matrix. Any optimal-MMS allocation must give each agent at least $1-\epsilon$ value but this can only be achieved if the allocation corresponds to either the rows or the columns of the matrix. Thus half of the agents (except one lucky agent who receives the last row/column) are destined to receive only $1-\epsilon < \text{MMS}_{i}$ value.

We extend this counterexample to a constant number of satisfied agents by dividing agents into $\floor{\frac{n}{2}}$ groups. We guarantee that only one group can be fully satisfied at a time using an intricate generalization of the valuation matrix to higher dimensions. The tensors in our construction are partitioned according to $\floor{\frac{n}{2}}$ sets of sheets (See Fig.~\ref{fig:slices}). 
Each group of two agents selects a different set of slices to compute $\text{MMS}_{i} = 1$. For all other slices (except the last slice in each set: $S_{j}(n)$), these agents will have $1 - \epsilon$ value.
Fig.~\ref{fig:T_shape} depicts the non-zero entries of the valuation tensor. The total number of non-zero goods is proportional to the number of agents multiplied by the number of groups, and thus, is quadratic in the number of agents.
\end{proof}


\Cref{thm:optimal_mms_failure} illustrates that if the goal is to reach an optimal-MMS threshold, the fraction of agents, $\alpha$, who receive their MMS guarantee, $\beta = 1$, goes to zero as the number of agents increases. Thus, we ask for what values of $\alpha$, $\MMS{\alpha}{1}$ is guaranteed to exist?

\section{Computing $\MMS{\alpha}{\beta}$ for $n-1$ Agents} \label{sec:approx}
Our main goal in this section is devising solutions for computing $\MMS{\alpha}{\beta}$ by removing a single agent.

\subsection{A $\MMS{\frac{2}{3}}{1}$ Algorithm for Three Agents}
It is worth noting that although MMS always exists for $n=2$ and can be achieved through the cut-and-choose protocol \citep{Bouveret2016}, computing such an allocation remains hard~\citep{amanatidis2017approximation}.
Nonetheless, for two agents an allocation that guarantees $\text{MMS}^{3}$ to each agent can be computed in polynomial time through an EF1 allocation~\citep{aigner2019envy}. We use this result to obtain the following theorem.

\begin{proposition} \label{prop:computing_n=3}
For three agents, $\MMS{\frac{2}{3}}{1}$ always exists and can be computed in polynomial time.
\end{proposition}
\begin{proof}
By the definition of $\MMS{\alpha}{\beta}$, we can select any arbitrary subset of agents $N'\subset N$ such that $|N'| = 2$. Then, run an EF1 algorithm on $N'$, which outputs an allocation $A$. 
Both agents in $N'$ are guaranteed to receive their $\text{MMS}^{3}$, i.e., for each $i\in N'$ we have $v_{i}(A_{i}) \geq \text{MMS}^{3}_{i}$, implying that $\frac{2}{3}$ of agents receive their MMS.
\end{proof}

By the construction proposed by \citet{kurokawa2018fair}, an MMS allocation does not always exist for three agents; thus, $\MMS{\frac{2}{3}}{1}$ is a tight bound for $n = 3$.

\subsection{A General Algorithm for $n \geq 4$}
To extend the analysis of approximate MMS for $n-1$ agents, we first provide an important lemma that enables us to focus on normalized instances in the remainder of this paper.
This lemma states that we can employ valid reductions  repeatedly (see \Cref{sec:preliminaries}) for computing $\MMS{\alpha}{\beta}$ allocations.

\begin{restatable}{lemma}{lemnormalize}\label{lem:normalize}
Given an instance $I = \ins{N, M, V}$, we can compute a normalized instance $I' = \ins{N', M', V'}$ in polynomial time such that any $\MMS{\alpha}{\beta}$ allocation on $I'$ implies $\MMS{\alpha}{\beta}$ on $I$.
\end{restatable}

We now focus attention on designing a polynomial-time algorithm that guarantees $\beta$ MMS for $n-1$ agents. The algorithm relies on removing an arbitrary agent and computing an EFX allocation for the remaining $n-1$ agents. By \cref{lem:order} an ordered instance can be generated (and easily converted back) from an unordered instance. A simple variant to the envy-graph algorithm satisfies EFX on ordered instances \citep{barman2017approximation}.
We show that applying this procedure to normalized (and thus ordered) instances satisfies $\MMS{\frac{n-1}{n}}{\frac{1}{2}(\frac{n+2}{n-1})}$. 
Since $\beta$ depends on $n$, we cannot trivially extend this result to any (not normalized) instances. Nonetheless, we prove that together with \cref{lem:normalize} and \cref{lem:order} we can compute  $\MMS{\frac{n-1}{n}}{\frac{1}{2}(\frac{n+2}{n-1})}$ for any instance in polynomial-time. 
The full proof of the theorem, along with necessary discussions, is deferred to the appendix.

\begin{restatable}{theorem}{thmGeneralN}\label{thm:GeneralN}
Given any instance of $n$ agents, $\MMS{\frac{n-1}{n}}{\frac{1}{2}(\frac{n+2}{n-1})}$ can be computed in polynomial time.
\end{restatable}

By \cref{prop:MMSAlpha}, we can add a dummy agent when $n=3$ and select the original set of agents to obtain the following.

\begin{corollary} \label{cor:MMS4}
For $n = 3$ agents, computing an allocation satisfying $\text{MMS}^{4}_{i}, \forall i\in N$ can be done in polynomial time.\footnote{Independently, \citet{aigner2019envy} provided an algorithm based on envy-free matchings for computing $\text{MMS}^{n+1}$ for 3 agents.}
\end{corollary}

\begin{remark}
Theorem~\ref{thm:GeneralN} immediately illustrates an intriguing interpolation between the approximation ratios of $\alpha$ and $\beta$: for $n < 7$, a better approximation of MMS values ($\beta$) can be achieved in polynomial time by sacrificing only one agent. 
Recall that the best approximation algorithms to date guarantee only $\MMS{1}{\frac{3}{4}}$~\citep{ghodsi2018fair} for general additive valuations. 
Table~\ref{tab:approx} shows the interpolation between $\alpha$ and $\beta$ for $n=4$ to $n=7$.

\begin{table}[t]\small
\centering
\begin{tabular}{c|c}
    $n$ & $\MMS{\alpha}{\beta}$ \\\hline
    4   & $\MMS{3/4}{1}$ \\
    5   & $\MMS{4/5}{7/8}$ \\
    6   & $\MMS{5/6}{4/5}$ \\
    7   & $\MMS{6/7}{3/4}$
\end{tabular}
\caption{\small Approximation bounds of $(\alpha, \beta)$ for various $n < 8$.}
\label{tab:approx}
\end{table}
\end{remark}

\section{The Existence of $\MMS{\frac{2}{3}}{1}$ Allocations}\label{sec:twothirdMMS}
\cref{thm:GeneralN} optimizes the fraction of agents, $\alpha$, and provides an efficient approach in achieving approximate MMS for $n-1$ agents.
Another plausible, and often practical, approach aims at maximizing the fraction of agents $\alpha$ who receive their MMS guarantee ($\beta = 1$).

It turns out that simple modifications to existing approximation algorithms can guarantee $\MMS{\frac{1}{2}}{1}$ allocations (see the appendix for detailed algorithms and proofs).
A key question is whether we can improve this bound for $\alpha$ beyond $\frac{1}{2}$ and show the existence of such allocations. In what follows, we show that $\MMS{\frac{2}{3}}{1}$ exists and discuss an algorithm achieving this bound in polynomial-time when $n<9$.

Our existence proof of $\MMS{\frac{2}{3}}{1}$ relies on combining techniques of bag-filling, strong normalization, and a variant of the lone divider procedure. Before discussing our main result, we briefly describe these techniques.

\paragraph{Bag-filling.}
Given a normalized instance, a good is \textit{high-value} for agent $i$ if $v_{i}(g) \geq \frac{1}{2}$; otherwise it is \textit{low-value}. 
The bag-filling algorithm is a greedy approach for forming bundles in a normalized instance. 
An agent initializes a bag with a high-value good, and then adds low-value goods until the bag is worth at least $1$. She then picks another high-value good and repeats this process.\footnote{Similar approaches have been used by \citet{garg2018approximating,garg2019improved,ghodsi2018fair} to compute $\MMS{1}{\frac{2}{3}}$ and $\MMS{1}{\frac{3}{4}}$ allocations. In their algorithms, the bag is filled until \textit{any} agent values it at least $1$.} 
For each bag, the total value never exceeds $1 + \frac{1}{2} = \frac{3}{2}$ since the last added good is low-value.

\begin{remark}\label{rem:bag-filling-example}
Notice that bag-filling alone cannot guarantee $\MMS{\frac{2}{3}}{1}$ because it may `run out' of low-value goods before filling $\Targ$ bundles if the remaining value consists entirely of high-value goods. Consider an instance with $n = 9$ agents with identical valuations over goods as follows: $5$ goods of value $0.99$, $5$ goods of value $0.01$, $1$ good of value $0.95$, $1$ good of value $0.05$, $3$ goods of value $0.55$, and $3$ goods of value $0.45$.
Here, $\text{MMS}_{i}^{n} = 1$ for all agents. The high-value goods have value more than $0.5$. During bag-filling, there will be three bundles of $\{0.99, 0.45\}$, one bundle of $\{0.99, 0.05\}$, and one bundle of $\{0.99, 0.01\}$. Since only $5$ bundles were filled and $\Targ = 6$, we need to fill one more bundle. Adding all remaining low-value goods to the next high-value good $0.95$ yields a total value $0.99$. The remaining value is tied up with the high-value goods with value $0.55$. \end{remark}


\cref{lem:normalize} enabled us to focus on normalized instances with the assumption that $v_{i}(M) = n$ for all agents. Our technique in this section requires a slightly stronger assumption.

\paragraph{Strong Normalization.}
We say an instance is strongly normalized if it is ordered, $v_{i}(M) = n$, and $\text{MMS}_{i}^{n} = 1$.
Observe that the definition of strong normalization implies that each bundle of an MMS partition is valued \textit{exactly} $1$. Furthermore, this implies that $v_{i}(g_{1}) \leq 1$ and that $v_{i}(g_{n+1}) \leq \frac{1}{2}$.

\begin{restatable}{lemma}{lemstrongnormalization}\label{lem:strong_normalization}
For any additive instance $I = \ins{N, M, V}$, there exists another strongly normalized instance $I' = \ins{N, M, V'}$ such that $I'$ is ordered and for all $i \in N$, $\text{MMS}_{i}^{n} = 1$ and $v_{i}(M) = n$. Furthermore, an $\MMS{\alpha}{\beta}$ allocation on $I'$ is also an $\MMS{\alpha}{\beta}$ allocation on $I$.
\end{restatable}

\cref{lem:strong_normalization} shows that any instance $I$ can be modified via a non-polynomial-time transformation into a new instance $I'$ 
that is strongly normalized such that an $\MMS{\alpha}{\beta}$ allocation on $I'$ implies an $\MMS{\alpha}{\beta}$ allocation on $I$.
Hence, we can focus only on strongly normalized instances.

\paragraph{Algorithm description.}
At its core, our algorithm implements the \textit{lone divider procedure} for indivisible goods on $n'= \Targ$ of the agents. The lone divider procedure is an extension of the proportional cake-cutting algorithm \citep{robertson1998cake} that leverages non-empty envy-free matchings in a bipartite graph. 
It proceeds as follows: given $n'$ agents in $N'$, first, an arbitrary agent divides the goods into $n'$ acceptable bundles $(A_{1}, \ldots, A_{n'})$.
Second, we construct a bipartite graph between agents in $N'$ and the bundles $(A_{1}, \ldots, A_{n'})$ wherein each edge connects an agent $i\in N'$ to a bundle $A_k$ such that $v_{i}(A_{k}) \geq 1$.
Notice that the divider will be adjacent to all bundles.
Next, we compute a non-empty envy-free matching, where no unmatched agent is adjacent to a bundle that was assigned to some other agent.
A non-empty envy-free matching always exists if $|\mathcal{N}(N')| \geq |N'|$, where $\mathcal{N}(N')$ denotes the set of bundles adjacent to $N'$ \citep{aigner2019envy}.
All matched agents receive their bundles and leave the procedure. 
After the removal of matched agents and goods, we repeat the same procedure over the remaining goods and agents until there are no remaining agents.
The formal algorithm, its detailed description, and technical proofs are provided in the appendix.

\begin{figure*}
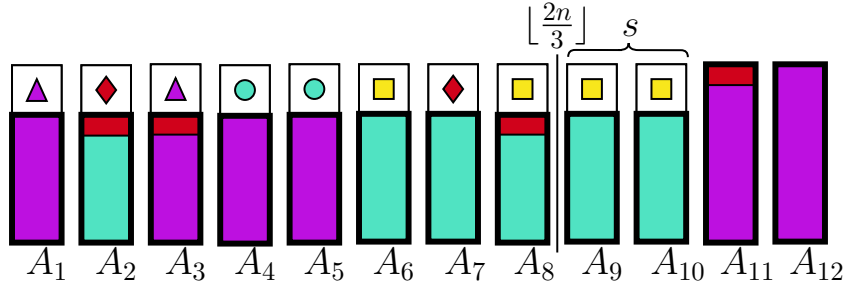

    \centering
    \subfile{images/lone_divider_partition}
    \caption{A sample MMS partition for the divider. The tiny shapes (triangle, circles, squares, and diamonds) represent high-value goods and are ordered from $g_1$ to $g_{10}$. Note that each bundle contains no more than one high-value good.
    All pieces indicated in red (including the red diamonds) indicate the goods that were allocated in previous iterations. There are two such high-value goods, so $n' = \Targ - 2 = 6$. The first step is pairing the lowest $2s$ remaining high-value goods into $s$ bundles: $\{g_{6}, g_{10}\}$ and $\{g_{8}, g_{9}\}$. 
    The second phase is restricted bag-filling that only uses low-value goods from the the remainder sets corresponding to already allocated high-value goods (within teal boxes). The solid borders show \textit{all} remainder sets. Next, the remaining goods, shown in purple are used for simple bag-filling to fill the remaining $2$ bundles.}
    \label{fig:lone_divider_example}
\end{figure*}

The primary challenge in proving the existence of $\MMS{\frac{2}{3}}{1}$ allocations lies in showing that in each iteration, given $n'$ remaining agents, the divider agent can form $n'$ bundles valued at least 1. 
We show that under the strong normalization assumption, the divider is always able to form $n'$ bundles either by pairing high-value goods or through \textit{restricted} bag-filling.

\begin{restatable}{theorem}{thmtwothirdsexistence}
\label{thm:two_thirds_existence}
A $\MMS{\frac{2}{3}}{1}$ allocation is guaranteed to exist for any subset of $\Targ$ agents.
\end{restatable}

Our proof relies on two technical cases based on the number of high-value goods ($h$) for the lone divider.
If $h \leq \Targ$, there are not many high-value goods, and thus, the simple bag-filling guarantees our desired result. The most challenging cases arise when $h > \Targ$, that is, when there are too many high-value goods. Let $s = h - \Targ$ be the number of excess high-value goods as shown in \cref{fig:lone_divider_example}.
Here, we show that the divider agent is able to form the necessary bundles (\ie the number of remaining agents) by adopting the following  strategies: 
if $n' \leq s$, we can simply pair the remaining high-value goods since there are at least $2n'$ high-value goods available, which implies $n'$ bundles are valued more than 1.
On the other hand, if $n' > s$, we are able to form $s$ bundles by pairing the remaining high-value goods.
This means that we still need an additional $n' - s$ bundles. The key to handling this step relies on the \textit{restricted bag-filling} procedure with exactly one high-value good per bundle but we restrict which low-value goods are used to fill bundles. In restricted bag-filling, we only use the low-value goods from bundles of the divider's $\text{MMS}^{n}$ partition whose high-value goods are no longer available (since those goods have been allocated in some previous round or assigned during pairing), as highlighted in \cref{fig:lone_divider_example} by teal boxes.

The main reason behind this restriction is to ensure that not too much value is `wasted' in the bag-filling procedure (as discussed in \Cref{rem:bag-filling-example}). 
The restricted bag-filling relies on the strong normalization assumption - this way we can guarantee that the low-value goods used in restricted bag-filling have sufficient value when bundled with at most one high-value good.

Our proof follows by showing that if $n' \geq 2s$, it is possible to form $n' - s$ bundles through restricted bag-filling. 
Otherwise if $n' < 2s$, using restricted bag-filling we can form $s$ bundles. In this case, we still need to form an additional $n' - 2s$ bundles, which is possible through simple bag filling by targeting for each bundle to have at most $\frac{3}{2}$ value. In the appendix, we provide formal proofs for each of these steps.

\begin{remark}
The strong normalization assumption in the proof of \cref{thm:two_thirds_existence} requires modifying the values of some goods depending on an MMS partition for agent $i$. Furthermore, the restricted bag-filling phase of \cref{thm:two_thirds_existence} selects available low-value goods from specific bundles of an MMS partition for agent $i$. Because of these dependencies on finding MMS partitions, \cref{thm:two_thirds_existence} does not imply a polynomial time algorithm for $\MMS{\frac{2}{3}}{1}$.
\end{remark}

We show that by relaxing the strong normalization assumption to the weaker normalization of \cref{lem:normalize}, we can devise a polynomial time algorithm for $\MMS{\frac{2}{3}}{1}$ when $n < 9$. Essentially, this algorithm is a simplified variant of our previous procedure that only includes simple bag-filling and pairing to form the required bundles in the lone divider procedure. Intuitively, since the total value of low-value goods is guaranteed to be at least $1 - v_{i}(g_{1})$, at least one bundle can be allocated during bag-filling. When $n < 9$, adding the pairing phase forms $\floor{\frac{n}{2}} + 1 > \Targ$ bundles. The detailed algorithm and the full proof of correctness is provided in the appendix.

\begin{restatable}{theorem}{thmtwothirdsproof}\label{thm:two_thirds_poly}
For $n < 9$, $\MMS{\frac{2}{3}}{1}$ can be computed in polynomial time. 
\end{restatable}

\cref{thm:two_thirds_poly} only guarantees a $\MMS{\frac{2}{3}}{1}$ allocation when $n < 9$. We show by construction that this is tight. We construct a family of instances with $n \geq 9$ in which a small error in computing the MMS bound causes bag-filling to stop before $\Targ$ bundles have been allocated. The challenge presented here is not unique to our algorithm. Any algorithm that satisfies $\MMS{\frac{2}{3}}{1}$ must be able to detect that $\text{MMS}_{i}^{n} < 1$ for all agents. However, this task is intractable even when agents have identical valuations \citep{Bouveret2016}. 

\begin{remark}[Tightness for \cref{thm:two_thirds_poly}]
\label{eg:approx_failure}
Consider $n \geq 9$ agents with identical valuations as follows: $n-1$ high-value goods of value $1 - \epsilon -\tilde{\epsilon}$, two goods of value $\frac{1}{2} - \epsilon$, one good of value $(n+1)\epsilon$, and $n-1$ goods of value $\tilde{\epsilon}$, where $\epsilon \gg \tilde{\epsilon}$. Then, $\text{MMS}_{i}^{n} = 1 - \epsilon$ for all agents.
During the bag-filling step in \cref{thm:two_thirds_poly}, three bundles are allocated accounting for three goods of value $1 - \epsilon - \tilde{\epsilon}$, both goods of value $\frac{1}{2} - \epsilon$, and the good with value $(n+1)\epsilon$. After bag-filling, the remaining high-value goods will be paired together to create a total of $\floor{\frac{n - 4}{2}}$ bundles of paired high-value goods. Thus, this algorithm only satisfies $\floor{\frac{n}{2}} + 1$ agents.
\end{remark}

Combining this construction with Theorem~\ref{thm:two_thirds_poly} implies that no smaller counter-example exists. We note that despite this bound, the algorithm is still practical as there is no bound on the number of goods. For example, more than 99\% of instances in the \textit{Spliddit} dataset deal with less than $9$ agents.

The analysis in \cref{thm:two_thirds_existence} and \cref{thm:two_thirds_poly} are indifferent to the set of agents $N'$. Consequently, we may select any subset of $\Targ$ agents to be given their full MMS. An immediate consequence of this result, together with \cref{prop:MMSAlpha}, is that there always exists an allocation of goods such that each agent receives at least the value of $\text{MMS}^{\ceil{\frac{3n}{2}}}$. Moreover, this allocation can be computed in polynomial time when $\ceil{\frac{3n}{2}} < 9$.

\begin{corollary} \label{cor:MMS3n2}
An allocation satisfying $\text{MMS}^{\ceil{\frac{3n}{2}}}$ always exists. Moreover, such allocations can be computed in polynomial time when $n < 6$.
\end{corollary}

\section{Empirical Evaluations} \label{sec:experiments}
\begin{figure}
    \centering
     \includegraphics[width=\linewidth]{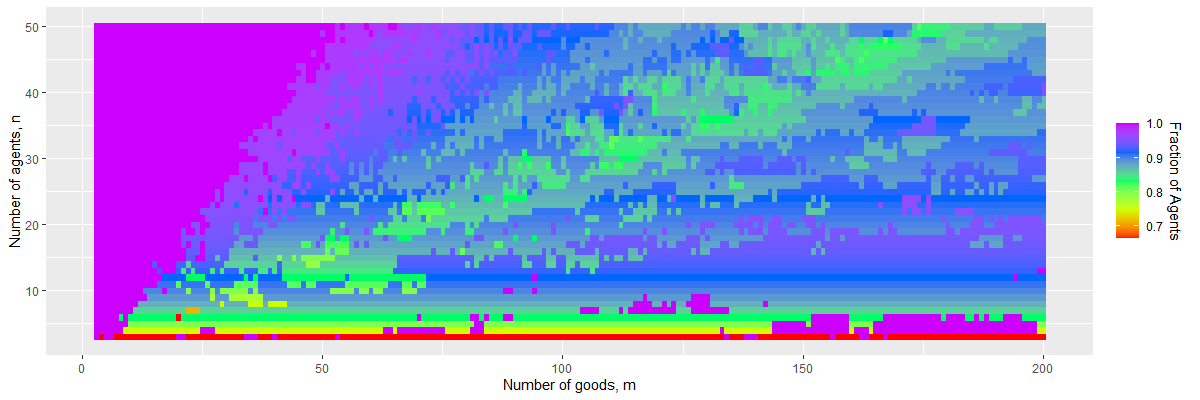}
    \caption{\small Fraction of agents, $\alpha$, receiving their maximin share, i.e., $\MMS{\alpha}{1}$.}
    \label{fig:empricalresults_wideplot}
\end{figure}

The polynomial algorithm of \cref{thm:two_thirds_poly} guarantees $\MMS{\frac{2}{3}}{1}$ only for $n < 9$. We empirically evaluate an extension of this algorithm on instances with up to $50$ agents and $200$ goods. 

This variant of the algorithm does not rely on the strong normalization assumption (see the appendix for the detailed description). It is identical to the algorithm of \cref{thm:two_thirds_poly} until either no more bundles can be allocated during the lone divider procedure or all $\Targ$ agents that were initially selected have received a bundle valued at least 1. In the latter case, any remaining goods are distributed among the unselected $n - \Targ$ agents using bag-filling with priority given to agents who accept the smallest bundles. Notice that this algorithm does not guarantee that all initially selected agents receive their MMS.

We focus on ordered instances---as the most difficult instances in achieving MMS~\citep{Bouveret2016}---and generate 1,000 instances for each combination of $n$ and $m$. Instances are sampled uniformly at random, ordered, and scaled such that $v_{i}(M) = n$ for all agents $i \in N$.


Figure~\ref{fig:empricalresults_wideplot} illustrates the fraction of agents who receive their MMS for $n = 3$ to $50$ agents and $m = 3$ to $200$ goods. In almost all instances, the algorithm goes beyond the $\frac{2}{3}$ bound: on average across all instances, more than 90\% of agents receive their MMS. Moreover, the fraction of agents receiving their MMS improves as either $n$ or $m$ increases.
We also observe linear bands where a lower fraction of agents are satisfied due to the ratio of goods to agents. In addition, when $m < 2n$ a large fraction of the agents receive their MMS and are removed during the reduction phase of normalization (see \cref{sec:preliminaries}). In the appendix, we provide further discussions and results.

\section{Discussion}
\Cref{thm:two_thirds_existence} proves the existence of $\MMS{\frac{2}{3}}{1}$ for any $n$, which implies the existence of $\text{MMS}^{\lceil\frac{3n}{2}\rceil}$. Therefore, improving the bound on $\MMS{\alpha}{1}$ for $\alpha > \frac{2}{3}$ and closing the gap between $\text{MMS}^{\lceil\frac{3n}{2}\rceil}$ and $\text{MMS}^{n+1}$ is an intriguing future direction.
\Cref{thm:two_thirds_poly} provides a tractable approach for computing $\MMS{\frac{2}{3}}{1}$ for $n < 9$. Yet, computing such allocations for any $n$, if possible, will require further techniques to circumvent computing exact MMS bounds. %
Another interesting avenue for future research is exploring PTAS algorithms that guarantee optimal-MMS \citep{heinen2018approximation} for a fraction of agents to achieve $\MMS{\alpha}{1}$ for $\alpha \in [\frac{2}{3}, \frac{n-1}{n}]$.

\section*{Acknowledgments}
Hadi Hosseini acknowledges the support from NSF grant \#1850076. We thank Sawyer Welden for initial help with the experimentation, and Nisarg Shah for helpful discussions. 
We are especially grateful to Erel Segal-Halevi for his valuable input that helped us improve the proof of \cref{thm:two_thirds_existence}. 
We thank the anonymous reviewers for their very helpful comments and suggestions.

\bibliographystyle{ACM-Reference-Format}
\bibliography{ref_MMS}


\clearpage
\appendix



\begin{center}
	\Large{Appendix}
\end{center}

\section{Material Omitted from Section~\ref{sec:optimalMMS}}
This section includes all the detailed discussions on constructing tensors $S$ and $T$, the relevant lemmas, and the missing proofs.

We begin with the generalized version of Theorem~\ref{thm:optimal_mms_failure}. We observe that when $d = \floor{\frac{n}{2}}$ the following theorem becomes Theorem~\ref{thm:optimal_mms_failure}.

\begin{restatable}{theorem}{thmOptMMSCounter}\label{thm:opt_mms_counter}
For any $n \geq 4$ and any $d \leq \floor{\frac{n}{2}}$, there exists an instance with $O(dn)$ goods where any optimal-MMS allocation guarantees at most $\ceil{\frac{n}{d}} + 1$ agents their MMS value.
\end{restatable}

Our construction is broken into three components: a tensor $S$, a tensor $T$, and $d$ groups of agents. The entries of tensor $S$ are chosen such that equi-partitions of the entries correspond to a set of slices of $S$. An equi-partition is a partition of the entries of a tensor where each set of the partition has the same sum. Here, all equi-partitions refer to $n$-partitions. Tensor $T$ provides perturbation so that the only equi-partitions of $S+T$ are aligned slices.

All tensors in our construction are order $d$ with dimensions $(n \times n \times \ldots \times n)$.\footnote{When referring to a tensor, the ``order'' is the number of dimensions that the tensor is embedded in. The ``dimensions'' of a tensor, however, refers to the number of entries along each axis, similar to the dimensions of a matrix.} We say that tensor $S$ is indexed by a $d$-tuple $(x_{1}, \ldots, x_{d})$ corresponding to entry $S[x_{1}, \ldots, x_{d}]$ or $S_{x_{1}, \ldots, x_{d}}$ for short. The $(d-1)$-order slices of tensor $S$ along dimension $j$ are given by $S_{j}(1)$ to $S_{j}(n)$ where $S_{j}(i) = \{S_{x_{1} \ldots x_{d}} \mid x_{j} = i\}$. 
We also define $S_{i; x_{j} = k}$ to be the entry where all indices are $i$ except index $x_{j} = k$. For example, $S_{i; x_{1} = n} = S_{nii\ldots i}$. Tensor addition is entry-wise ($(S+T)_{x_{1}, x_{2}, \ldots, x_{d}} = S_{x_{1}, x_{2}, \ldots, x_{d}} + T_{x_{1}, x_{2}, \ldots, x_{d}})$. Unless otherwise specified, the entries of a tensor have value $0$.

Figure~\ref{fig:tensor_shape} depicts the indexing of an order $3$ tensor with dimensions $(6 \times 6 \times 6)$. Figure~\ref{fig:slices} illustrates the distinct sets of slices of such a tensor. Lastly, Figure~\ref{fig:T_shape} depicts the non-zero entries of tensors $S$ and $T$ in our construction when $n = 6$ and $d = 3$.

\subsection{Construction of Tensor $S$}

Let $S$ be an order $d$ tensor whose non-zero entries are as follows:
\begin{itemize}
    \item For $1 \leq i < n$ and $j \in [d]$, $S_{ii\ldots i} = \frac{d^{n-i} - 1}{d^{n-i}}$
    \item For $1 \leq i < n$ and $j \in [d]$, $S_{i; x_{j} = n} = \frac{1}{d-1}\cdot\frac{1}{d^{n-i}}$
    \item $S_{nn\ldots n} = 1 - \sum_{i = 1}^{n-1} \frac{1}{d-1}\cdot\frac{1}{d^{n-i}}$
\end{itemize}

\begin{lemma}\label{lem:S_1slice}
Each $(d-1)$-order slice $S_{i}(j)$ has sum $1$.
\end{lemma}
\begin{proof}
Consider the non-zero entries within slice $S_{j}(i)$. First when $i \neq n$, the only non-zero entries of $S_{j}(i)$ are $S_{ii\ldots i}$ and $S_{i;x_{k}=n}$ for each $k \neq j$. Thus:
\begin{align*}
    \sum_{x \in S_{j}(i)} x &= S_{ii\ldots i} + \sum_{k \neq j}^{d} S_{i;x_{k}=n} \\
    &= \frac{d^{n-i} - 1}{d^{n-i}} + \sum_{k \neq j}^{d} \frac{1}{d-1}\cdot\frac{1}{d^{n-i}} \\
    &= \frac{d^{n-i} - 1}{d^{n-i}} + (d-1) \frac{1}{d-1}\cdot\frac{1}{d^{n-i}} \\
    &= \frac{d^{n-i} - 1}{d^{n-i}} + \frac{1}{d^{n-i}} \\
    &= 1
\end{align*}
When $i = n$, the non-zero entries of slice $S_{j}(i)$ are $S_{nn\ldots n}$ and $S_{k;x_{j}=n}$ for each $k \neq n$. Thus:
\footnotesize
\begin{align*}
    \sum_{x \in S_{j}(n)} x &= S_{nn\ldots n} + \sum_{k \neq n}^{n} S_{k;x_{j}=n} \\
    &= 1 - \sum_{i = 1}^{n-1} \frac{1}{d-1}\cdot\frac{1}{d^{n-i}} + \sum_{k = 1}^{n-1} \frac{1}{d-1}\cdot\frac{1}{d^{n-k}} \\
    &= 1
\end{align*}
\end{proof}
\normalsize
The combinatorial structure of $S$ has a few other desirable properties. First observe that the entries on the main diagonal $(S_{ii \ldots i})$ all have value greater than $\frac{1}{2}$. When considering equi-partitions of the entries of $S$, these entries must all be included in separate sets. Furthermore, the entries off the main diagonal increase when approaching $S_{nn\ldots n}$ in such a way that all $d$ of the entries with value $\frac{1}{d-1}\cdot\frac{1}{d^{n-i}}$ are required to account for a single entry with value $\frac{1}{d-1}\cdot\frac{1}{d^{n-(i+1)}}$.

\begin{lemma}\label{lem:S_slices}
Any equi-partition of the non-zero entries of $S$ corresponds to a set of slices\\ $\{S_{j_{1}}(1), S_{j_{2}}(2), \ldots, S_{j_{n}}(n)\}$.
\end{lemma}
\begin{proof}
Lemma~\ref{lem:S_1slice} implies that a set of slices of the form $\{S_{j_{1}}(1), S_{j_{2}}(2), \ldots, S_{j_{n}}(n)\}$ forms an equi-partition. In order to show that these equi-partitions are unique, consider the set which contains $S_{11\ldots 1}$. Since $1 - S_{11\ldots 1} = \frac{1}{d^{n-1}} < \frac{1}{d-1}\cdot\frac{1}{d^{n-2}}$. Thus the only entries which can be included with $S_{11\ldots 1}$ are of the form $\frac{1}{d-1}\cdot\frac{1}{d^{n-1}}$. Importantly, $d-1$ of these entries uniquely identify a $(d-1)$-order slice $S_{j_{1}}(1)$ where $S_{1;x_{j_{1}} = n}$ is the missing entry with value $\frac{1}{d-1}\cdot\frac{1}{d^{n-1}}$.

By induction, with a similar argument at each step, each successive entry on the main diagonal must be included in a set which corresponds exactly to a slice $S_{j_{i}}(i)$.
\end{proof}

\subsection{Construction of Tensor $T$}
In Lemma~\ref{lem:S_slices}, the slices of $S$ are not necessarily aligned slices. For example, in the case of three dimensions, the first slice might be frontal while the second one might be vertical. Since all slices have the same sum, there is no reason an agent wouldn't compute MMS by taking slices at random. In order to rectify this, we construct tensor $T$ so that equi-partitions of $S + T$ correspond only to aligned slices.

Let tensor $T$ be an order $d$ tensor with non-zero values as follows:
\setlist{nolistsep}
\begin{itemize}[noitemsep]
    \item For $1 \leq i < n-1$, $j \in [d]$, $T_{i; x_{j}=n} = -r_{ij}$
    \item For $1 \leq i < n-1$, $j \in [d]$, $T_{i; x_{j}=i+1} = u_{ij}$
    \item For $j \in [d]$, $T_{n-1; x_{j} = n} = x_{j}$
    \item $T_{nn\ldots n} = z_{1}$
\end{itemize}
Where
\begin{itemize}[noitemsep]
    \item $r_{ij} = \epsilon^{d(n-1) - di - j + 1}$
    \item $u_{1j} = r_{1j}$
    \item $u_{ij} = \frac{1}{d-1}\left(\sum_{k \neq j}^{d}{-u_{i-1,k}}\right) + \frac{d-2}{d-1}u_{i-1,j} + r_{ij} \approx r_{ij}$
    \item $x_{j} = u_{n-1,j} - r_{n-1,j}$
    \item $z_{j} = \left(\sum_{i = 1}^{n-2}{r_{ij}}\right) - x_{j}$
\end{itemize}

It is important to observe that the entries of $T$ are all proportional to powers of a very small $\epsilon$ term. While some of these values are negative, the negative entries only occur in indices where $S$ has positive value. Thus if $\epsilon$ is small enough, $S + T$ has only non-negative values. Another important observation is that the slices of $T$ all have sum $0$.

\begin{figure*}[t]
\nocaption
\centering \scriptsize
    \begin{subfigure}[t]{0.2\textwidth}
        \centering 
        \subfile{images/S_Shape_fig}
        \repeatsubcaption{fig:tensor_shape}{The indexing of tensor $S$.}
    \end{subfigure}\hfill
    \begin{subfigure}[t]{0.32\textwidth}
        \centering 
        \subfile{images/Slices_fig}
        \repeatsubcaption{fig:slices}{The slices of tensor $S$.}
    \end{subfigure}\hfill
    \begin{subfigure}[t]{0.33\textwidth}
    \centering
    \subfile{images/T_Shape_fig}\hfill
    \repeatsubcaption{fig:T_shape}{The non-zero entries of $S$.}
    \end{subfigure}
  \repeatcaption{fig:tensors}{\small The shape and slices of order $3$ tensors with dimensions $(6 \times 6 \times 6)$.}
\end{figure*}

\begin{lemma}\label{lem:T_ZeroSlices}
Each $(d-1)$-order slice $T_{j}(i)$ has sum $0$.
\end{lemma}
\begin{proof}
We first address the case of $i = 1$. Slice $T_{j}(1)$ contains the non-zero entries $T_{1;x_{k} = n}$ for each $k \neq j$ and $T_{1;x_{k}=2}$ for each $k \neq j$. For each $k \neq j$, the entry $T_{1;x_{k}=n} = -r_{1j}$ and $T_{1;x_{k}=2} = u_{1j}$. Since $u_{1j} = -r_{1j}$ the sum of these two entries is $0$. Adding over all $k \neq j$ yields that the first slice (independent of $j$) has sum $0$.

Now consider slice $T_{j}(i)$ for $1 < i < n - 1$. The non-zero entries of $T$ in this slice are $T_{i;x_{k} = n}$ and $T_{i;x_{k}=i+1}$ for $k \neq j$ and $T_{i-1;x_{j}=i}$. Adding these entries together yields:
\small
\allowdisplaybreaks
\begin{align*}
    &\sum_{x \in T_{j}(i)} x \\
    &= \sum_{k \neq j}^{d}T_{i;x_{k}=i+1} + \sum_{k \neq j}^{d}T_{i;x_{k}=n} + T_{i-1; x_{j}=i} \\
    &= \sum_{k \neq j}^{d}\left(u_{ik}\right) + \sum_{k \neq j}^{d}\left(-r_{ik}\right) + u_{i-1,j} \\
    &= \sum_{k \neq j}^{d}\left(\frac{1}{d-1}\left(\sum_{l \neq k}^{d}{-u_{i-1,l}}\right) + \frac{d-2}{d-1}u_{i-1,k} + r_{ik}\right)\\
    &\qquad + \sum_{k \neq j}^{d} \left(-r_{ik}\right) + u_{i-1,j} \\
    &= \sum_{k \neq j}^{d}\left(\frac{1}{d-1}\left(\sum_{l \neq k}^{d}{-u_{i-1,l}}\right) + \frac{d-2}{d-1}u_{i-1,k}\right) + u_{i-1,j} \\
    &= \frac{1}{d-1}\left(\sum_{k \neq j}^{d} \sum_{l \neq k}^{d} (-u_{i-1,l}) + (d-2)\sum_{k \neq j}^{d}(u_{i-1,k})\right) + u_{i-1,j} \\
    &= \frac{1}{d-1}\left((d-1)(-u_{i-1,j})\right) + u_{i-1,j} = 0\\
\end{align*}
\normalsize
The last step follows because each of the $(d-1)$ summations (for each $k \neq j$) adds a single $-u_{i-1,j}$ term. For each $k \neq j$, $(d-2)$ of the summations (all but $l = k$ of the $(d-1)$ inner summations) add a $-u_{i-1,k}$ term.

We next address the case where $i = n-1$. In this case, we observe that $i+1 = n$ so the only entries in this slice are $T_{n-1;x_{k} = n}$ for $k \neq j$ and $T_{n-2;x_{j}=n-1}$. Adding these entries together yields:
\small
\begin{align*}
    &\sum_{x \in biT_{j}(n-1)} x \\
    &= \sum_{k \neq j}^{d}T_{n-1;x_{k}=n} + T_{n-1;x_{j}=n-1} \\
    &= \sum_{k \neq j}^{d}\left(x_{k}\right) + u_{n-2,j} \\
    &= \sum_{k \neq j}^{d}\left(u_{n-1,k} - r_{n-1,k}\right) + u_{n-2,j} \\
    &= \sum_{k \neq j}^{d}\Bigg(\Bigg(\frac{1}{d-1}\left(\sum_{l \neq k}^{d} -u_{n-2,l}\right) \\
    &\qquad + \frac{d-2}{d-1}u_{i-2,k} + r_{n-1,k}\Bigg) - r_{n-1,k}\Bigg) + u_{n-2,j} \\
    &= \sum_{k \neq j}^{d}\left(\frac{1}{d-1}\left(\sum_{l \neq k}^{d} -u_{n-2,l}\right) + \frac{d-2}{d-1}u_{i-2,k} + \right) + u_{n-2,j} \\
    &= \frac{1}{d-1}\left(\sum_{k \neq j}^{d}\sum_{l \neq k}^{d} \left(-u_{n-2,l}\right) + (d-2)\sum_{k \neq j}^{d}\left(u_{i-2,k}\right) + \right) + u_{n-2,j} \\
    &= \frac{1}{d-1}((d-1)(-u_{n-2,j})) + u_{n-2,j} = 0\\
\end{align*}

\normalsize
We last address the case where $i = n$. In this case, the non-zero entries in this slice are $T_{i; x_{j} = n}$ for each $i \in [n-2]$, $T_{n-1;x_{j} = n}$, and $T_{nn \ldots n}$. Adding these entries together yields:
\small
\allowdisplaybreaks
\begin{align*}
    \sum_{x \in T_{j}(n)} x &= \sum_{k = 1}^{n-2} T_{k;x_{j}=n} + T_{n-1;x_{j}=n} + T_{nn \ldots n} \\
    &= \sum_{k = 1}^{n-2} \left(-r_{kj}\right) + x_{j} + z_{1} \\
\end{align*}
\normalsize
If $j = 1$ then the term $\sum_{k = 1}^{n-2} \left(-r_{kj}\right) + x_{j} = -z_{1}$ so this trivially sums to $0$. If $j \neq 1$, then we first observe that the sum of all entries of $T$ is $0$ by adding the sums of the slices $\{T_{1}(1), T_{1}(2), \ldots, T_{1}(n)\}$. Since $\sum_{i = 1}^{n}\sum_{x \in T_{j}(i)} x = 0$ and $\sum_{k = 1}^{n-1} \sum_{x \in T_{j}(k)} x = 0$, subtracting the two equations gives $\sum_{x \in T_{j}(n)} = 0$.
\end{proof}

Combining Lemmas~\ref{lem:S_slices}~and~\ref{lem:T_ZeroSlices} together immediately shows that any set of aligned slices of $(S+T)$ forms an equi-partition of the non-zero entries where each set has sum $1$. We now show that these partitions are unique. For intuition, the key point in our construction is that $(d-1)$ of the $T_{i;x_{j}=i+1}$ terms all lie in the same slice $T_{j}(i)$. Once $j$ has been chosen for $T_{j}(1)$ the extra term $T_{1; x_{j}=2}$ lies in slice $T_{j}(2)$.

\begin{lemma}\label{lem:ST_Equipartitions}
The only equi-partitions of the non-zero entries of $(S+T)$ correspond to aligned slices.
\end{lemma}
\begin{proof}
First observe that if $\epsilon$ is sufficiently small then the structure of Lemma~\ref{lem:S_slices} still holds. Consider first the set of an equi-partition which contains $(S+T)_{11 \ldots 1}$. In order for this set to have value $1$, it must have $(d-1)$ of the entries which receive $\frac{1}{d-1} \cdot \frac{1}{d^{n-1}}$ value from $S$. Again, selecting $(d-1)$ of these uniquely characterizes a slice $(S+T)_{j}(1)$. Observe that these entries have negative perturbations from $T$ corresponding to $-r_{1k}$ for each $k \neq j$. In order for these perturbations to be offset so that the set containing $(S+T)_{11 \ldots 1}$ has sum $1$, the corresponding entries with values $u_{1k}$ for $k \neq j$ must be included in this set. These entries correspond exactly to the slice $(S+T)_{j}(1)$.

As in the proof of Lemma~\ref{lem:S_slices}, we observe that the entries on the main diagonal $(S+T)_{ii \ldots i}$ must be in separate sets. Likewise, because $\epsilon$ is small, a similar argument enforces that $(d-1)$ of the entries with value $\frac{1}{d-1} \cdot \frac{1}{d^{n-i}}$ must be included in this set. We address the negative perturbations (from $T$) of these entries by induction. Observe that the smallest perturbations from $T$ occur in the earliest slices ($(1,1) = \argmin_{i,j} r_{ij}$). Note also that $(u_{1j}, -r_{ij}$) is the only pair of perturbations from $T$ which are excluded from slice $(S+T)_{j}(1)$. A similar argument for slice $(S+T)_{j}(1)$ holds for slice $(S+T)_{j}(2)$: in order for the set containing $(S+T)_{22 \ldots 2}$ to have value $1$, it must also have $(d-1)$ entries with the $-r_{2k}$ perturbations from $T$. Since $\epsilon^{2} \ll \epsilon$, each $-r_{2k}$ entry can only be offset with entries with perturbation $u_{2k}$. (All but one of the smaller perturbation entries have already been assigned to slice $1$.) However, by Lemma~\ref{lem:T_ZeroSlices}, we see that adding together all $k \neq l$ entries of the form $-r_{2k}$ and $u_{2k}$ leaves value $u_{1l}$ (since $u_{1l} + \sum_{k \neq l}^{d} u_{2k} + \sum_{k \neq l}^{d} -r_{2k} = 0$). Observe that only a single entry of the form $u_{1l}$ is available: notably $u_{1j}$. Thus the set which contains $(S+T)_{22 \ldots 2}$ corresponds to the slice $(S+T)_{j}(2)$.

By induction on this same argument we see that each set containing the entry $(S+T)_{ii \ldots i}$ for $1 \leq i < n$ must correspond to a slice $(S+T)_{j}(i)$ for some fixed $j$ (after $j$ is chosen for the first slice). We conclude this proof by observing that the remaining entries all lie in the last slice $(S+T)_{j}(n)$ and have total sum $1$ by Lemmas~\ref{lem:S_slices}~and~\ref{lem:T_ZeroSlices}.
\end{proof}

With the details of $S$ and $T$, we are now ready to prove the main theorem. We restate it here for convenience.

\thmOptMMSCounter*
\begin{proof}
Lemma~\ref{lem:ST_Equipartitions} implies that there are $d$ unique equi-partitions of the entries of $(S+T)$. We group agents arbitrarily into $d$ groups of size at least $2$ and at most $\ceil{\frac{n}{d}}$ each. We next show how to perturb $(S+T)$ for each group so that each group has a unique equi-partition. Furthermore, we enforce that the only entry with positive perturbation for all agents is the last entry on the main diagonal.
Let $P^{j}$ be an order $d$ tensor with non-zero values as follows:
\setlist{nolistsep}
\begin{itemize}[noitemsep]
    \item For $1 \leq i < n$, $P^{j}_{i;x_{j}=n} = -\tilde{\epsilon}$
    \item $P^{j}_{nn\ldots n} = (n-1) \tilde{\epsilon}$
\end{itemize}

The entries where $P^{j}$ is negative correspond to entries where $(S+T)$ has positive value. Thus $(S+T+P^{j})$ has only non-negative entries if $\tilde{\epsilon}$ is sufficiently small. Observe that the slices $\{(S+T+P^{j})_{j}(1), (S+T+P^{j})_{j}(2), \ldots, (S+T+P^{j})_{j}(n)\}$ form an equi-partition of the entries of $(S+T+P^{j})$ as $P^{j}$ has non-zero values only in the last slice $(S+T+P^{j})_{j}(n)$ and the total sum of its non-zero entries is $0$. If $\tilde{\epsilon}$ is sufficnetly small (so that the structure of $S + T$ dominates when creating an equi-partition), then none of the other sets of slices of $(S+T+P^{j})$ forms an exact equi-partition as every slice except for the last contains exactly one of the $-\tilde{\epsilon}$ perturbations.

We now create an instance $I = \ins{N, M, V}$ where $M$ contains one good for each non-zero entry of $(S+T)$. Agents in group $j$ value the goods according to the values of $(S+T+P^{j})$. Since every agent is able to form an equi-partition of the entries of their valuation tensor, all agents have $\text{MMS}_{i} = 1$. By allocating goods according to slices, each agent receives at least $(1 - \tilde{\epsilon})$ of their MMS value. Since $\tilde{\epsilon}$ is very small, the only way all agents can receive at least $1 - \tilde{\epsilon}$ value is if the goods are allocated according to aligned slices (Lemma~\ref{lem:ST_Equipartitions}). Under this allocation, at most one group and a single additional agent who receives the last slice (since $(S+T+P^{j})_{j}(n)$ has positive perturbation $(n-1)\tilde{\epsilon}$) receive their MMS value. Lastly we observe that there only $O(dn)$ goods in $M$.
\end{proof}

Again we note that if $d = \floor{\frac{n}{2}}$ in the above theorem, then we see that only $3$ agents ($4$ if $n$ is odd) receive their full MMS in any optimal-MMS allocation.

\section{Material Omitted from Section~\ref{sec:approx}}
\subsection{An alternative algorithm for achieving $\text{MMS}^{3}$ for two agents}

\subfile{algs/alg_n=3}

\begin{lemma}[\citet{Bouveret2016}]\label{lem:high_value_reduction}
Given an instance $I = \ins{N, M, V}$, for any agents $i\in N$ and any good $g\in M$, we have
$\text{MMS}_{i}^{n-1}(M \setminus g) \geq \text{MMS}_{i}^{n}(M)$.
\end{lemma}

\begin{theorem} \label{thm:2agents}
For $2$ agents, there exists a polynomial time algorithm that computes an allocation satisfying $\text{MMS}^{3}_{i}$ for all $i\in N$.
\end{theorem}
\begin{proof}
Algorithm~\ref{alg:n=3} starts by scaling the valuations in instance $I$ such that for all agents $i\in N$ and all goods in $g\in M$ we have $v'_{i}(g) = \frac{3}{v_{i}(M)} v_{i}(g)$. 
By scale invariance (Lemma \ref{lem:scale_invariance}) and since $\text{MMS}_{i} \leq \frac{v_{i}(M)}{n}$, we have $\forall i\in N$, $\text{MMS}_{i}^{3} \leq 1$. 

\paragraph{Case I.}
Consider two agents $i,j\in N$. 
Suppose agent $i$ is allocated a good $g$ with value $v'_{i}(g)\geq 1$ (Line 2), then by definition $v'_{i}(A_{i})\geq \text{MMS}_{i}^{3}$.
The remaining agent $j$ receives $A_{j} = M\setminus \{g\}$. By definition of MMS, $v_{j}(M\setminus \{g\}) \geq \text{MMS}_{j}^{2}(M\setminus\{g\})$.
By Lemma \ref{lem:high_value_reduction}, $\text{MMS}_{j}^{2}(M\setminus\{g\}) \geq \text{MMS}_{j}^{3}(M)$.
Therefore, $v_{j}(M\setminus \{g\}) \geq \text{MMS}_{j}^{3}(M)$.

\paragraph{Case II.}
Suppose that after scaling the valuations, no agent values any good greater or equal  to 1. Then, Algorithm~\ref{alg:n=3} proceeds by running an EF1 algorithm (e.g. round-robin), which results in an EF1 allocation.
For each agent $i\in N$, let $g^{*}$ be agent $i$'s highest valued good in agent $j$'s bundle. By EF1 we have $v_i(A_i) \geq v_i(A_j)-v_i(g^*)$. Adding $v_i(A_i)$ to both sides of the inequality, and since there are two agents, we get $2v_i(A_i) \geq v_i(M) - v_i(g^{*})$. Since $v_{i}(M) = 3$, $v_{i}(g^{*}) < 1$, we have that $v_i(A_i) \geq \frac{3}{2} - \frac{1}{2}$, implying that $v_i(A_i) \geq \text{MMS}^3_i$ according to Lemma~\ref{lem:scale_invariance} and the MMS bound, that is, $\text{MMS}_{i} \leq \frac{v_{i}(M)}{n}$.
\end{proof}

\subsection{Proof of \cref{lem:normalize}}

\lemnormalize*
\begin{proof}
Given an instance $I = \ins{N, M, V}$, we can compute a normalized instance $I' = \ins{N', M', V'}$ as follows: first, order the instance $I$ according to \cref{lem:order}. Call this instance $\hat{I}= \ins{\hat{N}, \hat{M}, \hat{V}}$ such that $\hat{n} = \hat{N}$ and $\hat{m} = \hat{M}$. Then, in each iteration by \Cref{lem:scale_invariance}, we scale the instance such that $\bar{v}_{i}(\hat{M}) = \bar{n}$. Then, apply one of the valid reductions among $\{g_{1}\}$ or $\{g_{\hat{n}}, g_{\hat{n}+1}\}$. An agent receives a bundle according to the valid reduction (\cref{lem:n_goods_reduction}). We then update the instance $\hat{I}$ by removing the allocated agent and their bundle and relabel goods from $g_{1}$ to $g_{\hat{m}}$. The procedure stops once there are no valid reductions remaining.

Let $k \leq n$ be the number of agents who were allocated and removed during the previous phase. 
By the definition of a valid reduction, the MMS of each agent only weakly increases after any reduction. Thus each agent who receives a bundle during the reduction phase receives at least their $\text{MMS}_{i}^{n}(M)$.
Any $\MMS{\alpha}{\beta}$ allocation on the resulting normalized instance $I'$ guarantees that $\floor{\alpha n'}$ agents receive at least $\beta \text{MMS}_{i}^{n'} \geq \beta \text{MMS}_{i}^{n}$ because $n' \leq n$. Combining this with the fact that $\floor{\alpha (n-k)} + k \geq \floor{\alpha n}$ implies that at least $\floor{\alpha n}$ agents receive $\beta \text{MMS}_{i}^{n}$ in the ordered instance corresponding to $I$. By \cref{lem:order} when we unorder the allocation, the resulting allocation for $I$ satisfies $\MMS{\alpha}{\beta}$.
\end{proof}

\subsection{Proof of \cref{thm:GeneralN}}
We first show that a modification to the envy-graph algorithm~\citep{lipton2004approximately} provides a bounded MMS guarantee to $n-1$ on normalized instances.

\paragraph{The modified envy-graph algorithm.}
The modified envy-graph algorithm (due to \citet{barman2017approximation}) proceeds as follows: given a normalized instance, in each round the highest valued good (breaking ties arbitrarily) among the remaining goods is assigned to an agent that is not envied by any agent. One such agent is guaranteed to exist by resolving cyclic envy relations in the envy-graph.

\begin{lemma}[\citet{barman2017approximation}]\label{lem:envy_graph_efx}
For any ordered instance, the modified envy-graph algorithm guarantees EFX.
\end{lemma}

\begin{restatable}{lemma}{lemefxgeneraln}\label{lem:efx_generaln}
Given a normalized instance, the modified envy-graph algorithm satisfies $\MMS{\frac{n-1}{n}}{\frac{1}{2}(\frac{n+2}{n-1})}$.
\end{restatable}
\begin{proof}
We first show that no agent receives an empty bundle when $m\geq n$. 
Let $A$ be the allocation produced by the modified envy graph algorithm agents in $N'\subset N$ such that $|N'| = n-1$.
Consider the first $n-1$ goods in an ordered instance. In each of the first $n-1$ rounds of the envy graph algorithm, there exists at least one unenvied agent.
By assumption, we know that $v_{i}(g_{1}) < 1$, $v_{i}(M) = n$. This implies that $\forall i\in N, v_{i}(g_{n}) > 0$; otherwise there exists at least one good $g$ with value $v_{i}(g) \geq \frac{n}{n-1} > 1$.
Since for all $i\in N'$, $v_{i}(g_{n}) > 0$ and the fact that the instance is ordered, all agents also have positive value for all goods ordered before $g_{n}$, i.e., for all $i\in N$, $v_{i}(g_1) \geq v_{i}(g_2) \geq \ldots \geq v_{i}(g_n) > 0$.
Thus any agent who has not yet received any good envies each agent who has already received a good. 
Since the algorithm does not allocate any good to an envied agent, no agent is allocated a second good until all bundles are allocated at least one good, and at most one bundle can receive two goods from $\{g_{1}, \ldots, g_{n}\}$. Therefore, all bundles are nonempty after allocating the first $n$ good, i.e., $\forall i\in N'$, $|A_{i}| > 0$.\footnote{Indeed, two agents may exchange bundles through cycle resolution when the $n$th good is allocated.}

Next, we must show that each agent's bundle is valued $v_{i}(A_{i}) \geq \frac{1}{2}(\frac{n+2}{n-1})$ once all goods are allocated.

Let $S$ be the set of agents who receive a single good in the allocation $A$ and let $s = |S|$. Notice that $s< n$, since one agent does not participate in the algorithm ($|N'| = n-1$).
By Lemma~\ref{lem:envy_graph_efx}, $A$ satisfies EFX, i.e., for all $j \in [n-1]$, $v_{i}(A_{i}) \geq v_{i}(A_{j} \setminus \hat{g_{j}})$ for any good $\hat{g_{j}} \in A_{j}$.
By summing over the set of agents who received at least two goods we have,
\[
((n - 1) - s)v_{i}(A_{i}) \geq v_{i}(M \setminus (\bigcup_{j \in S} A_{j})) - \sum_{\substack{j \neq i\\j \notin S}}^{n-1}v_{i}(\hat{g_{j}}).
\]
\noindent Since $v_{i}(M) = n$ and $v_{i}(\bigcup_{j \in S} A_{j}) \leq s$, we can write:
\[
((n - 1) - s)v_{i}(A_{i}) \geq (n - s) - \sum_{\substack{j \neq i\\j \notin S}}^{n-1}v_{i}(\hat{g_{j}}) \\
\]
\noindent The $\hat{g}_{j}$ are distinct and each worth at most $v_{i}(g_{n})$. Thus the two with highest value among them are worth at most $v_{i}(\{g_{n}, g_{n+1}\}) < 1$. Pulling these items out of the sum, and upper-bounding the remaining goods $v_{i}(\hat{g}_{j}) \leq v_{i}(g_{n+2}) \leq v_{i}(g_{n+1}) < \frac{1}{2}$, we can write:

\allowdisplaybreaks
\begin{align*}
    ((n - 1) - s)v_{i}(A_{i}) &\geq (n - s) - 1 - \frac{1}{2}(n - s - 4) \\
    ((n - 1) - s)v_{i}(A_{i}) &\geq (n - s) - \frac{1}{2}(n - s - 2) \\
    v_{i}(A_{i}) &\geq \frac{n - s}{n - s - 1} - \frac{1}{2}(\frac{n - s - 2}{n - s - 1}) \\
    v_{i}(A_{i}) &\geq \frac{1}{2}(\frac{n - s + 2}{n - s - 1}) \\
    v_{i}(A_{i}) &\geq \frac{1}{2}(\frac{n + 2}{n - 1}).
\end{align*}
Notice that the last step holds because $s < n$.
\end{proof}

\thmGeneralN*
\begin{proof}
Consider an instance of the problem $I$ with $n$ agents. We begin by computing a normalized instance $I'$ from $I$ following the procedure of Lemma~\ref{lem:normalize}. On the normalized instance, we compute an $\MMS{\frac{n'-1}{n'}}{\frac{1}{2}(\frac{n'+2}{n'-1})}$ allocation (Lemma~\ref{lem:efx_generaln}). We now prove that the resulting allocation satisfies $\MMS{\frac{n-1}{n}}{\frac{1}{2}(\frac{n+2}{n-1})}$.

The key observation is that the value of $\beta$ is dependent on $n$ and is monotonic decreasing: $\frac{n' + 2}{n' - 1} \geq \frac{n + 2}{n - 1}$ when $n' \leq n$ and that $\frac{n'-1}{n'} n' + (n - n') = n - 1$. Thus each agent who is allocated a bundle in the reduced instance $I'$ receives at least $\text{MMS}_{i}^{n}$. Furthermore, each agent allocated a bundle during the reductions of Lemma~\ref{lem:normalize} receives their $\text{MMS}_{i}^{n} \geq \text{MMS}_{i}^{n}$. In total $n-1$ agents receive at least $\frac{1}{2}(\frac{n+2}{n-1})\text{MMS}_{i}^{n}$. After unordering, by \cref{lem:order} the resulting allocation satisfies $\MMS{\frac{n-1}{n}}{\frac{1}{2}(\frac{n+2}{n-1})}$.
\end{proof}

\section{Material Omitted from Section~\ref{sec:twothirdMMS}}
\subsection{Computing $\MMS{\frac{1}{2}}{1}$ allocations}

There are several plausible ways to achieve MMS for half of the agents, that is, $\MMS{\frac{1}{2}}{1}$.
We first provide a simple algorithm relying on EF1, and then show a generalization that enables us to exploit the already known MMS approximation guarantees.

\begin{proposition} \label{prop:EF1MMS}
If $\forall i \in N, \forall g \in M, v_{i}(g) < \frac{v_{i}(M)}{n}$, then any EF1 algorithm on half of the agents guarantees $\MMS{\frac{1}{2}}{1}$.
\end{proposition}
\begin{proof}
Consider an EF1 allocation $A$ that allocates goods to the half of the agents.
This implies that $v_{i}(A_{i}) \geq v_{i}(A_{j}) - v_{i}(g_{j})$ for some good $g_{j} \in A_{j}$. As the EF1 allocation is over $\floor{\frac{n}{2}}$ agents, the union over all bundles is $M$. Thus
\allowdisplaybreaks
\begin{align*}
    \sum_{j = 1}^{\floor{\frac{n}{2}}} v_{i}(A_{i}) &\geq \sum_{j = 1}^{\floor{\frac{n}{2}}} \left(v_{i}(A_{j}) - v_{i}(g_{j})\right) \\
    \floor{\frac{n}{2}}v_{i}(A_{i}) &\geq v_{i}(M) - \sum_{j = 1}^{\floor{\frac{n}{2}}} v_{i}(g_{j}) \\
    \floor{\frac{n}{2}}v_{i}(A_{i}) &\geq v_{i}(M) - \sum_{j = 1}^{\floor{\frac{n}{2}}} \frac{v_{i}(M)}{n} \\
    \floor{\frac{n}{2}}v_{i}(A_{i}) &\geq v_{i}(M) - \floor{\frac{n}{2}} \frac{v_{i}(M)}{n} \\
    v_{i}(A_{i}) &\geq \frac{1}{\floor{\frac{n}{2}}}\left(v_{i}(M) - \floor{\frac{n}{2}}\frac{v_{i}(M)}{n}\right) \\
    v_{i}(A_{i}) &\geq \frac{v_{i}(M)}{\floor{\frac{n}{2}}} - \frac{v_{i}(M)}{n} \\
    v_{i}(A_{i}) &\geq \frac{v_{i}(M)}{\frac{n}{2}} - \frac{v_{i}(M)}{n} \\
    v_{i}(A_{i}) &\geq \frac{2v_{i}(M)}{n} - \frac{v_{i}(M)}{n} \\
    v_{i}(A_{i}) &\geq \frac{v_{i}(M)}{n} \\
    v_{i}(A_{i}) &\geq \text{MMS}_{i}^{n}
\end{align*}
\end{proof}

The immediate consequence of Proposition~\ref{prop:EF1MMS} is an intuitive algorithm for satisfying $\MMS{\frac{1}{2}}{1}$: apply valid reductions according to Lemma~\ref{lem:n_goods_reduction} followed by applying any EF1 algorithm (e.g. Round Robin or Maximum Nash welfare) on half of the remaining agents in the reduced instance.

Next we show that running any $\beta.\text{MMS}$ approximation algorithm with $\beta \geq \frac{1}{2}$ on half of the agents is sufficient to guarantee $\MMS{\frac{1}{2}}{1}$.

\begin{proposition}
Any algorithm satisfying at least $\MMS{1}{\frac{1}{2}}$ also guarantees $\MMS{\frac{1}{2}}{1}$ given half of the agents.
\end{proposition}
\begin{proof}
Consider an algorithm that guarantees $\MMS{1}{\beta}$ where $\beta \geq \frac{1}{2}$.
In any $\text{MMS}_{i}^{n}$ partition each bundle in the partition has value at least $\text{MMS}_{i}^{n}$. If these $n$ bundles are paired together to form $\floor{\frac{n}{2}}$ bundles, the resulting allocation forms a lower bound for $\text{MMS}_{i}^{\floor{n/2}}$. Each of these $\floor{n/2}$ bundles has value at least $2\text{MMS}_{i}^{n}$ which implies that $\text{MMS}_{i}^{\floor{n/2}}(M) \geq 2 \text{MMS}_{i}^{n}(M)$. When this algorithm is run on half of the agents, those agents are guaranteed $\beta \text{MMS}_{i}^{\floor{n/2}} \geq \frac{1}{2} \text{MMS}_{i}^{\floor{n/2}}(M) \geq \frac{1}{2} 2\text{MMS}_{i}^{n}(M)$. Thus, the resulting allocation satisfies $\MMS{\frac{1}{2}}{1}$.
\end{proof}

\begin{remark}
An EF1 algorithm alone is not sufficient for computing $\MMS{1}{\frac{1}{2}}$, and thus, cannot be directly used for satisfying $\MMS{\frac{1}{2}}{1}$.
Consider $3$ agents with identical valuations over $7$ goods:

\begin{center}
\begin{tabular}{c | c c c c c c c}
                & $g_{1}$ & $g_{2}$ & $g_{3}$ & $g_{4}$ & $g_{5}$ & $g_{6}$ & $g_{7}$ \\\hline
    $v_{i}(g_{j})$  & 0.99 & 0.99 & 0.4 & 0.4 & 0.2 & 0.01 & 0.01
\end{tabular}
\end{center}

The allocation $A = (\{g_{1}, g_{4}, g_{7}\}, \{g_{2}, g_{5}\}, \{g_{3}, g_{6}\})$ is EF1 but not $\MMS{1}{\frac{1}{2}}$.
\end{remark}

\subsection{Proof $\MMS{\frac{2}{3}}{1}$ Existence for $n > 4$}

\begin{lemma}[\citet{aigner2019envy}] \label{lem:EFMatching}
If $|\mathcal{N}(X)| \geq |X|$, then a non-empty envy-free matching exists and can be found in polynomial time. 
\end{lemma}

\lemstrongnormalization*
\begin{proof}
By \cref{lem:scale_invariance}, we can assume without loss of generality that the original instance $I$ is scaled so that $\text{MMS}_{i}^{n} = 1$ for all agents $i$ in $N$. We observe that this scaling implies that $v_{i}(M) \geq n$ for all agents. 

Let $(M_{1}, \ldots M_{n})$ be an $\text{MMS}_{i}^{n}$ partition for some agent, say agent $i$, for which $v_{i}(M) > n$. For each bundle $M_{j}$ which has value greater than $1$, we arbitrarily reduce the value of goods in $M_{j}$ until $v_{i}'(M_{j}) = 1$.\footnote{It does not matter which goods of $M_{j}$ are reduced in value.} After doing this for each bundle, $v'_{i}(M) = n$ for agent $i$. Observe that this procedure does not change the MMS of any agent because once the valuations are changed all bundles have value $1$ so $\text{MMS}_{i}^{'n} = 1$ both before and after the valuations are reduced.

Let $A = (A_{1}, \ldots, A_{n})$ be an allocation on instance $I'$ that satisfies $\MMS{\alpha}{\beta}$. When we add the removed value back to each good to return to the original instance $I$, the total value of any bundle does not decrease; i.e. for each $S \subseteq M$, $v_{i}(S) \geq v'_{i}(S)$. Thus for each agent $v_{i}(A_{i}) \geq v'_{i}(A_{i})$. This implies that $A$ also satisfies $\MMS{\alpha}{\beta}$ on $I$.

We lastly observe that we may order $I'$ to make an ordered instance $I''$ with $\text{MMS}_{i}^{n} = 1$ and $v_{i}(M) = n$ for all $i \in N$. By \cref{lem:order}, any $\MMS{\alpha}{\beta}$ allocation on $I''$ implies an $\MMS{\alpha}{\beta}$ allocation on $I'$ which in turn implies an $\MMS{\alpha}{\beta}$ allocation on $I$.
\end{proof}

\cref{alg:twothird_one} details the steps of our algorithm for guaranteeing $\MMS{\frac{2}{3}}{1}$ allocations,

\subfile{algs/lone_divider_two_thirds}

\begin{figure*}
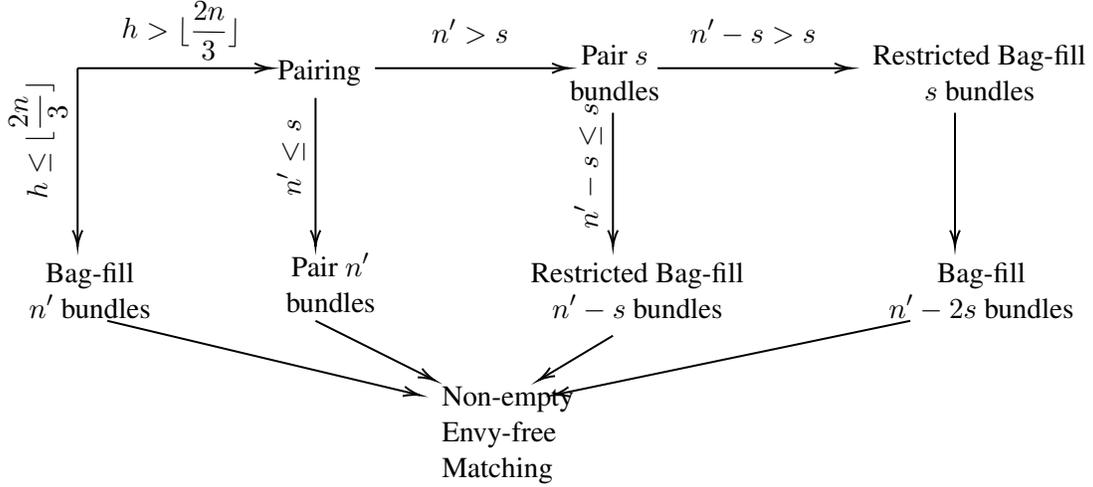

    \centering
    \subfile{images/lone_divider_decision_tree}
    \caption{The steps of our algorithm for dividing the remaining goods into bundles at each iteration by the lone divider in the $\MMS{\frac{2}{3}}{1}$ existence proof.}
    \label{fig:lone_divider_decision_tree}
\end{figure*}

\thmtwothirdsexistence*
\begin{proof}
By combining \cref{lem:order} and \cref{lem:strong_normalization}, we assume that the given instance is ordered and strongly normalized. Thus for all agents $\text{MMS}_{i}^{n} = 1$ and $v_{i}(g_{1}) < 1$.

We begin by selecting a subset $N' \subseteq N$ with $n' = |N'| = \Targ$. Our approach will be through the lone divider technique. At each iteration, we select one remaining agent to partition the remaining goods $M'$ into $n'$ bundles $(D_{1}, \ldots, D_{n'})$ each with value at least $1$. We create a bipartite graph with edges connecting the remaining agents $N'$ with the bundles that they value at least $1$. There are at least $n'$ acceptable bundles for the agent who formed the partition, so $|\mathcal{N}(N')| \geq |N'|$. Thus by \cref{lem:EFMatching}, a non-empty envy-free matching exists. For every edge in the matching, we allocate the bundle to the corresponding agent and then iterate this procedure on the remaining agents and goods. In order to guarantee that this process continues, we require that each bundle contains exactly one good from $\{g_{1}, \ldots, g_{\Targ}\}$.

Consider an arbitrary iteration of this procedure and an arbitrary agent, say agent $i$, who is selected to divide the items. We say that a good $g$ is high-value for $i$ if $v_{i}(g) > \frac{1}{2}$ and low-value otherwise. Let $h$ be the number of goods that are high-value for agent $i$. Let $A = (A_{1}, \ldots, A_{n})$ be an $\text{MMS}_{i}^{n}$ partition for agent $i$. By the strong normalization assumption, $\text{MMS}_{i}^{n} = 1$ and $v_{i}(M) = n$, meaning that the value of each bundle of $A$ is exactly $1$. Thus no bundle of $A$ contains more than one high-value good. Suppose without loss of generality that the bundles of $A$ are labeled so that each high-value goods $g_{j}$ is in the respective bundle $A_{j}$ for $j \in [h]$. We further define the remainders of these bundles to be $R_{j} \coloneqq A_{j} \setminus \{g_{j}\}$. These remainder sets have value $r_{j} = v_{i}(R_{j}) = 1 - v_{i}(g_{j})$.

Suppose that $n' = \Targ - k$ which implies that $k$ bundles $(B_{c})_{c=1}^{k}$ have already been allocated to other agents in previous rounds of the lone divider algorithm. If agent $i$ valued any such bundle $B_{c}$ at least $1$, then agent $i$ would have been envious in the matching where $B_{c}$ was allocated. Thus $v_{i}(B_{c}) < 1$ for each $c \in [k]$. Furthermore, since we maintain the constraint that each bundle in every partition for the lone divider has exactly one good from $\{g_{1}, \ldots, g_{\Targ}\}$, there are $\Targ - k = n'$ of these items remaining.

We now consider two cases depending on the number of high-value goods.

\paragraph{Case 1: $h \leq \Targ$.}
The total remaining value is at least $n - k$. We let agent $i$ form bundles using bag-filling with each bag initialized with one of the remaining $n'$ goods from $\{g_{1}, \ldots, g_{\Targ}\}$. The bags are filled with low-value goods until they have value in the range $[1, \frac{3}{2})$. This process is continued until the value from remaining goods is less than $1$. Suppose that $t$ bundles are filled this way. Then the total value of all goods is less than $\frac{3}{2} t + 1 > n - k$. Thus $t \geq \frac{2}{3}n  - \frac{2}{3}k - \frac{2}{3}$. Since $t,n,k$ are all integers, $t \geq \Targ - k = n'$.

\paragraph{Case 2: $h > \Targ$.}
When $h > \Targ$, we must be more careful about the remaining high-value goods. We define the number of surplus high-value goods as $s = h - \Targ$. In any partition of the goods into $\Targ$ bundles (counting the $k$ bundles $(B_{c})_{c=1}^{k}$), there will necessarily be at least $s$ high-value goods in a bundle which contains another high-value good. Thus our first step is to pair high-value goods.

\paragraph{Step 1: Pairs.}
Define $H^{+}$ to be the $\Targ - k$ highest-valued remaining goods (those contained in $\{g_{1}, \ldots, g_{\Targ}\}$). Likewise, define $H^{-}$ to be the other $h - \Targ$ high-value goods (those not contained in $\{g_{1}, \ldots, g_{\Targ}\}$). 

Agent $i$ forms $\min(n', s)$ bundles by pairing the lowest-value goods from $H^{+}$ with the goods from $H^{-}$. If $n' \leq s$ then agent $i$ will form $n'$ bundles by pairing, so we are done.

Otherwise, $s < n'$, and agent $i$ still needs to construct $n' - s$ bundles. There are still a total of $(h - k) - 2s = ((s + \Targ) - k) - 2s = n' - s$ high-value goods remaining. These goods all fall within $\{g_{1}, \ldots, g_{\Targ}\}$.

In \cref{fig:lone_divider_ex}, the high-value goods $g_{2}$ and $g_{7}$ were allocated in the previous round, so there remains $n' = 6$ agents. The high-value goods $g_{6}$, $g_{8}$, $g_{9}$, and $g_{10}$ are first paired into two bundles $\{g_{6}, g_{10}\}$ and $\{g_{8}, g_{9}\}$.

\begin{figure*}
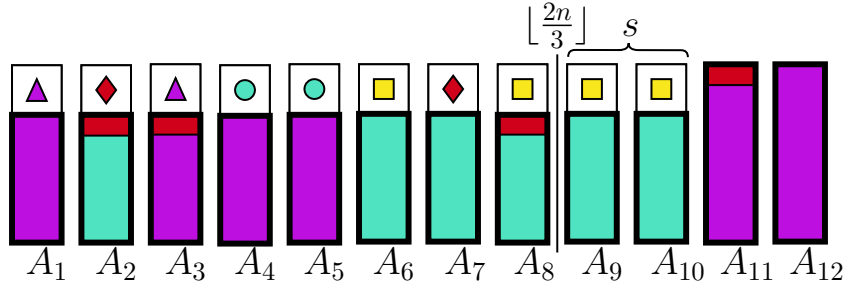

    \centering
    \subfile{images/lone_divider_partition}
    \caption{
    An example MMS partition for the current divider. Sections in red illustrate the goods that were allocated in previous iterations. There are two such goods, so $n' = \Targ - 2 = 6$. The first step is pairing the lowest $2s$ remaining high-value goods into $s$ bundles: $\{g_{6}, g_{10}\}$ and $\{g_{8}, g_{9}\}$. The second phase is restricted bag-filling. Here, the only low-value goods considered are the remainder sets corresponding to high-value goods which were already used: $\{R'_{2}, R'_{6}, R'_{7}, R'_{8}, R'_{9}, R'_{10}\}$ shown in teal. The solid boxes illustrate the remainder sets $R_{j}$ of which the teal sections are $R'_{j}$. All remaining goods, shown in purple, are then available for simple bag-filling to fill the remaining $2$ bundles.}
   \label{fig:lone_divider_ex}
\end{figure*}

\paragraph{Step 2: Restricted Bag-filling.}
The next step is to construct $\min(n' - s, s)$ bundles from bag-filling using only a subset of the available low-value goods. If $n' - s$ is smaller then there are sufficiently many bundles formed for the envy-free matching and we are done. Otherwise, we now have $2s$ acceptable bundles leaving $n' - 2s$ high-value goods from $\{g_{1}, \ldots, g_{\Targ}\}$.

For the restricted bag-filling we consider only the low-value goods which fall in bundles of $A$ whose high-value goods have already been accounted for. Thus we define two specific types of \textit{remainder sets} as follows:

\begin{itemize}
    \item The $k$ remainder sets $(R_{c}^{U})_{c = 1}^{k}$ corresponding to the high-value goods $(g_{c}^{U})_{c = 1}^{k}$ previously allocated in the bundles $(B_{c})_{c = 1}^{k}$. We likewise define $r_{c}^{U} = v_{i}(R_{c}^{U})$ for $c \in [k]$. In \cref{fig:lone_divider_example}, the previously allocated high-value goods are $g_{1}^{U} = g_{2}$ and $g_{2}^{U} = g_{7}$ which were allocated in bundles $B_{1}$ and $B_{2}$. The corresponding remainder sets are $R_{1}^{U} = R_{2}$ and $R_{2}^{U} = R_{7}$.
    \item The $2s$ remainder sets $(R_{d}^{P})_{d = 1}^{2s}$ of the $s$ pairs constructed in step 1. We likewise define $r_{d}^{P} = v_{i}(R_{d}^{P})$ for $d \in [2s]$. In \cref{fig:lone_divider_ex}, there are four such bundles $R_{1}^{P} = R_{6}$, $R_{2}^{P} = R_{8}$, $R_{3}^{P} = R_{9}$, and $R_{4}^{P} = R_{10}$.
\end{itemize}

By definition, the total value of these remainder sets is: 
\[
v_{i}\left(\bigcup_{c=1}^{k} R_{c}^{U} \cup \bigcup_{d = 1}^{2s} R_{d}^{P}\right) = \sum_{c=1}^{k} r_{c}^{U} + \sum_{d = 1}^{2s} r_{d}^{P}.
\]

We further define $R'_{j} \coloneqq R_{j} \setminus (\bigcup_{c=1}^{k} B_{c})$; that is, $R'_{j}$ is the subset of $R_{j}$ that remains after removing any goods that were previously allocated in the $k$ unacceptable bundles $(B_{c})_{c=1}^{k}$. Since $v_{i}(B_{u}) < 1$ for each $c \in [k]$, the total value of low-value goods of each $B_{c}$ is less than $r_{c}^{U}$ (because $r_{c}^{U} + v_{i}(g_{c}^{U}) = v_{i}(R_{c}^{U} \cup \{g_{c}^{U}\}) = v_{i}(A_{c}^{U}) = 1 > v_{i}(B_{c})$ and $g_{c}^{U} \in B_{c}$). Thus the total value of low-value goods allocated in the $k$ unacceptable bundles is at most $\sum_{c=1}^{k} r_{c}^{U}$. Furthermore, the total remaining value of the remainder sets $(R_{c}^{U})_{c = 1}^{k}$ and $(R_{d}^{U})_{d = 1}^{2s}$ is more than $\sum\limits_{d = 1}^{2s} r_{d}^{P}$. That is, 
\[
v_{i}\left(\bigcup_{c=1}^{k} R_{c}^{'U} \cup \bigcup_{d = 1}^{2s} R_{d}^{'P}\right) > \sum_{d = 1}^{2s} r_{d}^{P}.
\]

We then run a restricted bag-filling where we initialize each bundle with a single high-value good $g_{j}$ with $j \in [\Targ - s]$ (when this step starts there are $n' - s$ such goods). We fill bags by adding the remainder sets $R_{c}^{'U}$ and $R_{d}^{'P}$ arbitrarily until the bag has a value of $1$.\footnote{Observe that $r_{c}^{U}$ and $r_{d}^{P}$ are both less than $\frac{1}{2}$, so by ``glueing'' these low-value goods together, we can treat the entire remainder set as a single low-value good.} 

The bag-filling stops when either there are no more high-value goods or when there are no more remainder sets to fill another bundle. In the former case, $n' - s$ bundles are filled yielding a total of $n'$ bundles when we include the bundles formed in Step 1. We now consider the latter case.

Let $t$ be the number of bundles completed during restricted bag-filling. For each $a \in [t]$, the $a$-th bag contains one high-value good, call it $g_{a}^{H}$ for $g_{a}^{H} \in \{g_{1}, \ldots, g_{\Targ - s}\}$, and some remainder sets $R'_{j}$. Since the bundle is worth at least $1$, the total value of the remainder sets in the bag is at least $r_{a}^{H} = 1 - v_{i}(g_{a}^{H})$. Let $R_{a}^{'L}$ be the last remainder set added to bag $a$ with value $r_{a}^{'L}$. Then the total value of the remainder sets in bag $a$ is at most $r_{a}^{H} + r_{a}^{'L} \leq r_{a}^{H} + r_{a}^{L}$. Furthermore, since the $(t+1)$-th bag could not be filled, the total value of the remainder sets in this last bag is less than $r_{t+1}^{H}$. Therefore the total value of remainder sets satisfies
\begin{align*}
    v_{i}\left(\bigcup_{c = 1}^{k} R_{c}^{'U} \cup \bigcup_{d}^{2s} R_{d}^{'P} \right) &< r_{t+1}^{H} + \sum_{a = 1}^{t}(r_{a}^{H} + r_{a}^{L}) \\
    &= \left(\sum_{a = 1}^{t+1} r_{a}^{H}\right) + \left(\sum_{a = 1}^{t} r_{a}^{L}\right)
\end{align*}

Combining this with the lower bound of the remainder sets yields
\[
\left(\sum_{a = 1}^{t+1} r_{a}^{H}\right) + \left(\sum_{a = 1}^{t} r_{a}^{L}\right) > \sum_{d = 1}^{2s} r_{d}^{P}.
\]
In the left-hand side there are $2t + 1$ terms, while on the right-hand side there are $2s$ terms. We now show that each term in the left-hand side is less than or equal to a unique term in the right-hand side. Since the left-hand side is larger than the right-hand side, this implies that the the left-hand side must have more terms than the right-hand side. Thus $t \geq s$ and agent $i$ successfully constructs $s$ bundles during the restricted bag-filling.

\begin{itemize}
    \item Consider first the terms $r_{a}^{H}$ and compare them to the terms $r_{d}^{P}$ for $d \in [s]$ - the remainders of the $s$ high-value goods of $H^{+}$ used in pairing. Since the pairs in Step 1 were constructed from the lowest valued goods in $\{g_{1}, \ldots, g_{\Targ}\}$, every remainder $r_{d}^{P}$ is larger than any remainder $r_{a}^{H}$. Therefore, every term $r_{a}^{H}$ is weakly smaller than the corresponding term $r_{d}^{P}$.
    \item Consider now the terms $r_{a}^{L}$ and compare them to $r_{d}^{P}$ for $d \in \{s+1, \ldots, 2s\}$ - the remainders of the $s$ high-value goods of $H^{-}$ used in pairing. Each of the $r_{a}^{L}$ is unique, so it is either equal to some remainder $r_{d}^{P}$ or to some remainder $r_{c}^{U}$. All remainders $r_{c}^{U}$ correspond to high-value goods $g_{j}$ for $j \in [\Targ]$ so they are smaller than the remainders $r_{d}^{P}$ for $d \in \{s + 1, \ldots, 2s\}$. Therefore, every $r_{a}^{L}$ is weakly smaller than some unique term $r_{d}^{P}$.
\end{itemize}

\paragraph{Step 3: Plain Bag-filling.}
Up to this point, agent $i$ has used a total value of at most $3s + k/2$:
\begin{itemize}
    \item The $2s$ high-value goods used in Step 1 plus their remainders $r_{d}^{P}$.
    \item An additional $s$ high-value goods used in Step 2.
    \item The remainders $r_{c}^{U}$ for $c \in [k]$ used in Step 2, each of which has a value of less than $\frac{1}{2}$.
\end{itemize}

Taking account of each of the $k$ unacceptable bundles which have value less than $1$, the total remaining value is more than 
\begin{align*}
n - (3s + k/2) - k \geq \frac{3}{2}\Targ - 3s - \frac{3}{2}k \\
\frac{3}{2}(\Targ - 2s - k) = \frac{3}{2}(n' - 2s).
\end{align*}
Furthermore, there are exactly $n' - 2s$ high-value goods remaining, and the value of all other remaining goods is at most $\frac{1}{2}$. Therefore, by plain bag-filling as in Case 1, agent $i$ can fill the required $n' - 2s$ bags.
\end{proof}

\subsection{Proof of \cref{thm:two_thirds_poly}}

\subfile{algs/lone_divider_poly}

\thmtwothirdsproof*
\begin{proof}
Our algorithm follows the intuition of the proof of \cref{thm:two_thirds_existence}. The main difference is that we relax strong normalization (\cref{lem:strong_normalization}) to normalization (\cref{lem:normalize}), and we skip the restricted bag-filling step. We show that the resulting algorithm, seen in \cref{alg:lone_divider_poly}, computes an $\MMS{\frac{2}{3}}{1}$ allocation when $n < 9$.

Using \cref{lem:normalize}, we assume that the instance is normalized; that is, the instance is ordered, scaled so that $v_{i}(M) = n$, and reduced $v_{i}(g_{1}) < 1$ and $v_{i}(\{g_{n}, g_{n+1}\}) < 1$. As in the proof of \cref{thm:two_thirds_existence}, we use the lone divider technique. We begin with $N'$ as an arbitrary subset of $\Targ$ agents. We select one agent to divide items into $n' = |N'|$ bundles. We create a bipartite graph with agents of $N'$ on one side and the bundles $(D_{1}, \ldots, D_{n'})$ on the other. Since there are $n'$ bundles which are each acceptable to the divider, an envy-free matching exists by \cref{lem:EFMatching}. We find such a matching and allocate the matched bundles to the corresponding agents. Then we remove the agents and bundles from the market and repeat this process until no agents remain. As in \cref{thm:two_thirds_existence}, we require that each allocated bundle contains exactly one good from $\{g_{1}, \ldots, g_{\Targ}\}$.

As in \cref{thm:two_thirds_existence}, the key challenge is to show that at each iteration, the divider is able to form $n'$ bundles. Consider an arbitrary iteration with an arbitrary divider $i$. Suppose that $k$ bundles $(B_{c})_{c=1}^{k}$ have previously been allocated. Since agent $i$ was not envious of the agents who received these bundles, $v_{i}(B_{c}) < 1$. We define a good $g$ to be high-value if $v_{i}(g) \geq \frac{1}{2}$. Since $v_{i}(\{g_{n}, g_{n+1}\}) < 1$, $v_{i}(g_{n+1}) < \frac{1}{2}$ implying that there are at most $n$ high-value goods. 

We consider two cases dependent upon the number of high-value goods.

\paragraph{Case 1 $h \leq \Targ$.}
The total value remaining is at least $n - k$. In this case, agent $i$ forms bundles by bag-filling with each bag initialized with one of the remaining $n'$ goods from $\{g_{1}, \ldots, g_{\Targ}\}$. Each bag is filled using low-value goods in order of decreasing value until the value of the bundle falls in the range $[1, \frac{3}{2})$. This process is continued until the remaining goods have value less than $1$. Suppose that $t$ bundles are filled this way. Then the total value of all goods is less than $\frac{3}{2}t + 1 > n - k$. Thus $t \geq \frac{2}{3}n - \frac{2}{3}k - \frac{2}{3}$. Since $t, n, k$ are all integers, $t \geq \Targ - k = n'$.

\paragraph{Case 2 $h > \Targ$.}
In this case, define $s = h - \Targ$ to be the number of surplus high-value goods. Agent $i$ begins by pairing $\min(n', s)$ bundles each containing one high-value good from $\{g_{1}, \ldots, g_{\Targ}\}$ and one high-value good from $\{g_{\Targ + 1}, \ldots, g_{h}\}$. If $n' \leq s$, then there are sufficiently many bundles for the envy-free matching and we are done. If $s < n'$, then we let agent $i$ form the remaining bundles using the simple bag-filling of Case 1.

We now show that when $n < 9$, this bag-filling will form at least $n' - s$ bundles. Suppose that $t$ bundles were formed during this bag-filling.

Consider the following cases which depend on $h$, $s$, and $t$.
    
\begin{enumerate}
    \item Suppose $t \geq n' - s$. Then since $\Targ - s \geq n' - s$, we have sufficiently many bundles continue to the envy-free matching.
    
    \item Suppose $t \leq n' - 2s$. Since the instance is ordered, the $(h - t)$ high-value goods not allocated during bag-filling all have value at most $v_{i}(g_{t+1})$. Thus the total remaining low-value goods have value at least $n - k - \frac{3}{2}t - (n' + s - t)v_{i}(g_{t+1})$.
    
    \begin{align*}
        n - k - \frac{3}{2}t - (n' + s - t)v_{i}(g_{t+1}) &\geq \frac{3}{2}\Targ - k - \frac{3}{2}t - (n' + s - t)v_{i}(g_{t+1}) \\
        &= \frac{1}{2}\Targ + (\Targ - k) - \frac{3}{2} t - (n' + s - t)v_{i}(g_{t+1}) \\
        &= \frac{1}{2}\Targ + n' + s - s - \frac{3}{2} t - (n' + s - t)v_{i}(g_{t+1}) \\
        &= \frac{1}{2}\Targ - \frac{1}{2}t - s + (n' + s - t) - (n' + s - t)v_{i}(g_{t+1}) \\
        &= \frac{1}{2}(\Targ - 2s - t) + (n' + s - t)(1 - v_{i}(g_{t+1}))
    \end{align*}
    
    Since $t \leq n' - 2s$, $\Targ - 2s - t \geq 0$ and $n' + s - t > 0$ implying that the total remaining value from low-value goods is at least $1 - v_{i}(g_{t+1})$ and another bundle could have been filled. This contradicts that bag filling stopped with $t \leq n' - 2s$.
    
    \item The remaining cases arise when $n' - 2s < t < n' - s$. We enumerate each case to show that they do not occur when $n < 9$. We consider cases as a tuple of the form $(n, n', s, t)$ for clarity. Since $n < 9$, and $n' - 2s < t < n' - s$, there are $12$ cases, seen in \cref{tab:lone_divider_cases}
    
    The cases $(8, 5, 3, 0)$, $(8, 4, 3, 0)$, $(8, 3, 2, 0)$, $(7, 4, 3, 0)$, $(7, 3, 2, 0)$, $(6, 3, 2, 0)$, and $(5, 3, 2, 0)$ assume that no bundles are allocated during bag-filling. However, combining the facts that there is $n - k$ value remaining, $v_{i}(g) \leq v_{i}(g_{1}) < 1$, there are at most $n - k$ high-value goods remaining, and $v_{i}(M) = n$, there must be at least $(n - k) - (n - k)(v_{i}(g_{1})) \geq (n - k)(1 - v_{i}(g_{1})) \geq 1 - v_{i}(g_{1})$ value from low-value goods. Thus at least one bundle is allocated during bag-filling and these cases do not occur.
    
    The cases $(6, 4, 2, 1)$ and $(8, 5, 3, 1)$ only occur if agent $i$ believes there are $n$ high-value goods. Notice also that $n' = \Targ$ in both cases so no bundles were previously allocated. Since $v_{i}(\{g_{n}, g_{n+1}\}) < 1$ all low-value goods are worth less than $1 - v_{i}(g_{n})$. If $v_{i}(\{g_{1}, g_{n+1}\}) < 1$ then the first bundle from bag-filling was $\{g_{1}, g_{n+1}\}$. In this case, the amount of value remaining from low-value goods is at least:
    \begin{align*}
        &n - (v_{i}(g_{1}) + v_{i}(g_{n+1})) - \sum_{k = 2}^{n} v_{i}(g_{k}) \\
        &= n - (v_{i}(g_{n}) + v_{i}(g_{n+1})) - \sum_{k = 1}^{n-1} v_{i}(g_{k}) \\
        &> n - 1 - \sum_{k = 1}^{n-1} v_{i}(g_{k}) \\
        &\geq (1 - v_{i}(g_{2})) + (n - 2) - \sum_{k = 1, k \neq 2}^{n-1} v_{i}(g_{k}) \\
        &> (1 - v_{i}(g_{2})) + (n - 2) - (n-2) \\
        &> 1 - v_{i}(g_{2})
    \end{align*}
    
    If $\{g_{1}, g_{n+1}\}$ is not worth $1$ agent $i$, then because low-value goods are added in descending order, the maximum value from low-value goods allocated in the first bundle is $2(1 - v_{i}(g_{1}))$. Thus the value remaining from low-value goods is at least:
    \begin{align*}
        &n - \sum_{k = 1}^{n} v_{i}(g_{k}) - 2(1 - v_{i}(g_{1})) \\
        &\geq (n - 2) - \sum_{k = 2}^{n-1} v_{i}(g_{k}) + (1 - v_{i}(g_{1})) + (1 - v_{i}(g_{n}) \\
        &\qquad - 2(1 - v_{i}(g_{1}))  \\
        &\geq (n - 2) - \sum_{k = 2}^{n-1} v_{i}(g_{k}) \\
        &= (1 - v_{i}(g_{2})) + (n - 3) - \sum_{k = 3}^{n-1} v_{i}(g_{k}) \\
        &\geq (1 - v_{i}(g_{2})) + (n - 3) - (n - 3) \\
        &= 1 - v_{i}(g_{2})
    \end{align*}
    In either case, there is at least enough value from low-value goods to fill another bundle during bag-filling contradicting that bag-filling stopped. Thus the cases $(6, 4, 2, 1)$ and $(8, 5, 3, 1)$ do not occur.
   
    The next cases we address are $(7, 4, 2, 1)$ and $(8, 5, 2, 2)$. In these cases, $n' = \Targ$, and there are $n-1$ high-value goods. We observe that the amount of value from low-value goods in these cases is at least:
    \begin{align*}
        &n - \sum_{k = 1}^{n-1} v_{i}(g_{k}) \\
        &\geq 1 + (n - 1) - v_{i}(g_{1}) - v_{i}(g_{2}) - v_{i}(g_{3}) - \sum_{k = 4}^{n-1} v_{i}(g_{k}) \\
        &= 1 + (1 - v_{i}(g_{1})) + (1 - v_{i}(g_{2})) + (1 - v_{i}(g_{3})) \\
        &\qquad + (n - 4) - \sum_{k = 4}^{n-1} v_{i}(g_{k}) \\
        &> 1 + (1 - v_{i}(g_{1})) + (1 - v_{i}(g_{2})) + (1 - v_{i}(g_{3})) \\
        &\qquad + (n - 4) - (n-4) \\
        &> (\frac{3}{2} - v_{i}(g_{1})) + (\frac{3}{2} - v_{i}(g_{2})) + (1 - v_{i}(g_{3}))
    \end{align*}
    Since the first two bundles have total value at most $\frac{3}{2}$, there is at least enough value to fill the first two bundles and still have enough low-value goods to fill a third bundle. This contradicts that only one or two bundles were filled during bag-filling (based on the cases). Thus neither cases $(7, 4, 2, 1)$ and $(8, 5, 2, 2)$ occur.
    
    The last case is $(8, 4, 2, 1)$. In this case, one bundle has previously been allocated which leaves at least $n - k = 7$ value. Accounting for the $4$ high-value goods used in pairing and the $\frac{3}{2}$ value allocated in the first bundle from bag-filling, there is at least $7 - 4 - \frac{3}{2} = \frac{3}{2}$ value remaining from only low-value goods. Thus a second bundle is filled during bag-filling and the case $(8, 4, 2, 1)$ does not occur.
    
    \begin{table}[h!]
        \centering
        \begin{tabular}{c|c|c|c}
            $n$ & $n'$ & $s$ & t \\\hline
            8 & 5 & 2 & 2 \\
            8 & 5 & 3 & 1 \\
            8 & 5 & 3 & 0 \\
            8 & 4 & 2 & 1 \\
            8 & 4 & 3 & 0 \\
            8 & 3 & 2 & 0 \\
            7 & 4 & 2 & 1 \\
            7 & 4 & 3 & 0 \\
            7 & 3 & 2 & 0 \\
            6 & 4 & 2 & 1 \\
            6 & 3 & 2 & 0 \\
            5 & 3 & 2 & 0
        \end{tabular}
        \caption{The remaining cases where $n' - 2s < t < n' - s$.}
        \label{tab:lone_divider_cases}
    \end{table}
\end{enumerate}
\end{proof}



\section{Additional Material for Section~\ref{sec:experiments}}

\Cref{fig:numbers_reduced} shows the number of agents that were removed during the reduction phase. It is important to observe that with uniformly random valuations, the more goods there are, the closer each good's value is to $\frac{n}{m}$ after valuations are scaled so that $v_{i}(M) = n$. 

This observation likely explains the strong banding apparent in \cref{fig:numbers_reduced}. When $m < 2n$, the pigeonhole principle guarantees that there is a bundle of any MMS partition which contains exactly one good. Using a simple swapping argument, we see that $v_{i}(g_{1}) \geq \text{MMS}_{i}^{n}$ when $m < 2n$. While we do not explicitly implement this reduction in our experiment, the band close to $m = 2n$ illustrates that random valuations will still exhibit this behavior.

\begin{figure}[h!]
    \centering
    \includegraphics[width=0.8\textwidth]{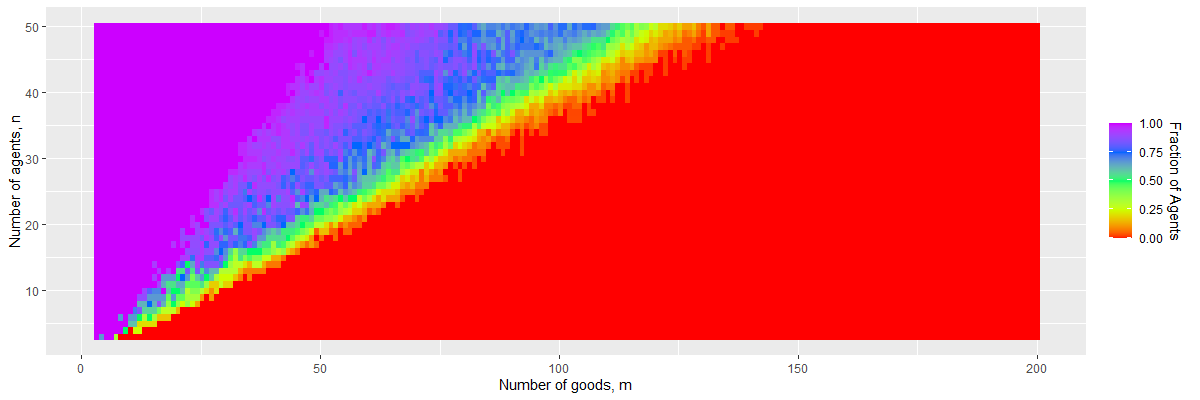}
    \caption{The fraction of agents removed during the reduction phase of normalization in our experiments.}
    \label{fig:numbers_reduced}
\end{figure}

\end{document}

%% file: images/S_Shape_fig.tex
\tikzset{every picture/.style={line width=0.75pt}} 
    \begin{tikzpicture}[x=0.75pt,y=0.75pt,yscale=-.4,xscale=.4]

\draw  [fill={rgb, 255:red, 240; green, 240; blue, 40 }  ,fill opacity=1 ] (86.44,240.98) -- (102.75,224.68) -- (126.93,224.68) -- (126.93,248.86) -- (110.62,265.17) -- (86.44,265.17) -- cycle ; \draw   (126.93,224.68) -- (110.62,240.98) -- (86.44,240.98) ; \draw   (110.62,240.98) -- (110.62,265.17) ;
\draw  [fill={rgb, 255:red, 240; green, 200; blue, 40 }  ,fill opacity=1 ] (110.62,240.98) -- (126.93,224.68) -- (151.12,224.68) -- (151.12,248.86) -- (134.81,265.17) -- (110.62,265.17) -- cycle ; \draw   (151.12,224.68) -- (134.81,240.98) -- (110.62,240.98) ; \draw   (134.81,240.98) -- (134.81,265.17) ;
\draw  [fill={rgb, 255:red, 240; green, 160; blue, 40 }  ,fill opacity=1 ] (134.81,240.98) -- (151.12,224.68) -- (175.3,224.68) -- (175.3,248.86) -- (158.99,265.17) -- (134.81,265.17) -- cycle ; \draw   (175.3,224.68) -- (158.99,240.98) -- (134.81,240.98) ; \draw   (158.99,240.98) -- (158.99,265.17) ;
\draw  [fill={rgb, 255:red, 240; green, 120; blue, 40 }  ,fill opacity=1 ] (158.99,240.98) -- (175.3,224.68) -- (199.48,224.68) -- (199.48,248.86) -- (183.17,265.17) -- (158.99,265.17) -- cycle ; \draw   (199.48,224.68) -- (183.17,240.98) -- (158.99,240.98) ; \draw   (183.17,240.98) -- (183.17,265.17) ;
\draw  [fill={rgb, 255:red, 240; green, 80; blue, 40 }  ,fill opacity=1 ] (183.17,240.98) -- (199.48,224.68) -- (223.67,224.68) -- (223.67,248.86) -- (207.36,265.17) -- (183.17,265.17) -- cycle ; \draw   (223.67,224.68) -- (207.36,240.98) -- (183.17,240.98) ; \draw   (207.36,240.98) -- (207.36,265.17) ;
\draw  [fill={rgb, 255:red, 40; green, 40; blue, 40 }  ,fill opacity=1 ] (288.91,159.44) -- (305.22,143.13) -- (329.4,143.13) -- (329.4,167.31) -- (313.09,183.62) -- (288.91,183.62) -- cycle ; \draw   (329.4,143.13) -- (313.09,159.44) -- (288.91,159.44) ; \draw   (313.09,159.44) -- (313.09,183.62) ;
\draw  [fill={rgb, 255:red, 80; green, 40; blue, 40 }  ,fill opacity=1 ] (272.6,175.75) -- (288.91,159.44) -- (313.09,159.44) -- (313.09,183.62) -- (296.78,199.93) -- (272.6,199.93) -- cycle ; \draw   (313.09,159.44) -- (296.78,175.75) -- (272.6,175.75) ; \draw   (296.78,175.75) -- (296.78,199.93) ;
\draw  [fill={rgb, 255:red, 120; green, 40; blue, 40 }  ,fill opacity=1 ] (256.29,192.06) -- (272.6,175.75) -- (296.78,175.75) -- (296.78,199.93) -- (280.47,216.24) -- (256.29,216.24) -- cycle ; \draw   (296.78,175.75) -- (280.47,192.06) -- (256.29,192.06) ; \draw   (280.47,192.06) -- (280.47,216.24) ;
\draw  [fill={rgb, 255:red, 160; green, 40; blue, 40 }  ,fill opacity=1 ] (239.98,208.37) -- (256.29,192.06) -- (280.47,192.06) -- (280.47,216.24) -- (264.16,232.55) -- (239.98,232.55) -- cycle ; \draw   (280.47,192.06) -- (264.16,208.37) -- (239.98,208.37) ; \draw   (264.16,208.37) -- (264.16,232.55) ;
\draw  [fill={rgb, 255:red, 200; green, 40; blue, 40 }  ,fill opacity=1 ] (223.67,224.68) -- (239.98,208.37) -- (264.16,208.37) -- (264.16,232.55) -- (247.85,248.86) -- (223.67,248.86) -- cycle ; \draw   (264.16,208.37) -- (247.85,224.68) -- (223.67,224.68) ; \draw   (247.85,224.68) -- (247.85,248.86) ;
\draw  [fill={rgb, 255:red, 240; green, 40; blue, 40 }  ,fill opacity=1 ] (207.36,240.98) -- (223.67,224.68) -- (247.85,224.68) -- (247.85,248.86) -- (231.54,265.17) -- (207.36,265.17) -- cycle ; \draw   (247.85,224.68) -- (231.54,240.98) -- (207.36,240.98) ; \draw   (231.54,240.98) -- (231.54,265.17) ;
\draw  [fill={rgb, 255:red, 40; green, 40; blue, 80 }  ,fill opacity=1 ] (288.91,135.25) -- (305.22,118.94) -- (329.4,118.94) -- (329.4,143.13) -- (313.09,159.44) -- (288.91,159.44) -- cycle ; \draw   (329.4,118.94) -- (313.09,135.25) -- (288.91,135.25) ; \draw   (313.09,135.25) -- (313.09,159.44) ;
\draw  [fill={rgb, 255:red, 80; green, 40; blue, 80 }  ,fill opacity=1 ] (272.6,151.56) -- (288.91,135.25) -- (313.09,135.25) -- (313.09,159.44) -- (296.78,175.75) -- (272.6,175.75) -- cycle ; \draw   (313.09,135.25) -- (296.78,151.56) -- (272.6,151.56) ; \draw   (296.78,151.56) -- (296.78,175.75) ;
\draw  [fill={rgb, 255:red, 120; green, 40; blue, 80 }  ,fill opacity=1 ] (256.29,167.87) -- (272.6,151.56) -- (296.78,151.56) -- (296.78,175.75) -- (280.47,192.06) -- (256.29,192.06) -- cycle ; \draw   (296.78,151.56) -- (280.47,167.87) -- (256.29,167.87) ; \draw   (280.47,167.87) -- (280.47,192.06) ;
\draw  [fill={rgb, 255:red, 160; green, 40; blue, 80 }  ,fill opacity=1 ] (239.98,184.18) -- (256.29,167.87) -- (280.47,167.87) -- (280.47,192.06) -- (264.16,208.37) -- (239.98,208.37) -- cycle ; \draw   (280.47,167.87) -- (264.16,184.18) -- (239.98,184.18) ; \draw   (264.16,184.18) -- (264.16,208.37) ;
\draw  [fill={rgb, 255:red, 200; green, 40; blue, 80 }  ,fill opacity=1 ] (223.67,200.49) -- (239.98,184.18) -- (264.16,184.18) -- (264.16,208.37) -- (247.85,224.68) -- (223.67,224.68) -- cycle ; \draw   (264.16,184.18) -- (247.85,200.49) -- (223.67,200.49) ; \draw   (247.85,200.49) -- (247.85,224.68) ;
\draw  [fill={rgb, 255:red, 40; green, 40; blue, 120 }  ,fill opacity=1 ] (288.91,111.07) -- (305.22,94.76) -- (329.4,94.76) -- (329.4,118.94) -- (313.09,135.25) -- (288.91,135.25) -- cycle ; \draw   (329.4,94.76) -- (313.09,111.07) -- (288.91,111.07) ; \draw   (313.09,111.07) -- (313.09,135.25) ;
\draw  [fill={rgb, 255:red, 80; green, 40; blue, 120 }  ,fill opacity=1 ] (272.6,127.38) -- (288.91,111.07) -- (313.09,111.07) -- (313.09,135.25) -- (296.78,151.56) -- (272.6,151.56) -- cycle ; \draw   (313.09,111.07) -- (296.78,127.38) -- (272.6,127.38) ; \draw   (296.78,127.38) -- (296.78,151.56) ;
\draw  [fill={rgb, 255:red, 120; green, 40; blue, 120 }  ,fill opacity=1 ] (256.29,143.69) -- (272.6,127.38) -- (296.78,127.38) -- (296.78,151.56) -- (280.47,167.87) -- (256.29,167.87) -- cycle ; \draw   (296.78,127.38) -- (280.47,143.69) -- (256.29,143.69) ; \draw   (280.47,143.69) -- (280.47,167.87) ;
\draw  [fill={rgb, 255:red, 160; green, 40; blue, 120 }  ,fill opacity=1 ] (239.98,160) -- (256.29,143.69) -- (280.47,143.69) -- (280.47,167.87) -- (264.16,184.18) -- (239.98,184.18) -- cycle ; \draw   (280.47,143.69) -- (264.16,160) -- (239.98,160) ; \draw   (264.16,160) -- (264.16,184.18) ;
\draw  [fill={rgb, 255:red, 200; green, 40; blue, 120 }  ,fill opacity=1 ] (223.67,176.31) -- (239.98,160) -- (264.16,160) -- (264.16,184.18) -- (247.85,200.49) -- (223.67,200.49) -- cycle ; \draw   (264.16,160) -- (247.85,176.31) -- (223.67,176.31) ; \draw   (247.85,176.31) -- (247.85,200.49) ;
\draw  [fill={rgb, 255:red, 240; green, 240; blue, 80 }  ,fill opacity=1 ] (86.44,216.8) -- (102.75,200.49) -- (126.93,200.49) -- (126.93,224.68) -- (110.62,240.98) -- (86.44,240.98) -- cycle ; \draw   (126.93,200.49) -- (110.62,216.8) -- (86.44,216.8) ; \draw   (110.62,216.8) -- (110.62,240.98) ;
\draw  [fill={rgb, 255:red, 240; green, 200; blue, 80 }  ,fill opacity=1 ] (110.62,216.8) -- (126.93,200.49) -- (151.12,200.49) -- (151.12,224.68) -- (134.81,240.98) -- (110.62,240.98) -- cycle ; \draw   (151.12,200.49) -- (134.81,216.8) -- (110.62,216.8) ; \draw   (134.81,216.8) -- (134.81,240.98) ;
\draw  [fill={rgb, 255:red, 240; green, 160; blue, 80 }  ,fill opacity=1 ] (134.81,216.8) -- (151.12,200.49) -- (175.3,200.49) -- (175.3,224.68) -- (158.99,240.98) -- (134.81,240.98) -- cycle ; \draw   (175.3,200.49) -- (158.99,216.8) -- (134.81,216.8) ; \draw   (158.99,216.8) -- (158.99,240.98) ;
\draw  [fill={rgb, 255:red, 240; green, 120; blue, 80 }  ,fill opacity=1 ] (158.99,216.8) -- (175.3,200.49) -- (199.48,200.49) -- (199.48,224.68) -- (183.17,240.98) -- (158.99,240.98) -- cycle ; \draw   (199.48,200.49) -- (183.17,216.8) -- (158.99,216.8) ; \draw   (183.17,216.8) -- (183.17,240.98) ;
\draw  [fill={rgb, 255:red, 240; green, 80; blue, 80 }  ,fill opacity=1 ] (183.17,216.8) -- (199.48,200.49) -- (223.67,200.49) -- (223.67,224.68) -- (207.36,240.98) -- (183.17,240.98) -- cycle ; \draw   (223.67,200.49) -- (207.36,216.8) -- (183.17,216.8) ; \draw   (207.36,216.8) -- (207.36,240.98) ;
\draw  [fill={rgb, 255:red, 240; green, 40; blue, 80 }  ,fill opacity=1 ] (207.36,216.8) -- (223.67,200.49) -- (247.85,200.49) -- (247.85,224.68) -- (231.54,240.98) -- (207.36,240.98) -- cycle ; \draw   (247.85,200.49) -- (231.54,216.8) -- (207.36,216.8) ; \draw   (231.54,216.8) -- (231.54,240.98) ;
\draw  [fill={rgb, 255:red, 240; green, 240; blue, 120 }  ,fill opacity=1 ] (86.44,192.62) -- (102.75,176.31) -- (126.93,176.31) -- (126.93,200.49) -- (110.62,216.8) -- (86.44,216.8) -- cycle ; \draw   (126.93,176.31) -- (110.62,192.62) -- (86.44,192.62) ; \draw   (110.62,192.62) -- (110.62,216.8) ;
\draw  [fill={rgb, 255:red, 240; green, 200; blue, 120 }  ,fill opacity=1 ] (110.62,192.62) -- (126.93,176.31) -- (151.12,176.31) -- (151.12,200.49) -- (134.81,216.8) -- (110.62,216.8) -- cycle ; \draw   (151.12,176.31) -- (134.81,192.62) -- (110.62,192.62) ; \draw   (134.81,192.62) -- (134.81,216.8) ;
\draw  [fill={rgb, 255:red, 240; green, 160; blue, 120 }  ,fill opacity=1 ] (134.81,192.62) -- (151.12,176.31) -- (175.3,176.31) -- (175.3,200.49) -- (158.99,216.8) -- (134.81,216.8) -- cycle ; \draw   (175.3,176.31) -- (158.99,192.62) -- (134.81,192.62) ; \draw   (158.99,192.62) -- (158.99,216.8) ;
\draw  [fill={rgb, 255:red, 240; green, 120; blue, 120 }  ,fill opacity=1 ] (158.99,192.62) -- (175.3,176.31) -- (199.48,176.31) -- (199.48,200.49) -- (183.17,216.8) -- (158.99,216.8) -- cycle ; \draw   (199.48,176.31) -- (183.17,192.62) -- (158.99,192.62) ; \draw   (183.17,192.62) -- (183.17,216.8) ;
\draw  [fill={rgb, 255:red, 240; green, 80; blue, 120 }  ,fill opacity=1 ] (183.17,192.62) -- (199.48,176.31) -- (223.67,176.31) -- (223.67,200.49) -- (207.36,216.8) -- (183.17,216.8) -- cycle ; \draw   (223.67,176.31) -- (207.36,192.62) -- (183.17,192.62) ; \draw   (207.36,192.62) -- (207.36,216.8) ;
\draw  [fill={rgb, 255:red, 240; green, 40; blue, 120 }  ,fill opacity=1 ] (207.36,192.62) -- (223.67,176.31) -- (247.85,176.31) -- (247.85,200.49) -- (231.54,216.8) -- (207.36,216.8) -- cycle ; \draw   (247.85,176.31) -- (231.54,192.62) -- (207.36,192.62) ; \draw   (231.54,192.62) -- (231.54,216.8) ;
\draw  [fill={rgb, 255:red, 40; green, 40; blue, 160 }  ,fill opacity=1 ] (288.91,86.89) -- (305.22,70.58) -- (329.4,70.58) -- (329.4,94.76) -- (313.09,111.07) -- (288.91,111.07) -- cycle ; \draw   (329.4,70.58) -- (313.09,86.89) -- (288.91,86.89) ; \draw   (313.09,86.89) -- (313.09,111.07) ;
\draw  [fill={rgb, 255:red, 80; green, 40; blue, 160 }  ,fill opacity=1 ] (272.6,103.2) -- (288.91,86.89) -- (313.09,86.89) -- (313.09,111.07) -- (296.78,127.38) -- (272.6,127.38) -- cycle ; \draw   (313.09,86.89) -- (296.78,103.2) -- (272.6,103.2) ; \draw   (296.78,103.2) -- (296.78,127.38) ;
\draw  [fill={rgb, 255:red, 120; green, 40; blue, 160 }  ,fill opacity=1 ] (256.29,119.5) -- (272.6,103.2) -- (296.78,103.2) -- (296.78,127.38) -- (280.47,143.69) -- (256.29,143.69) -- cycle ; \draw   (296.78,103.2) -- (280.47,119.5) -- (256.29,119.5) ; \draw   (280.47,119.5) -- (280.47,143.69) ;
\draw  [fill={rgb, 255:red, 160; green, 40; blue, 160 }  ,fill opacity=1 ] (239.98,135.81) -- (256.29,119.5) -- (280.47,119.5) -- (280.47,143.69) -- (264.16,160) -- (239.98,160) -- cycle ; \draw   (280.47,119.5) -- (264.16,135.81) -- (239.98,135.81) ; \draw   (264.16,135.81) -- (264.16,160) ;
\draw  [fill={rgb, 255:red, 200; green, 40; blue, 160 }  ,fill opacity=1 ] (223.67,152.12) -- (239.98,135.81) -- (264.16,135.81) -- (264.16,160) -- (247.85,176.31) -- (223.67,176.31) -- cycle ; \draw   (264.16,135.81) -- (247.85,152.12) -- (223.67,152.12) ; \draw   (247.85,152.12) -- (247.85,176.31) ;
\draw  [fill={rgb, 255:red, 240; green, 240; blue, 160 }  ,fill opacity=1 ] (86.44,168.43) -- (102.75,152.12) -- (126.93,152.12) -- (126.93,176.31) -- (110.62,192.62) -- (86.44,192.62) -- cycle ; \draw   (126.93,152.12) -- (110.62,168.43) -- (86.44,168.43) ; \draw   (110.62,168.43) -- (110.62,192.62) ;
\draw  [fill={rgb, 255:red, 240; green, 200; blue, 160 }  ,fill opacity=1 ] (110.62,168.43) -- (126.93,152.12) -- (151.12,152.12) -- (151.12,176.31) -- (134.81,192.62) -- (110.62,192.62) -- cycle ; \draw   (151.12,152.12) -- (134.81,168.43) -- (110.62,168.43) ; \draw   (134.81,168.43) -- (134.81,192.62) ;
\draw  [fill={rgb, 255:red, 240; green, 160; blue, 160 }  ,fill opacity=1 ] (134.81,168.43) -- (151.12,152.12) -- (175.3,152.12) -- (175.3,176.31) -- (158.99,192.62) -- (134.81,192.62) -- cycle ; \draw   (175.3,152.12) -- (158.99,168.43) -- (134.81,168.43) ; \draw   (158.99,168.43) -- (158.99,192.62) ;
\draw  [fill={rgb, 255:red, 240; green, 120; blue, 160 }  ,fill opacity=1 ] (158.99,168.43) -- (175.3,152.12) -- (199.48,152.12) -- (199.48,176.31) -- (183.17,192.62) -- (158.99,192.62) -- cycle ; \draw   (199.48,152.12) -- (183.17,168.43) -- (158.99,168.43) ; \draw   (183.17,168.43) -- (183.17,192.62) ;
\draw  [fill={rgb, 255:red, 240; green, 80; blue, 160 }  ,fill opacity=1 ] (183.17,168.43) -- (199.48,152.12) -- (223.67,152.12) -- (223.67,176.31) -- (207.36,192.62) -- (183.17,192.62) -- cycle ; \draw   (223.67,152.12) -- (207.36,168.43) -- (183.17,168.43) ; \draw   (207.36,168.43) -- (207.36,192.62) ;
\draw  [fill={rgb, 255:red, 240; green, 40; blue, 160 }  ,fill opacity=1 ] (207.36,168.43) -- (223.67,152.12) -- (247.85,152.12) -- (247.85,176.31) -- (231.54,192.62) -- (207.36,192.62) -- cycle ; \draw   (247.85,152.12) -- (231.54,168.43) -- (207.36,168.43) ; \draw   (231.54,168.43) -- (231.54,192.62) ;
\draw  [fill={rgb, 255:red, 40; green, 40; blue, 200 }  ,fill opacity=1 ] (288.91,62.7) -- (305.22,46.39) -- (329.4,46.39) -- (329.4,70.58) -- (313.09,86.89) -- (288.91,86.89) -- cycle ; \draw   (329.4,46.39) -- (313.09,62.7) -- (288.91,62.7) ; \draw   (313.09,62.7) -- (313.09,86.89) ;
\draw  [fill={rgb, 255:red, 80; green, 40; blue, 200 }  ,fill opacity=1 ] (272.6,79.01) -- (288.91,62.7) -- (313.09,62.7) -- (313.09,86.89) -- (296.78,103.2) -- (272.6,103.2) -- cycle ; \draw   (313.09,62.7) -- (296.78,79.01) -- (272.6,79.01) ; \draw   (296.78,79.01) -- (296.78,103.2) ;
\draw  [fill={rgb, 255:red, 120; green, 40; blue, 200 }  ,fill opacity=1 ] (256.29,95.32) -- (272.6,79.01) -- (296.78,79.01) -- (296.78,103.2) -- (280.47,119.5) -- (256.29,119.5) -- cycle ; \draw   (296.78,79.01) -- (280.47,95.32) -- (256.29,95.32) ; \draw   (280.47,95.32) -- (280.47,119.5) ;
\draw  [fill={rgb, 255:red, 160; green, 40; blue, 200 }  ,fill opacity=1 ] (239.98,111.63) -- (256.29,95.32) -- (280.47,95.32) -- (280.47,119.5) -- (264.16,135.81) -- (239.98,135.81) -- cycle ; \draw   (280.47,95.32) -- (264.16,111.63) -- (239.98,111.63) ; \draw   (264.16,111.63) -- (264.16,135.81) ;
\draw  [fill={rgb, 255:red, 200; green, 40; blue, 200 }  ,fill opacity=1 ] (223.67,127.94) -- (239.98,111.63) -- (264.16,111.63) -- (264.16,135.81) -- (247.85,152.12) -- (223.67,152.12) -- cycle ; \draw   (264.16,111.63) -- (247.85,127.94) -- (223.67,127.94) ; \draw   (247.85,127.94) -- (247.85,152.12) ;
\draw  [fill={rgb, 255:red, 240; green, 240; blue, 200 }  ,fill opacity=1 ] (86.44,144.25) -- (102.75,127.94) -- (126.93,127.94) -- (126.93,152.12) -- (110.62,168.43) -- (86.44,168.43) -- cycle ; \draw   (126.93,127.94) -- (110.62,144.25) -- (86.44,144.25) ; \draw   (110.62,144.25) -- (110.62,168.43) ;
\draw  [fill={rgb, 255:red, 240; green, 200; blue, 200 }  ,fill opacity=1 ] (110.62,144.25) -- (126.93,127.94) -- (151.12,127.94) -- (151.12,152.12) -- (134.81,168.43) -- (110.62,168.43) -- cycle ; \draw   (151.12,127.94) -- (134.81,144.25) -- (110.62,144.25) ; \draw   (134.81,144.25) -- (134.81,168.43) ;
\draw  [fill={rgb, 255:red, 240; green, 160; blue, 200 }  ,fill opacity=1 ] (134.81,144.25) -- (151.12,127.94) -- (175.3,127.94) -- (175.3,152.12) -- (158.99,168.43) -- (134.81,168.43) -- cycle ; \draw   (175.3,127.94) -- (158.99,144.25) -- (134.81,144.25) ; \draw   (158.99,144.25) -- (158.99,168.43) ;
\draw  [fill={rgb, 255:red, 240; green, 120; blue, 200 }  ,fill opacity=1 ] (158.99,144.25) -- (175.3,127.94) -- (199.48,127.94) -- (199.48,152.12) -- (183.17,168.43) -- (158.99,168.43) -- cycle ; \draw   (199.48,127.94) -- (183.17,144.25) -- (158.99,144.25) ; \draw   (183.17,144.25) -- (183.17,168.43) ;
\draw  [fill={rgb, 255:red, 240; green, 80; blue, 200 }  ,fill opacity=1 ] (183.17,144.25) -- (199.48,127.94) -- (223.67,127.94) -- (223.67,152.12) -- (207.36,168.43) -- (183.17,168.43) -- cycle ; \draw   (223.67,127.94) -- (207.36,144.25) -- (183.17,144.25) ; \draw   (207.36,144.25) -- (207.36,168.43) ;
\draw  [fill={rgb, 255:red, 240; green, 40; blue, 200 }  ,fill opacity=1 ] (207.36,144.25) -- (223.67,127.94) -- (247.85,127.94) -- (247.85,152.12) -- (231.54,168.43) -- (207.36,168.43) -- cycle ; \draw   (247.85,127.94) -- (231.54,144.25) -- (207.36,144.25) ; \draw   (231.54,144.25) -- (231.54,168.43) ;
\draw  [fill={rgb, 255:red, 40; green, 240; blue, 240 }  ,fill opacity=1 ] (167.99,38.52) -- (184.3,22.21) -- (208.48,22.21) -- (208.48,46.39) -- (192.17,62.7) -- (167.99,62.7) -- cycle ; \draw   (208.48,22.21) -- (192.17,38.52) -- (167.99,38.52) ; \draw   (192.17,38.52) -- (192.17,62.7) ;
\draw  [fill={rgb, 255:red, 40; green, 200; blue, 240 }  ,fill opacity=1 ] (192.17,38.52) -- (208.48,22.21) -- (232.67,22.21) -- (232.67,46.39) -- (216.36,62.7) -- (192.17,62.7) -- cycle ; \draw   (232.67,22.21) -- (216.36,38.52) -- (192.17,38.52) ; \draw   (216.36,38.52) -- (216.36,62.7) ;
\draw  [fill={rgb, 255:red, 40; green, 160; blue, 240 }  ,fill opacity=1 ] (216.36,38.52) -- (232.67,22.21) -- (256.85,22.21) -- (256.85,46.39) -- (240.54,62.7) -- (216.36,62.7) -- cycle ; \draw   (256.85,22.21) -- (240.54,38.52) -- (216.36,38.52) ; \draw   (240.54,38.52) -- (240.54,62.7) ;
\draw  [fill={rgb, 255:red, 40; green, 120; blue, 240 }  ,fill opacity=1 ] (240.54,38.52) -- (256.85,22.21) -- (281.03,22.21) -- (281.03,46.39) -- (264.72,62.7) -- (240.54,62.7) -- cycle ; \draw   (281.03,22.21) -- (264.72,38.52) -- (240.54,38.52) ; \draw   (264.72,38.52) -- (264.72,62.7) ;
\draw  [fill={rgb, 255:red, 40; green, 80; blue, 240 }  ,fill opacity=1 ] (264.72,38.52) -- (281.03,22.21) -- (305.22,22.21) -- (305.22,46.39) -- (288.91,62.7) -- (264.72,62.7) -- cycle ; \draw   (305.22,22.21) -- (288.91,38.52) -- (264.72,38.52) ; \draw   (288.91,38.52) -- (288.91,62.7) ;
\draw  [fill={rgb, 255:red, 40; green, 40; blue, 240 }  ,fill opacity=1 ] (288.91,38.52) -- (305.22,22.21) -- (329.4,22.21) -- (329.4,46.39) -- (313.09,62.7) -- (288.91,62.7) -- cycle ; \draw   (329.4,22.21) -- (313.09,38.52) -- (288.91,38.52) ; \draw   (313.09,38.52) -- (313.09,62.7) ;
\draw  [fill={rgb, 255:red, 80; green, 240; blue, 240 }  ,fill opacity=1 ] (151.68,54.83) -- (167.99,38.52) -- (192.17,38.52) -- (192.17,62.7) -- (175.86,79.01) -- (151.68,79.01) -- cycle ; \draw   (192.17,38.52) -- (175.86,54.83) -- (151.68,54.83) ; \draw   (175.86,54.83) -- (175.86,79.01) ;
\draw  [fill={rgb, 255:red, 80; green, 200; blue, 240 }  ,fill opacity=1 ] (175.86,54.83) -- (192.17,38.52) -- (216.36,38.52) -- (216.36,62.7) -- (200.05,79.01) -- (175.86,79.01) -- cycle ; \draw   (216.36,38.52) -- (200.05,54.83) -- (175.86,54.83) ; \draw   (200.05,54.83) -- (200.05,79.01) ;
\draw  [fill={rgb, 255:red, 80; green, 160; blue, 240 }  ,fill opacity=1 ] (200.05,54.83) -- (216.36,38.52) -- (240.54,38.52) -- (240.54,62.7) -- (224.23,79.01) -- (200.05,79.01) -- cycle ; \draw   (240.54,38.52) -- (224.23,54.83) -- (200.05,54.83) ; \draw   (224.23,54.83) -- (224.23,79.01) ;
\draw  [fill={rgb, 255:red, 80; green, 120; blue, 240 }  ,fill opacity=1 ] (224.23,54.83) -- (240.54,38.52) -- (264.72,38.52) -- (264.72,62.7) -- (248.41,79.01) -- (224.23,79.01) -- cycle ; \draw   (264.72,38.52) -- (248.41,54.83) -- (224.23,54.83) ; \draw   (248.41,54.83) -- (248.41,79.01) ;
\draw  [fill={rgb, 255:red, 80; green, 80; blue, 240 }  ,fill opacity=1 ] (248.41,54.83) -- (264.72,38.52) -- (288.91,38.52) -- (288.91,62.7) -- (272.6,79.01) -- (248.41,79.01) -- cycle ; \draw   (288.91,38.52) -- (272.6,54.83) -- (248.41,54.83) ; \draw   (272.6,54.83) -- (272.6,79.01) ;
\draw  [fill={rgb, 255:red, 80; green, 40; blue, 240 }  ,fill opacity=1 ] (272.6,54.83) -- (288.91,38.52) -- (313.09,38.52) -- (313.09,62.7) -- (296.78,79.01) -- (272.6,79.01) -- cycle ; \draw   (313.09,38.52) -- (296.78,54.83) -- (272.6,54.83) ; \draw   (296.78,54.83) -- (296.78,79.01) ;
\draw  [fill={rgb, 255:red, 120; green, 240; blue, 240 }  ,fill opacity=1 ] (135.37,71.14) -- (151.68,54.83) -- (175.86,54.83) -- (175.86,79.01) -- (159.55,95.32) -- (135.37,95.32) -- cycle ; \draw   (175.86,54.83) -- (159.55,71.14) -- (135.37,71.14) ; \draw   (159.55,71.14) -- (159.55,95.32) ;
\draw  [fill={rgb, 255:red, 120; green, 200; blue, 240 }  ,fill opacity=1 ] (159.55,71.14) -- (175.86,54.83) -- (200.05,54.83) -- (200.05,79.01) -- (183.74,95.32) -- (159.55,95.32) -- cycle ; \draw   (200.05,54.83) -- (183.74,71.14) -- (159.55,71.14) ; \draw   (183.74,71.14) -- (183.74,95.32) ;
\draw  [fill={rgb, 255:red, 120; green, 160; blue, 240 }  ,fill opacity=1 ] (183.74,71.14) -- (200.05,54.83) -- (224.23,54.83) -- (224.23,79.01) -- (207.92,95.32) -- (183.74,95.32) -- cycle ; \draw   (224.23,54.83) -- (207.92,71.14) -- (183.74,71.14) ; \draw   (207.92,71.14) -- (207.92,95.32) ;
\draw  [fill={rgb, 255:red, 120; green, 120; blue, 240 }  ,fill opacity=1 ] (207.92,71.14) -- (224.23,54.83) -- (248.41,54.83) -- (248.41,79.01) -- (232.1,95.32) -- (207.92,95.32) -- cycle ; \draw   (248.41,54.83) -- (232.1,71.14) -- (207.92,71.14) ; \draw   (232.1,71.14) -- (232.1,95.32) ;
\draw  [fill={rgb, 255:red, 120; green, 80; blue, 255 }  ,fill opacity=1 ] (232.1,71.14) -- (248.41,54.83) -- (272.6,54.83) -- (272.6,79.01) -- (256.29,95.32) -- (232.1,95.32) -- cycle ; \draw   (272.6,54.83) -- (256.29,71.14) -- (232.1,71.14) ; \draw   (256.29,71.14) -- (256.29,95.32) ;
\draw  [fill={rgb, 255:red, 120; green, 40; blue, 240 }  ,fill opacity=1 ] (256.29,71.14) -- (272.6,54.83) -- (296.78,54.83) -- (296.78,79.01) -- (280.47,95.32) -- (256.29,95.32) -- cycle ; \draw   (296.78,54.83) -- (280.47,71.14) -- (256.29,71.14) ; \draw   (280.47,71.14) -- (280.47,95.32) ;
\draw  [fill={rgb, 255:red, 160; green, 240; blue, 240 }  ,fill opacity=1 ] (119.06,87.45) -- (135.37,71.14) -- (159.55,71.14) -- (159.55,95.32) -- (143.24,111.63) -- (119.06,111.63) -- cycle ; \draw   (159.55,71.14) -- (143.24,87.45) -- (119.06,87.45) ; \draw   (143.24,87.45) -- (143.24,111.63) ;
\draw  [fill={rgb, 255:red, 160; green, 200; blue, 240 }  ,fill opacity=1 ] (143.24,87.45) -- (159.55,71.14) -- (183.74,71.14) -- (183.74,95.32) -- (167.43,111.63) -- (143.24,111.63) -- cycle ; \draw   (183.74,71.14) -- (167.43,87.45) -- (143.24,87.45) ; \draw   (167.43,87.45) -- (167.43,111.63) ;
\draw  [fill={rgb, 255:red, 160; green, 160; blue, 240 }  ,fill opacity=1 ] (167.43,87.45) -- (183.74,71.14) -- (207.92,71.14) -- (207.92,95.32) -- (191.61,111.63) -- (167.43,111.63) -- cycle ; \draw   (207.92,71.14) -- (191.61,87.45) -- (167.43,87.45) ; \draw   (191.61,87.45) -- (191.61,111.63) ;
\draw  [fill={rgb, 255:red, 160; green, 120; blue, 240 }  ,fill opacity=1 ] (191.61,87.45) -- (207.92,71.14) -- (232.1,71.14) -- (232.1,95.32) -- (215.79,111.63) -- (191.61,111.63) -- cycle ; \draw   (232.1,71.14) -- (215.79,87.45) -- (191.61,87.45) ; \draw   (215.79,87.45) -- (215.79,111.63) ;
\draw  [fill={rgb, 255:red, 160; green, 80; blue, 240 }  ,fill opacity=1 ] (215.79,87.45) -- (232.1,71.14) -- (256.29,71.14) -- (256.29,95.32) -- (239.98,111.63) -- (215.79,111.63) -- cycle ; \draw   (256.29,71.14) -- (239.98,87.45) -- (215.79,87.45) ; \draw   (239.98,87.45) -- (239.98,111.63) ;
\draw  [fill={rgb, 255:red, 160; green, 40; blue, 240 }  ,fill opacity=1 ] (239.98,87.45) -- (256.29,71.14) -- (280.47,71.14) -- (280.47,95.32) -- (264.16,111.63) -- (239.98,111.63) -- cycle ; \draw   (280.47,71.14) -- (264.16,87.45) -- (239.98,87.45) ; \draw   (264.16,87.45) -- (264.16,111.63) ;
\draw  [fill={rgb, 255:red, 200; green, 240; blue, 240 }  ,fill opacity=1 ] (102.75,103.76) -- (119.06,87.45) -- (143.24,87.45) -- (143.24,111.63) -- (126.93,127.94) -- (102.75,127.94) -- cycle ; \draw   (143.24,87.45) -- (126.93,103.76) -- (102.75,103.76) ; \draw   (126.93,103.76) -- (126.93,127.94) ;
\draw  [fill={rgb, 255:red, 200; green, 200; blue, 240 }  ,fill opacity=1 ] (126.93,103.76) -- (143.24,87.45) -- (167.43,87.45) -- (167.43,111.63) -- (151.12,127.94) -- (126.93,127.94) -- cycle ; \draw   (167.43,87.45) -- (151.12,103.76) -- (126.93,103.76) ; \draw   (151.12,103.76) -- (151.12,127.94) ;
\draw  [fill={rgb, 255:red, 200; green, 160; blue, 240 }  ,fill opacity=1 ] (151.12,103.76) -- (167.43,87.45) -- (191.61,87.45) -- (191.61,111.63) -- (175.3,127.94) -- (151.12,127.94) -- cycle ; \draw   (191.61,87.45) -- (175.3,103.76) -- (151.12,103.76) ; \draw   (175.3,103.76) -- (175.3,127.94) ;
\draw  [fill={rgb, 255:red, 200; green, 120; blue, 240 }  ,fill opacity=1 ] (175.3,103.76) -- (191.61,87.45) -- (215.79,87.45) -- (215.79,111.63) -- (199.48,127.94) -- (175.3,127.94) -- cycle ; \draw   (215.79,87.45) -- (199.48,103.76) -- (175.3,103.76) ; \draw   (199.48,103.76) -- (199.48,127.94) ;
\draw  [fill={rgb, 255:red, 200; green, 80; blue, 240 }  ,fill opacity=1 ] (199.48,103.76) -- (215.79,87.45) -- (239.98,87.45) -- (239.98,111.63) -- (223.67,127.94) -- (199.48,127.94) -- cycle ; \draw   (239.98,87.45) -- (223.67,103.76) -- (199.48,103.76) ; \draw   (223.67,103.76) -- (223.67,127.94) ;
\draw  [fill={rgb, 255:red, 200; green, 40; blue, 240 }  ,fill opacity=1 ] (223.67,103.76) -- (239.98,87.45) -- (264.16,87.45) -- (264.16,111.63) -- (247.85,127.94) -- (223.67,127.94) -- cycle ; \draw   (264.16,87.45) -- (247.85,103.76) -- (223.67,103.76) ; \draw   (247.85,103.76) -- (247.85,127.94) ;
\draw  [fill={rgb, 255:red, 240; green, 240; blue, 240 }  ,fill opacity=1 ] (86.44,120.07) -- (102.75,103.76) -- (126.93,103.76) -- (126.93,127.94) -- (110.62,144.25) -- (86.44,144.25) -- cycle ; \draw   (126.93,103.76) -- (110.62,120.07) -- (86.44,120.07) ; \draw   (110.62,120.07) -- (110.62,144.25) ;
\draw  [fill={rgb, 255:red, 240; green, 200; blue, 240 }  ,fill opacity=1 ] (110.62,120.07) -- (126.93,103.76) -- (151.12,103.76) -- (151.12,127.94) -- (134.81,144.25) -- (110.62,144.25) -- cycle ; \draw   (151.12,103.76) -- (134.81,120.07) -- (110.62,120.07) ; \draw   (134.81,120.07) -- (134.81,144.25) ;
\draw  [fill={rgb, 255:red, 240; green, 160; blue, 240 }  ,fill opacity=1 ] (134.81,120.07) -- (151.12,103.76) -- (175.3,103.76) -- (175.3,127.94) -- (158.99,144.25) -- (134.81,144.25) -- cycle ; \draw   (175.3,103.76) -- (158.99,120.07) -- (134.81,120.07) ; \draw   (158.99,120.07) -- (158.99,144.25) ;
\draw  [fill={rgb, 255:red, 240; green, 120; blue, 240 }  ,fill opacity=1 ] (158.99,120.07) -- (175.3,103.76) -- (199.48,103.76) -- (199.48,127.94) -- (183.17,144.25) -- (158.99,144.25) -- cycle ; \draw   (199.48,103.76) -- (183.17,120.07) -- (158.99,120.07) ; \draw   (183.17,120.07) -- (183.17,144.25) ;
\draw  [fill={rgb, 255:red, 240; green, 80; blue, 240 }  ,fill opacity=1 ] (183.17,120.07) -- (199.48,103.76) -- (223.67,103.76) -- (223.67,127.94) -- (207.36,144.25) -- (183.17,144.25) -- cycle ; \draw   (223.67,103.76) -- (207.36,120.07) -- (183.17,120.07) ; \draw   (207.36,120.07) -- (207.36,144.25) ;
\draw  [fill={rgb, 255:red, 240; green, 40; blue, 240 }  ,fill opacity=1 ] (207.36,120.07) -- (223.67,103.76) -- (247.85,103.76) -- (247.85,127.94) -- (231.54,144.25) -- (207.36,144.25) -- cycle ; \draw   (247.85,103.76) -- (231.54,120.07) -- (207.36,120.07) ; \draw   (231.54,120.07) -- (231.54,144.25) ;
\draw [color={rgb, 255:red, 126; green, 211; blue, 33 }  ,draw opacity=1 ][line width=2.25]    (350.87,107.71) -- (240.71,123.92) ;
\draw [shift={(236.75,124.5)}, rotate = 351.63] [color={rgb, 255:red, 126; green, 211; blue, 33 }  ,draw opacity=1 ][line width=2.25]    (17.49,-5.26) .. controls (11.12,-2.23) and (5.29,-0.48) .. (0,0) .. controls (5.29,0.48) and (11.12,2.23) .. (17.49,5.26)   ;

\draw (68.1,110.7) node   {$S_{111}$};
\draw (379.1,105.7) node   {$S_{611}$};
\draw (355.1,181.7) node   {$S_{666}$};
\draw (65.5,260.7) node   {$S_{161}$};
\draw (155.9,14.5) node   {$S_{116}$};
\draw (358.7,24.1) node   {$S_{616}$};
\draw (250.3,274.5) node   {$S_{661}$};

\end{tikzpicture}

%% file: images/Slices_fig.tex
\tikzset{every picture/.style={line width=0.75pt}} 
    \begin{tikzpicture}[x=0.75pt,y=0.75pt,yscale=-.44,xscale=.44]

\draw  [fill={rgb, 255:red, 255; green, 40; blue, 40 }  ,fill opacity=1 ] (388.56,45) -- (396.12,37.44) -- (493.56,37.44) -- (493.56,134.88) -- (486,142.44) -- (388.56,142.44) -- cycle ; \draw   (493.56,37.44) -- (486,45) -- (388.56,45) ; \draw   (486,45) -- (486,142.44) ;
\draw  [fill={rgb, 255:red, 255; green, 80; blue, 80 }  ,fill opacity=1 ] (381,52.56) -- (388.56,45) -- (486,45) -- (486,142.44) -- (478.44,150) -- (381,150) -- cycle ; \draw   (486,45) -- (478.44,52.56) -- (381,52.56) ; \draw   (478.44,52.56) -- (478.44,150) ;
\draw  [fill={rgb, 255:red, 255; green, 120; blue, 120 }  ,fill opacity=1 ] (373.56,60) -- (381.12,52.44) -- (478.56,52.44) -- (478.56,149.88) -- (471,157.44) -- (373.56,157.44) -- cycle ; \draw   (478.56,52.44) -- (471,60) -- (373.56,60) ; \draw   (471,60) -- (471,157.44) ;
\draw  [fill={rgb, 255:red, 255; green, 160; blue, 160 }  ,fill opacity=1 ] (366,67.56) -- (373.56,60) -- (471,60) -- (471,157.44) -- (463.44,165) -- (366,165) -- cycle ; \draw   (471,60) -- (463.44,67.56) -- (366,67.56) ; \draw   (463.44,67.56) -- (463.44,165) ;
\draw  [fill={rgb, 255:red, 255; green, 200; blue, 200 }  ,fill opacity=1 ] (358.56,75) -- (366.12,67.44) -- (463.56,67.44) -- (463.56,164.88) -- (456,172.44) -- (358.56,172.44) -- cycle ; \draw   (463.56,67.44) -- (456,75) -- (358.56,75) ; \draw   (456,75) -- (456,172.44) ;
\draw  [fill={rgb, 255:red, 255; green, 240; blue, 240 }  ,fill opacity=1 ] (351,82.56) -- (358.56,75) -- (456,75) -- (456,172.44) -- (448.44,180) -- (351,180) -- cycle ; \draw   (456,75) -- (448.44,82.56) -- (351,82.56) ; \draw   (448.44,82.56) -- (448.44,180) ;
\draw  [fill={rgb, 255:red, 240; green, 255; blue, 240 }  ,fill opacity=1 ] (32.18,74.18) -- (62.36,44) -- (77.18,44) -- (77.18,148.82) -- (47,179) -- (32.18,179) -- cycle ; \draw   (77.18,44) -- (47,74.18) -- (32.18,74.18) ; \draw   (47,74.18) -- (47,179) ;
\draw  [fill={rgb, 255:red, 200; green, 255; blue, 200 }  ,fill opacity=1 ] (47,74.18) -- (77.18,44) -- (92,44) -- (92,148.82) -- (61.82,179) -- (47,179) -- cycle ; \draw   (92,44) -- (61.82,74.18) -- (47,74.18) ; \draw   (61.82,74.18) -- (61.82,179) ;
\draw  [fill={rgb, 255:red, 160; green, 255; blue, 160 }  ,fill opacity=1 ] (62,74.18) -- (92.18,44) -- (107,44) -- (107,148.82) -- (76.82,179) -- (62,179) -- cycle ; \draw   (107,44) -- (76.82,74.18) -- (62,74.18) ; \draw   (76.82,74.18) -- (76.82,179) ;
\draw  [fill={rgb, 255:red, 120; green, 255; blue, 120 }  ,fill opacity=1 ] (77,74.18) -- (107.18,44) -- (122,44) -- (122,148.82) -- (91.82,179) -- (77,179) -- cycle ; \draw   (122,44) -- (91.82,74.18) -- (77,74.18) ; \draw   (91.82,74.18) -- (91.82,179) ;
\draw  [fill={rgb, 255:red, 80; green, 255; blue, 80 }  ,fill opacity=1 ] (92,74.18) -- (122.18,44) -- (137,44) -- (137,148.82) -- (106.82,179) -- (92,179) -- cycle ; \draw   (137,44) -- (106.82,74.18) -- (92,74.18) ; \draw   (106.82,74.18) -- (106.82,179) ;
\draw  [fill={rgb, 255:red, 40; green, 255; blue, 40 }  ,fill opacity=1 ] (107,74.18) -- (137.18,44) -- (152,44) -- (152,148.82) -- (121.82,179) -- (107,179) -- cycle ; \draw   (152,44) -- (121.82,74.18) -- (107,74.18) ; \draw   (121.82,74.18) -- (121.82,179) ;
\draw  [fill={rgb, 255:red, 40; green, 40; blue, 255 }  ,fill opacity=1 ] (185,164.97) -- (229.97,120) -- (320,120) -- (320,135.03) -- (275.03,180) -- (185,180) -- cycle ; \draw   (320,120) -- (275.03,164.97) -- (185,164.97) ; \draw   (275.03,164.97) -- (275.03,180) ;
\draw  [fill={rgb, 255:red, 80; green, 80; blue, 255 }  ,fill opacity=1 ] (185,149.95) -- (229.97,104.97) -- (320,104.97) -- (320,120) -- (275.03,164.97) -- (185,164.97) -- cycle ; \draw   (320,104.97) -- (275.03,149.95) -- (185,149.95) ; \draw   (275.03,149.95) -- (275.03,164.97) ;
\draw  [fill={rgb, 255:red, 120; green, 120; blue, 255 }  ,fill opacity=1 ] (185,134.92) -- (229.97,89.95) -- (320,89.95) -- (320,104.97) -- (275.03,149.95) -- (185,149.95) -- cycle ; \draw   (320,89.95) -- (275.03,134.92) -- (185,134.92) ; \draw   (275.03,134.92) -- (275.03,149.95) ;
\draw  [fill={rgb, 255:red, 160; green, 160; blue, 255 }  ,fill opacity=1 ] (185,119.97) -- (229.97,75) -- (320,75) -- (320,90.03) -- (275.03,135) -- (185,135) -- cycle ; \draw   (320,75) -- (275.03,119.97) -- (185,119.97) ; \draw   (275.03,119.97) -- (275.03,135) ;
\draw  [fill={rgb, 255:red, 200; green, 200; blue, 255 }  ,fill opacity=1 ] (185,104.97) -- (229.97,60) -- (320,60) -- (320,75.03) -- (275.03,120) -- (185,120) -- cycle ; \draw   (320,60) -- (275.03,104.97) -- (185,104.97) ; \draw   (275.03,104.97) -- (275.03,120) ;
\draw  [fill={rgb, 255:red, 240; green, 240; blue, 255 }  ,fill opacity=1 ] (185,89.95) -- (229.97,44.97) -- (320,44.97) -- (320,60) -- (275.03,104.97) -- (185,104.97) -- cycle ; \draw   (320,44.97) -- (275.03,89.95) -- (185,89.95) ; \draw   (275.03,89.95) -- (275.03,104.97) ;

\draw (31.5,191) node  [align=left] {$\displaystyle S_{1}( 1)$};
\draw (131.5,190) node  [align=left] {$\displaystyle S_{1}( 6)$};
\draw (250.5,70) node  [align=left] {$\displaystyle S_{2}( 1)$};
\draw (230.5,198) node  [align=left] {$\displaystyle S_{2}( 6)$};
\draw (395.5,130) node [scale=0.9] [align=left] {$\displaystyle S_{3}( 1)$};
\draw (500.5,25) node  [align=left] {$\displaystyle S_{3}( 6)$};

\end{tikzpicture}

%% file: images/T_Shape_fig.tex
\tikzset{every picture/.style={line width=0.75pt}} 
    \begin{tikzpicture}[x=0.75pt,y=0.75pt,yscale=-1,xscale=1]

\draw  [fill={rgb, 255:red, 208; green, 2; blue, 27 }  ,fill opacity=1 ] (66.2,41) -- (71,36.2) -- (86.2,36.2) -- (86.2,47.4) -- (81.4,52.2) -- (66.2,52.2) -- cycle ; \draw   (86.2,36.2) -- (81.4,41) -- (66.2,41) ; \draw   (81.4,41) -- (81.4,52.2) ;
\draw  [fill={rgb, 255:red, 208; green, 2; blue, 27 }  ,fill opacity=1 ] (81.4,52.2) -- (86.2,47.4) -- (101.4,47.4) -- (101.4,58.6) -- (96.6,63.4) -- (81.4,63.4) -- cycle ; \draw   (101.4,47.4) -- (96.6,52.2) -- (81.4,52.2) ; \draw   (96.6,52.2) -- (96.6,63.4) ;
\draw  [fill={rgb, 255:red, 208; green, 2; blue, 27 }  ,fill opacity=1 ] (51,29.8) -- (55.8,25) -- (71,25) -- (71,36.2) -- (66.2,41) -- (51,41) -- cycle ; \draw   (71,25) -- (66.2,29.8) -- (51,29.8) ; \draw   (66.2,29.8) -- (66.2,41) ;
\draw  [fill={rgb, 255:red, 208; green, 2; blue, 27 }  ,fill opacity=1 ] (35.8,18.6) -- (40.6,13.8) -- (55.8,13.8) -- (55.8,25) -- (51,29.8) -- (35.8,29.8) -- cycle ; \draw   (55.8,13.8) -- (51,18.6) -- (35.8,18.6) ; \draw   (51,18.6) -- (51,29.8) ;
\draw  [fill={rgb, 255:red, 255; green, 255; blue, 255 }  ,fill opacity=0.5 ] (31,23.4) -- (35.8,18.6) -- (51,18.6) -- (51,29.8) -- (46.2,34.6) -- (31,34.6) -- cycle ; \draw   (51,18.6) -- (46.2,23.4) -- (31,23.4) ; \draw   (46.2,23.4) -- (46.2,34.6) ;
\draw  [fill={rgb, 255:red, 255; green, 255; blue, 255 }  ,fill opacity=0.5 ] (26.2,28.2) -- (31,23.4) -- (46.2,23.4) -- (46.2,34.6) -- (41.4,39.4) -- (26.2,39.4) -- cycle ; \draw   (46.2,23.4) -- (41.4,28.2) -- (26.2,28.2) ; \draw   (41.4,28.2) -- (41.4,39.4) ;
\draw  [fill={rgb, 255:red, 255; green, 255; blue, 255 }  ,fill opacity=0.5 ][dash pattern={on 4.5pt off 4.5pt}] (21.4,33) -- (26.2,28.2) -- (41.4,28.2) -- (41.4,39.4) -- (36.6,44.2) -- (21.4,44.2) -- cycle ; \draw  [dash pattern={on 4.5pt off 4.5pt}] (41.4,28.2) -- (36.6,33) -- (21.4,33) ; \draw  [dash pattern={on 4.5pt off 4.5pt}] (36.6,33) -- (36.6,44.2) ;
\draw  [fill={rgb, 255:red, 208; green, 2; blue, 27 }  ,fill opacity=1 ] (71.8,84.2) -- (76.6,79.4) -- (91.8,79.4) -- (91.8,90.6) -- (87,95.4) -- (71.8,95.4) -- cycle ; \draw   (91.8,79.4) -- (87,84.2) -- (71.8,84.2) ; \draw   (87,84.2) -- (87,95.4) ;
\draw  [fill={rgb, 255:red, 208; green, 2; blue, 27 }  ,fill opacity=1 ] (11.8,98.6) -- (16.6,93.8) -- (31.8,93.8) -- (31.8,105) -- (27,109.8) -- (11.8,109.8) -- cycle ; \draw   (31.8,93.8) -- (27,98.6) -- (11.8,98.6) ; \draw   (27,98.6) -- (27,109.8) ;
\draw  [fill={rgb, 255:red, 255; green, 255; blue, 255 }  ,fill opacity=0.5 ] (11.8,87.4) -- (16.6,82.6) -- (31.8,82.6) -- (31.8,93.8) -- (27,98.6) -- (11.8,98.6) -- cycle ; \draw   (31.8,82.6) -- (27,87.4) -- (11.8,87.4) ; \draw   (27,87.4) -- (27,98.6) ;
\draw  [fill={rgb, 255:red, 255; green, 255; blue, 255 }  ,fill opacity=0.5 ] (11.8,76.2) -- (16.6,71.4) -- (31.8,71.4) -- (31.8,82.6) -- (27,87.4) -- (11.8,87.4) -- cycle ; \draw   (31.8,71.4) -- (27,76.2) -- (11.8,76.2) ; \draw   (27,76.2) -- (27,87.4) ;
\draw  [fill={rgb, 255:red, 255; green, 255; blue, 255 }  ,fill opacity=0.5 ][line width=0.75]  (11.8,65) -- (16.6,60.2) -- (31.8,60.2) -- (31.8,71.4) -- (27,76.2) -- (11.8,76.2) -- cycle ; \draw  [line width=0.75]  (31.8,60.2) -- (27,65) -- (11.8,65) ; \draw  [line width=0.75]  (27,65) -- (27,76.2) ;
\draw  [fill={rgb, 255:red, 144; green, 19; blue, 254 }  ,fill opacity=1 ] (111.8,74.6) -- (116.6,69.8) -- (131.8,69.8) -- (131.8,81) -- (127,85.8) -- (111.8,85.8) -- cycle ; \draw   (131.8,69.8) -- (127,74.6) -- (111.8,74.6) ; \draw   (127,74.6) -- (127,85.8) ;
\draw  [fill={rgb, 255:red, 184; green, 233; blue, 134 }  ,fill opacity=1 ] (96.6,63.4) -- (101.4,58.6) -- (116.6,58.6) -- (116.6,69.8) -- (111.8,74.6) -- (96.6,74.6) -- cycle ; \draw   (116.6,58.6) -- (111.8,63.4) -- (96.6,63.4) ; \draw   (111.8,63.4) -- (111.8,74.6) ;
\draw  [fill={rgb, 255:red, 184; green, 233; blue, 134 }  ,fill opacity=1 ] (91.8,79.4) -- (96.6,74.6) -- (111.8,74.6) -- (111.8,85.8) -- (107,90.6) -- (91.8,90.6) -- cycle ; \draw   (111.8,74.6) -- (107,79.4) -- (91.8,79.4) ; \draw   (107,79.4) -- (107,90.6) ;
\draw  [fill={rgb, 255:red, 74; green, 144; blue, 226 }  ,fill opacity=1 ] (91.8,68.2) -- (96.6,63.4) -- (111.8,63.4) -- (111.8,74.6) -- (107,79.4) -- (91.8,79.4) -- cycle ; \draw   (111.8,63.4) -- (107,68.2) -- (91.8,68.2) ; \draw   (107,68.2) -- (107,79.4) ;
\draw  [fill={rgb, 255:red, 184; green, 233; blue, 134 }  ,fill opacity=1 ] (107,68.2) -- (111.8,63.4) -- (127,63.4) -- (127,74.6) -- (122.2,79.4) -- (107,79.4) -- cycle ; \draw   (127,63.4) -- (122.2,68.2) -- (107,68.2) ; \draw   (122.2,68.2) -- (122.2,79.4) ;
\draw  [fill={rgb, 255:red, 184; green, 233; blue, 134 }  ,fill opacity=1 ] (76.6,57) -- (81.4,52.2) -- (96.6,52.2) -- (96.6,63.4) -- (91.8,68.2) -- (76.6,68.2) -- cycle ; \draw   (96.6,52.2) -- (91.8,57) -- (76.6,57) ; \draw   (91.8,57) -- (91.8,68.2) ;
\draw  [fill={rgb, 255:red, 184; green, 233; blue, 134 }  ,fill opacity=1 ] (71.8,73) -- (76.6,68.2) -- (91.8,68.2) -- (91.8,79.4) -- (87,84.2) -- (71.8,84.2) -- cycle ; \draw   (91.8,68.2) -- (87,73) -- (71.8,73) ; \draw   (87,73) -- (87,84.2) ;
\draw  [fill={rgb, 255:red, 74; green, 144; blue, 226 }  ,fill opacity=1 ] (71.8,61.8) -- (76.6,57) -- (91.8,57) -- (91.8,68.2) -- (87,73) -- (71.8,73) -- cycle ; \draw   (91.8,57) -- (87,61.8) -- (71.8,61.8) ; \draw   (87,61.8) -- (87,73) ;
\draw  [fill={rgb, 255:red, 184; green, 233; blue, 134 }  ,fill opacity=1 ] (87,61.8) -- (91.8,57) -- (107,57) -- (107,68.2) -- (102.2,73) -- (87,73) -- cycle ; \draw   (107,57) -- (102.2,61.8) -- (87,61.8) ; \draw   (102.2,61.8) -- (102.2,73) ;
\draw  [fill={rgb, 255:red, 184; green, 233; blue, 134 }  ,fill opacity=1 ] (56.6,50.6) -- (61.4,45.8) -- (76.6,45.8) -- (76.6,57) -- (71.8,61.8) -- (56.6,61.8) -- cycle ; \draw   (76.6,45.8) -- (71.8,50.6) -- (56.6,50.6) ; \draw   (71.8,50.6) -- (71.8,61.8) ;
\draw  [fill={rgb, 255:red, 184; green, 233; blue, 134 }  ,fill opacity=1 ] (51.8,66.6) -- (56.6,61.8) -- (71.8,61.8) -- (71.8,73) -- (67,77.8) -- (51.8,77.8) -- cycle ; \draw   (71.8,61.8) -- (67,66.6) -- (51.8,66.6) ; \draw   (67,66.6) -- (67,77.8) ;
\draw  [fill={rgb, 255:red, 74; green, 144; blue, 226 }  ,fill opacity=1 ] (51.8,55.4) -- (56.6,50.6) -- (71.8,50.6) -- (71.8,61.8) -- (67,66.6) -- (51.8,66.6) -- cycle ; \draw   (71.8,50.6) -- (67,55.4) -- (51.8,55.4) ; \draw   (67,55.4) -- (67,66.6) ;
\draw  [fill={rgb, 255:red, 184; green, 233; blue, 134 }  ,fill opacity=1 ] (67,55.4) -- (71.8,50.6) -- (87,50.6) -- (87,61.8) -- (82.2,66.6) -- (67,66.6) -- cycle ; \draw   (87,50.6) -- (82.2,55.4) -- (67,55.4) ; \draw   (82.2,55.4) -- (82.2,66.6) ;
\draw  [fill={rgb, 255:red, 184; green, 233; blue, 134 }  ,fill opacity=1 ] (36.6,44.2) -- (41.4,39.4) -- (56.6,39.4) -- (56.6,50.6) -- (51.8,55.4) -- (36.6,55.4) -- cycle ; \draw   (56.6,39.4) -- (51.8,44.2) -- (36.6,44.2) ; \draw   (51.8,44.2) -- (51.8,55.4) ;
\draw  [fill={rgb, 255:red, 184; green, 233; blue, 134 }  ,fill opacity=1 ] (31.8,60.2) -- (36.6,55.4) -- (51.8,55.4) -- (51.8,66.6) -- (47,71.4) -- (31.8,71.4) -- cycle ; \draw   (51.8,55.4) -- (47,60.2) -- (31.8,60.2) ; \draw   (47,60.2) -- (47,71.4) ;
\draw  [fill={rgb, 255:red, 74; green, 144; blue, 226 }  ,fill opacity=1 ] (31.8,49) -- (36.6,44.2) -- (51.8,44.2) -- (51.8,55.4) -- (47,60.2) -- (31.8,60.2) -- cycle ; \draw   (51.8,44.2) -- (47,49) -- (31.8,49) ; \draw   (47,49) -- (47,60.2) ;
\draw  [fill={rgb, 255:red, 184; green, 233; blue, 134 }  ,fill opacity=1 ] (47,49) -- (51.8,44.2) -- (67,44.2) -- (67,55.4) -- (62.2,60.2) -- (47,60.2) -- cycle ; \draw   (67,44.2) -- (62.2,49) -- (47,49) ; \draw   (62.2,49) -- (62.2,60.2) ;
\draw  [fill={rgb, 255:red, 184; green, 233; blue, 134 }  ,fill opacity=1 ] (16.6,37.8) -- (21.4,33) -- (36.6,33) -- (36.6,44.2) -- (31.8,49) -- (16.6,49) -- cycle ; \draw   (36.6,33) -- (31.8,37.8) -- (16.6,37.8) ; \draw   (31.8,37.8) -- (31.8,49) ;
\draw  [fill={rgb, 255:red, 184; green, 233; blue, 134 }  ,fill opacity=1 ] (11.8,53.8) -- (16.6,49) -- (31.8,49) -- (31.8,60.2) -- (27,65) -- (11.8,65) -- cycle ; \draw   (31.8,49) -- (27,53.8) -- (11.8,53.8) ; \draw   (27,53.8) -- (27,65) ;
\draw  [fill={rgb, 255:red, 74; green, 144; blue, 226 }  ,fill opacity=1 ] (11.8,42.6) -- (16.6,37.8) -- (31.8,37.8) -- (31.8,49) -- (27,53.8) -- (11.8,53.8) -- cycle ; \draw   (31.8,37.8) -- (27,42.6) -- (11.8,42.6) ; \draw   (27,42.6) -- (27,53.8) ;
\draw  [fill={rgb, 255:red, 184; green, 233; blue, 134 }  ,fill opacity=1 ] (27,42.6) -- (31.8,37.8) -- (47,37.8) -- (47,49) -- (42.2,53.8) -- (27,53.8) -- cycle ; \draw   (47,37.8) -- (42.2,42.6) -- (27,42.6) ; \draw   (42.2,42.6) -- (42.2,53.8) ;
\draw  [fill={rgb, 255:red, 255; green, 255; blue, 255 }  ,fill opacity=0.5 ] (42.2,42.6) -- (47,37.8) -- (62.2,37.8) -- (62.2,49) -- (57.4,53.8) -- (42.2,53.8) -- cycle ; \draw   (62.2,37.8) -- (57.4,42.6) -- (42.2,42.6) ; \draw   (57.4,42.6) -- (57.4,53.8) ;
\draw  [fill={rgb, 255:red, 255; green, 255; blue, 255 }  ,fill opacity=0.5 ] (57.4,42.6) -- (62.2,37.8) -- (77.4,37.8) -- (77.4,49) -- (72.6,53.8) -- (57.4,53.8) -- cycle ; \draw   (77.4,37.8) -- (72.6,42.6) -- (57.4,42.6) ; \draw   (72.6,42.6) -- (72.6,53.8) ;
\draw  [fill={rgb, 255:red, 255; green, 255; blue, 255 }  ,fill opacity=0.5 ] (72.6,42.6) -- (77.4,37.8) -- (92.6,37.8) -- (92.6,49) -- (87.8,53.8) -- (72.6,53.8) -- cycle ; \draw   (92.6,37.8) -- (87.8,42.6) -- (72.6,42.6) ; \draw   (87.8,42.6) -- (87.8,53.8) ;
\draw  [fill={rgb, 255:red, 208; green, 2; blue, 27 }  ,fill opacity=1 ] (102.2,61.8) -- (107,57) -- (122.2,57) -- (122.2,68.2) -- (117.4,73) -- (102.2,73) -- cycle ; \draw   (122.2,57) -- (117.4,61.8) -- (102.2,61.8) ; \draw   (117.4,61.8) -- (117.4,73) ;
\draw  [fill={rgb, 255:red, 208; green, 2; blue, 27 }  ,fill opacity=1 ] (97.4,55.4) -- (102.2,50.6) -- (117.4,50.6) -- (117.4,61.8) -- (112.6,66.6) -- (97.4,66.6) -- cycle ; \draw   (117.4,50.6) -- (112.6,55.4) -- (97.4,55.4) ; \draw   (112.6,55.4) -- (112.6,66.6) ;
\draw  [fill={rgb, 255:red, 208; green, 2; blue, 27 }  ,fill opacity=1 ] (92.6,49) -- (97.4,44.2) -- (112.6,44.2) -- (112.6,55.4) -- (107.8,60.2) -- (92.6,60.2) -- cycle ; \draw   (112.6,44.2) -- (107.8,49) -- (92.6,49) ; \draw   (107.8,49) -- (107.8,60.2) ;
\draw  [fill={rgb, 255:red, 208; green, 2; blue, 27 }  ,fill opacity=1 ] (87.8,42.6) -- (92.6,37.8) -- (107.8,37.8) -- (107.8,49) -- (103,53.8) -- (87.8,53.8) -- cycle ; \draw   (107.8,37.8) -- (103,42.6) -- (87.8,42.6) ; \draw   (103,42.6) -- (103,53.8) ;
\draw  [fill={rgb, 255:red, 208; green, 2; blue, 27 }  ,fill opacity=1 ] (51.8,89) -- (56.6,84.2) -- (71.8,84.2) -- (71.8,95.4) -- (67,100.2) -- (51.8,100.2) -- cycle ; \draw   (71.8,84.2) -- (67,89) -- (51.8,89) ; \draw   (67,89) -- (67,100.2) ;
\draw  [fill={rgb, 255:red, 208; green, 2; blue, 27 }  ,fill opacity=1 ] (31.8,93.8) -- (36.6,89) -- (51.8,89) -- (51.8,100.2) -- (47,105) -- (31.8,105) -- cycle ; \draw   (51.8,89) -- (47,93.8) -- (31.8,93.8) ; \draw   (47,93.8) -- (47,105) ;

\end{tikzpicture}

%% file: images/lone_divider_partition.tex
\tikzset{every picture/.style={line width=0.75pt}} 

\begin{tikzpicture}[x=0.75pt,y=0.75pt,yscale=-0.5,xscale=0.5]

\draw  [fill={rgb, 255:red, 189; green, 16; blue, 224 }  ,fill opacity=1 ] (10,61) -- (60,61) -- (60,241) -- (10,241) -- cycle ;
\draw  [fill={rgb, 255:red, 189; green, 16; blue, 224 }  ,fill opacity=1 ] (290,60) -- (340,60) -- (340,240) -- (290,240) -- cycle ;
\draw  [fill={rgb, 255:red, 189; green, 16; blue, 224 }  ,fill opacity=1 ] (220,61) -- (270,61) -- (270,241) -- (220,241) -- cycle ;
\draw  [fill={rgb, 255:red, 80; green, 227; blue, 194 }  ,fill opacity=1 ] (80,61) -- (130,61) -- (130,241) -- (80,241) -- cycle ;
\draw  [fill={rgb, 255:red, 80; green, 227; blue, 194 }  ,fill opacity=1 ] (360,60) -- (410,60) -- (410,240) -- (360,240) -- cycle ;
\draw  [fill={rgb, 255:red, 80; green, 227; blue, 194 }  ,fill opacity=1 ] (500,60) -- (550,60) -- (550,240) -- (500,240) -- cycle ;
\draw  [fill={rgb, 255:red, 80; green, 227; blue, 194 }  ,fill opacity=1 ] (570,60) -- (620,60) -- (620,240) -- (570,240) -- cycle ;
\draw  [fill={rgb, 255:red, 80; green, 227; blue, 194 }  ,fill opacity=1 ] (640,60) -- (690,60) -- (690,240) -- (640,240) -- cycle ;
\draw  [fill={rgb, 255:red, 255; green, 255; blue, 255 }  ,fill opacity=1 ] (10,61) -- (60,61) -- (60,111) -- (10,111) -- cycle ;
\draw  [fill={rgb, 255:red, 255; green, 255; blue, 255 }  ,fill opacity=1 ] (80,61) -- (130,61) -- (130,111) -- (80,111) -- cycle ;
\draw  [fill={rgb, 255:red, 255; green, 255; blue, 255 }  ,fill opacity=1 ] (290,60) -- (340,60) -- (340,110) -- (290,110) -- cycle ;
\draw  [fill={rgb, 255:red, 255; green, 255; blue, 255 }  ,fill opacity=1 ] (220,61) -- (270,61) -- (270,111) -- (220,111) -- cycle ;
\draw  [fill={rgb, 255:red, 255; green, 255; blue, 255 }  ,fill opacity=1 ] (360,60) -- (410,60) -- (410,110) -- (360,110) -- cycle ;
\draw  [fill={rgb, 255:red, 255; green, 255; blue, 255 }  ,fill opacity=1 ] (500,60) -- (550,60) -- (550,110) -- (500,110) -- cycle ;
\draw  [fill={rgb, 255:red, 255; green, 255; blue, 255 }  ,fill opacity=1 ] (570,60) -- (620,60) -- (620,110) -- (570,110) -- cycle ;
\draw  [fill={rgb, 255:red, 255; green, 255; blue, 255 }  ,fill opacity=1 ] (640,60) -- (690,60) -- (690,110) -- (640,110) -- cycle ;
\draw    (560,50) -- (560,250) ;
\draw  [fill={rgb, 255:red, 208; green, 2; blue, 27 }  ,fill opacity=1 ] (105,72.5) -- (115,86) -- (105,99.5) -- (95,86) -- cycle ;
\draw  [fill={rgb, 255:red, 248; green, 231; blue, 28 }  ,fill opacity=1 ] (375,75) -- (395,75) -- (395,95) -- (375,95) -- cycle ;
\draw  [fill={rgb, 255:red, 248; green, 231; blue, 28 }  ,fill opacity=1 ] (515,75) -- (535,75) -- (535,95) -- (515,95) -- cycle ;
\draw  [fill={rgb, 255:red, 248; green, 231; blue, 28 }  ,fill opacity=1 ] (585,75) -- (605,75) -- (605,95) -- (585,95) -- cycle ;
\draw  [fill={rgb, 255:red, 248; green, 231; blue, 28 }  ,fill opacity=1 ] (655,75) -- (675,75) -- (675,95) -- (655,95) -- cycle ;
\draw  [fill={rgb, 255:red, 80; green, 227; blue, 194 }  ,fill opacity=1 ] (305,85) .. controls (305,79.48) and (309.48,75) .. (315,75) .. controls (320.52,75) and (325,79.48) .. (325,85) .. controls (325,90.52) and (320.52,95) .. (315,95) .. controls (309.48,95) and (305,90.52) .. (305,85) -- cycle ;
\draw  [fill={rgb, 255:red, 80; green, 227; blue, 194 }  ,fill opacity=1 ] (235,86) .. controls (235,80.48) and (239.48,76) .. (245,76) .. controls (250.52,76) and (255,80.48) .. (255,86) .. controls (255,91.52) and (250.52,96) .. (245,96) .. controls (239.48,96) and (235,91.52) .. (235,86) -- cycle ;
\draw  [fill={rgb, 255:red, 189; green, 16; blue, 224 }  ,fill opacity=1 ] (35,75) -- (45,97) -- (25,97) -- cycle ;
\draw   (691.7,52.8) .. controls (691.71,48.13) and (689.39,45.79) .. (684.72,45.78) -- (641.94,45.68) .. controls (635.27,45.66) and (631.94,43.32) .. (631.95,38.65) .. controls (631.94,43.32) and (628.61,45.64) .. (621.94,45.63)(624.94,45.64) -- (578.02,45.52) .. controls (573.35,45.51) and (571.01,47.84) .. (571,52.51) ;
\draw  [fill={rgb, 255:red, 189; green, 16; blue, 224 }  ,fill opacity=1 ] (710,60) -- (760,60) -- (760,240) -- (710,240) -- cycle ;
\draw  [fill={rgb, 255:red, 189; green, 16; blue, 224 }  ,fill opacity=1 ] (780,60) -- (830,60) -- (830,240) -- (780,240) -- cycle ;
\draw  [fill={rgb, 255:red, 80; green, 227; blue, 194 }  ,fill opacity=1 ] (430,60) -- (480,60) -- (480,240) -- (430,240) -- cycle ;
\draw  [fill={rgb, 255:red, 255; green, 255; blue, 255 }  ,fill opacity=1 ] (430,60) -- (480,60) -- (480,110) -- (430,110) -- cycle ;
\draw  [fill={rgb, 255:red, 208; green, 2; blue, 27 }  ,fill opacity=1 ] (455,71.5) -- (465,85) -- (455,98.5) -- (445,85) -- cycle ;
\draw  [fill={rgb, 255:red, 189; green, 16; blue, 224 }  ,fill opacity=1 ] (150,60) -- (200,60) -- (200,240) -- (150,240) -- cycle ;
\draw  [fill={rgb, 255:red, 255; green, 255; blue, 255 }  ,fill opacity=1 ] (150,60) -- (200,60) -- (200,110) -- (150,110) -- cycle ;
\draw  [fill={rgb, 255:red, 189; green, 16; blue, 224 }  ,fill opacity=1 ] (175,74) -- (185,96) -- (165,96) -- cycle ;
\draw  [fill={rgb, 255:red, 208; green, 2; blue, 27 }  ,fill opacity=1 ] (80,111) -- (130,111) -- (130,132) -- (80,132) -- cycle ;
\draw  [fill={rgb, 255:red, 208; green, 2; blue, 27 }  ,fill opacity=1 ] (150,110) -- (200,110) -- (200,131) -- (150,131) -- cycle ;
\draw  [fill={rgb, 255:red, 208; green, 2; blue, 27 }  ,fill opacity=1 ] (500,110) -- (550,110) -- (550,131) -- (500,131) -- cycle ;
\draw  [fill={rgb, 255:red, 208; green, 2; blue, 27 }  ,fill opacity=1 ] (710,60) -- (760,60) -- (760,81) -- (710,81) -- cycle ;
\draw  [color={rgb, 255:red, 0; green, 0; blue, 0 }  ,draw opacity=1 ][line width=2.25]  (80,111) -- (130,111) -- (130,240) -- (80,240) -- cycle ;
\draw  [color={rgb, 255:red, 0; green, 0; blue, 0 }  ,draw opacity=1 ][line width=2.25]  (10,111) -- (60,111) -- (60,240) -- (10,240) -- cycle ;
\draw  [color={rgb, 255:red, 0; green, 0; blue, 0 }  ,draw opacity=1 ][line width=2.25]  (150,111) -- (200,111) -- (200,240) -- (150,240) -- cycle ;
\draw  [color={rgb, 255:red, 0; green, 0; blue, 0 }  ,draw opacity=1 ][line width=2.25]  (220,112) -- (270,112) -- (270,241) -- (220,241) -- cycle ;
\draw  [color={rgb, 255:red, 0; green, 0; blue, 0 }  ,draw opacity=1 ][line width=2.25]  (290,111) -- (340,111) -- (340,240) -- (290,240) -- cycle ;
\draw  [color={rgb, 255:red, 0; green, 0; blue, 0 }  ,draw opacity=1 ][line width=2.25]  (360,110) -- (410,110) -- (410,239) -- (360,239) -- cycle ;
\draw  [color={rgb, 255:red, 0; green, 0; blue, 0 }  ,draw opacity=1 ][line width=2.25]  (430,110) -- (480,110) -- (480,239) -- (430,239) -- cycle ;
\draw  [color={rgb, 255:red, 0; green, 0; blue, 0 }  ,draw opacity=1 ][line width=2.25]  (500,111) -- (550,111) -- (550,240) -- (500,240) -- cycle ;
\draw  [color={rgb, 255:red, 0; green, 0; blue, 0 }  ,draw opacity=1 ][line width=2.25]  (570,110) -- (620,110) -- (620,239) -- (570,239) -- cycle ;
\draw  [color={rgb, 255:red, 0; green, 0; blue, 0 }  ,draw opacity=1 ][line width=2.25]  (640,110) -- (690,110) -- (690,239) -- (640,239) -- cycle ;
\draw  [color={rgb, 255:red, 0; green, 0; blue, 0 }  ,draw opacity=1 ][line width=2.25]  (710,60) -- (760,60) -- (760,240) -- (710,240) -- cycle ;
\draw  [color={rgb, 255:red, 0; green, 0; blue, 0 }  ,draw opacity=1 ][line width=2.25]  (780,60) -- (830,60) -- (830,239) -- (780,239) -- cycle ;

\draw (522,-5) node [anchor=north west][inner sep=0.75pt]  [font=\Large]  {$\lfloor \frac{2n}{3} \rfloor $};
\draw (21,242.4) node [anchor=north west][inner sep=0.75pt]  [font=\Large]  {$A_{1}$};
\draw (91,242.4) node [anchor=north west][inner sep=0.75pt]  [font=\Large]  {$A_{2}$};
\draw (161,242.4) node [anchor=north west][inner sep=0.75pt]  [font=\Large]  {$A_{3}$};
\draw (231,242.4) node [anchor=north west][inner sep=0.75pt]  [font=\Large]  {$A_{4}$};
\draw (301,242.4) node [anchor=north west][inner sep=0.75pt]  [font=\Large]  {$A_{5}$};
\draw (371,242.4) node [anchor=north west][inner sep=0.75pt]  [font=\Large]  {$A_{6}$};
\draw (443,242.4) node [anchor=north west][inner sep=0.75pt]  [font=\Large]  {$A_{7}$};
\draw (513,242.4) node [anchor=north west][inner sep=0.75pt]  [font=\Large]  {$A_{8}$};
\draw (581,242.4) node [anchor=north west][inner sep=0.75pt]  [font=\Large]  {$A_{9}$};
\draw (622,12.4) node [anchor=north west][inner sep=0.75pt]  [font=\Large]  {$s$};
\draw (651,242.4) node [anchor=north west][inner sep=0.75pt]  [font=\Large]  {$A_{10}$};
\draw (721,242.4) node [anchor=north west][inner sep=0.75pt]  [font=\Large]  {$A_{11}$};
\draw (791,242.4) node [anchor=north west][inner sep=0.75pt]  [font=\Large]  {$A_{12}$};

\end{tikzpicture}

%% file: algs/alg_n=3.tex
\begin{algorithm}[t] \small
\SetKwInOut{Input}{Input}\SetKwInOut{Output}{Output}
    \Input{Instance $I=\ins{N, M, V}$ such that $n = |N| = 2$}
    \Output{Allocation  $A = (A_{1}, A_{2})$}
    
    $I'= \ins{N, M, V'} \leftarrow$ scaled instance such that for each $i\in N$, $v'_{i}(M)  = n +1$  \tcp*{\scriptsize{ Lemma \ref{lem:scale_invariance}}}
    
    \uIf{$\exists g\in M$ such that $v'_i(g) \geq 1$ for some $i\in N$} {
        $A_{i} \leftarrow \{g\}$\;
        $A_{j} \leftarrow M \setminus \{g\}$ for $j \neq i \in N$ \;
    }
    \Else {
        $A \leftarrow$ Run an EF1 algorithm on $I'$;
    }
\caption{Computing MMS$^{3}_{i}$ for $2$ agents}
\label{alg:n=3}
\end{algorithm}

%% file: algs/lone_divider_two_thirds.tex
\begin{algorithm}[t] \small
\SetKwInOut{Input}{Input}
\SetKwInOut{Output}{Output}
    \Input{An ordered, strongly normalized instance $I=\ins{N, M, V}$}
    \Output{A $\MMS{\frac{2}{3},1}$ allocation on $I$}
    Let $N'$ be a subset of $\Targ$ agents from $N$\;
    Let $M' = M$\;
    \While{$|N'| > 0$} {
        Select an agent $i$ from $N'$ to be the divider\;
        Let $h = |\{g \vert v_{i}(g) > \frac{1}{2}\}|$ and $s = \max(h - \Targ, 0)$\;
        \uIf{$s = 0$} {
            Bag-fill $n'$ bundles\;
        }
        \Else{
            Pair $\min(n', s)$ bundles\;
            \If{$s < n'$}{
                Restricted bag-fill $\min(n' - s, s)$ bundles\;
                \If{$s < n' - s$} {
                    Bag-fill $n' - 2s$ bundles\;
                }
            }
        }
        Let $G = (N' \cup D, E)$ be the bipartite graph\;
        Find a non-empty envy-free matching $\mathcal{M}$ of $G$\;
        \ForEach{$(j, D_{k}) \in \mathcal{M}$} {
            Allocate $D_{k}$ to agent $j$\;
            $N' = N' \setminus \{j\}$\; 
            $M' = M' \setminus D_{k}$\;
        }
    }
\caption{A $\MMS{\frac{2}{3}}{1}$ algorithm}
\label{alg:twothird_one}
\end{algorithm}

%% file: images/lone_divider_decision_tree.tex
\tikzset{every picture/.style={line width=0.75pt}} 

\tikzset{every picture/.style={line width=0.75pt}} 

\begin{tikzpicture}[x=0.75pt,y=0.75pt,yscale=-0.75,xscale=.75]

\draw    (50,50) -- (178,50) ;
\draw [shift={(180,50)}, rotate = 180] [color={rgb, 255:red, 0; green, 0; blue, 0 }  ][line width=0.75]    (10.93,-3.29) .. controls (6.95,-1.4) and (3.31,-0.3) .. (0,0) .. controls (3.31,0.3) and (6.95,1.4) .. (10.93,3.29)   ;
\draw    (50,50) -- (50,168) ;
\draw [shift={(50,170)}, rotate = 270] [color={rgb, 255:red, 0; green, 0; blue, 0 }  ][line width=0.75]    (10.93,-3.29) .. controls (6.95,-1.4) and (3.31,-0.3) .. (0,0) .. controls (3.31,0.3) and (6.95,1.4) .. (10.93,3.29)   ;
\draw    (640,80) -- (640,168) ;
\draw [shift={(640,170)}, rotate = 270] [color={rgb, 255:red, 0; green, 0; blue, 0 }  ][line width=0.75]    (10.93,-3.29) .. controls (6.95,-1.4) and (3.31,-0.3) .. (0,0) .. controls (3.31,0.3) and (6.95,1.4) .. (10.93,3.29)   ;
\draw    (250,50) -- (378,50) ;
\draw [shift={(380,50)}, rotate = 180] [color={rgb, 255:red, 0; green, 0; blue, 0 }  ][line width=0.75]    (10.93,-3.29) .. controls (6.95,-1.4) and (3.31,-0.3) .. (0,0) .. controls (3.31,0.3) and (6.95,1.4) .. (10.93,3.29)   ;
\draw    (210,70) -- (210,168) ;
\draw [shift={(210,170)}, rotate = 270] [color={rgb, 255:red, 0; green, 0; blue, 0 }  ][line width=0.75]    (10.93,-3.29) .. controls (6.95,-1.4) and (3.31,-0.3) .. (0,0) .. controls (3.31,0.3) and (6.95,1.4) .. (10.93,3.29)   ;
\draw    (440,50) -- (568,50) ;
\draw [shift={(570,50)}, rotate = 180] [color={rgb, 255:red, 0; green, 0; blue, 0 }  ][line width=0.75]    (10.93,-3.29) .. controls (6.95,-1.4) and (3.31,-0.3) .. (0,0) .. controls (3.31,0.3) and (6.95,1.4) .. (10.93,3.29)   ;
\draw    (410,80) -- (410,168) ;
\draw [shift={(410,170)}, rotate = 270] [color={rgb, 255:red, 0; green, 0; blue, 0 }  ][line width=0.75]    (10.93,-3.29) .. controls (6.95,-1.4) and (3.31,-0.3) .. (0,0) .. controls (3.31,0.3) and (6.95,1.4) .. (10.93,3.29)   ;
\draw    (70,220) -- (278.05,269.54) ;
\draw [shift={(280,270)}, rotate = 193.39] [color={rgb, 255:red, 0; green, 0; blue, 0 }  ][line width=0.75]    (10.93,-3.29) .. controls (6.95,-1.4) and (3.31,-0.3) .. (0,0) .. controls (3.31,0.3) and (6.95,1.4) .. (10.93,3.29)   ;
\draw    (210,220) -- (288.21,259.11) ;
\draw [shift={(290,260)}, rotate = 206.57] [color={rgb, 255:red, 0; green, 0; blue, 0 }  ][line width=0.75]    (10.93,-3.29) .. controls (6.95,-1.4) and (3.31,-0.3) .. (0,0) .. controls (3.31,0.3) and (6.95,1.4) .. (10.93,3.29)   ;
\draw    (410,230) -- (361.71,258.97) ;
\draw [shift={(360,260)}, rotate = 329.03999999999996] [color={rgb, 255:red, 0; green, 0; blue, 0 }  ][line width=0.75]    (10.93,-3.29) .. controls (6.95,-1.4) and (3.31,-0.3) .. (0,0) .. controls (3.31,0.3) and (6.95,1.4) .. (10.93,3.29)   ;
\draw    (610,220) -- (371.96,269.59) ;
\draw [shift={(370,270)}, rotate = 348.23] [color={rgb, 255:red, 0; green, 0; blue, 0 }  ][line width=0.75]    (10.93,-3.29) .. controls (6.95,-1.4) and (3.31,-0.3) .. (0,0) .. controls (3.31,0.3) and (6.95,1.4) .. (10.93,3.29)   ;

\draw (2,140) node [anchor=north west][inner sep=0.75pt]  [rotate=-270] [align=left] {$\displaystyle h\leq \lfloor \frac{2n}{3}\rfloor $ };
\draw (11.01,178.99) node [anchor=north west][inner sep=0.75pt]  [rotate=-0.01] [align=left] {\begin{minipage}[lt]{50.93200000000001pt}\setlength\topsep{0pt}
\begin{center}
Bag-fill\\$\displaystyle n'$ bundles
\end{center}

\end{minipage}};
\draw (180,42) node [anchor=north west][inner sep=0.75pt]   [align=left] {\begin{minipage}[lt]{34.487356000000005pt}\setlength\topsep{0pt}
\begin{center}
Pairing
\end{center}

\end{minipage}};
\draw (293,262) node [anchor=north west][inner sep=0.75pt]   [align=left] {Non-empty \\Envy-free\\Matching};
\draw (77.98,2.01) node [anchor=north west][inner sep=0.75pt]  [rotate=-0.79] [align=left] {$\displaystyle h >\lfloor \frac{2n}{3}\rfloor $ };
\draw (182,140) node [anchor=north west][inner sep=0.75pt]  [rotate=-270] [align=left] {\begin{minipage}[lt]{31.495356pt}\setlength\topsep{0pt}
\begin{center}
$\displaystyle n'\leq s$
\end{center}

\end{minipage}};
\draw (184,172) node [anchor=north west][inner sep=0.75pt]   [align=left] {\begin{minipage}[lt]{38.476644pt}\setlength\topsep{0pt}
\begin{center}
Pair $\displaystyle n'$\\bundles
\end{center}

\end{minipage}};
\draw (281,15) node [anchor=north west][inner sep=0.75pt]   [align=left] {\begin{minipage}[lt]{34.952pt}\setlength\topsep{0pt}
\begin{center}
$ $$\displaystyle n' >s$
\end{center}

\end{minipage}};
\draw (375,32) node [anchor=north west][inner sep=0.75pt]   [align=left] {\begin{minipage}[lt]{38.476644pt}\setlength\topsep{0pt}
\begin{center}
Pair $\displaystyle s$\\bundles
\end{center}

\end{minipage}};
\draw (382,160) node [anchor=north west][inner sep=0.75pt]  [rotate=-270] [align=left] {\begin{minipage}[lt]{47.6pt}\setlength\topsep{0pt}
\begin{center}
$\displaystyle n'-s\leq s$
\end{center}

\end{minipage}};
\draw (351,179) node [anchor=north west][inner sep=0.75pt]   [align=left] {\begin{minipage}[lt]{82.67644pt}\setlength\topsep{0pt}
\begin{center}
Restricted Bag-fill\\$\displaystyle n'-s$ bundles
\end{center}

\end{minipage}};
\draw (460,15) node [anchor=north west][inner sep=0.75pt]   [align=left] {\begin{minipage}[lt]{47.588644pt}\setlength\topsep{0pt}
\begin{center}
$\displaystyle n'-s >s$
\end{center}

\end{minipage}};
\draw (581,32) node [anchor=north west][inner sep=0.75pt]   [align=left] {\begin{minipage}[lt]{82.67644pt}\setlength\topsep{0pt}
\begin{center}
Restricted Bag-fill\\$\displaystyle s$ bundles
\end{center}

\end{minipage}};
\draw (591.01,178.98) node [anchor=north west][inner sep=0.75pt]  [rotate=-0.01] [align=left] {\begin{minipage}[lt]{72.80556pt}\setlength\topsep{0pt}
\begin{center}
Bag-fill\\$\displaystyle n'-2s$ bundles
\end{center}

\end{minipage}};

\end{tikzpicture}

%% file: algs/lone_divider_poly.tex
\begin{algorithm}[t] \small
\SetKwInOut{Input}{Input}
\SetKwInOut{Output}{Output}
    \Input{A normalized instance $I=\ins{N, M, V}$}
    \Output{A $\MMS{\frac{2}{3},1}$ allocation $A$ on $I$}
    Let $N'$ be a subset of $\Targ$ agents from $N$\;
    Let $M'  = M$ \;
    \While{$|N'| > 0$ and $|M'| > 0$ } {
        Select an agent $i$ from $N'$ to be the divider\;
        Let $h = |\{g \vert v_{i}(g) > \frac{1}{2}\}|$ and $s = \max(h - \Targ, 0)$\;
        \uIf{$s = 0$} {
            Bag-fill $n'$ bundles\;
        }
        \Else{
            Pair $\min(n', s)$ bundles\;
            \If{$s < n'$}{
                Bag-fill $n' - s$ bundles\;
            }
        }
        Let $G = (N' \cup D, E)$ be the acceptability graph\;
        Find a non-empty envy-free matching $\mathcal{M}$ of $G$\;
        \ForEach{$(j, D_{k}) \in \mathcal{M}$} {
            Allocate $D_{k}$ to agent $j$\;
            $N' = N' \setminus \{j\}$\; 
            $M' = M' \setminus D_{k}$\;
        }
    }
\caption{A polynomial-time variant of \cref{alg:twothird_one} without restricted bag-filling.}
\label{alg:lone_divider_poly}
\end{algorithm}